\documentclass[twoside,11pt]{article}

\usepackage{jmlr2e_arxiv}

\usepackage{xcolor,colortbl}
\usepackage{amscd,amsmath,amstext,amsfonts,amsbsy,amssymb,mathtools}
\usepackage{graphicx,framed,enumerate}
\usepackage{caption}
\usepackage{subcaption}
\usepackage{array,multirow,multicol,longtable,booktabs}
\usepackage{threeparttable,rotating,fancyhdr, paralist}
\usepackage{algorithm}
\usepackage{algorithmic}
\usepackage[english]{babel}
\usepackage{verbatim,float}
\usepackage{lscape}
\usepackage{booktabs}
\newtheorem{dn}{Definition}[section]

\newtheorem{cy}[dn]{Note}

\newtheorem{dl}[dn]{Theorem}
\newtheorem{md}[dn]{Proposition}
\newtheorem{bd}[dn]{Lemma}

\newtheorem{assump}[dn]{Assumption}

\newcolumntype{H}{>{\setbox0=\hbox\bgroup}c<{\egroup}@{}}

\makeatletter
\newcommand*{\indep}{%
	\mathbin{%
		\mathpalette{\@indep}{}%
	}%
}
\newcommand*{\nindep}{%
	\mathbin{
		\mathpalette{\@indep}{\not}
	}%
}
\newcommand*{\@indep}[2]{%
	\sbox0{$#1\perp\m@th$}
	\sbox2{$#1=$}
	\sbox4{$#1\vcenter{}$}
	\rlap{\copy0}
	\dimen@=\dimexpr\ht2-\ht4-.2pt\relax
	\kern\dimen@
	{#2}%
	\kern\dimen@
	\copy0 
} 
\makeatother

\ShortHeadings{Structure learning of undirected graphical models}{Nguyen and Chiogna}
\firstpageno{1}

\begin{document}
	
	\title{Structure learning of undirected graphical models for count data}
	
	\author{\name Nguyen Thi Kim Hue \email nguyen@stat.unipd.it \\
		\addr Department of Statistical Sciences\\
		University of Padova\\
		Via C. Battisti, 241 - 35121 Padova, Italy
		\AND
		\name Monica Chiogna \email monica.chiogna2@unibo.it \\
		\addr Department of Statistical Sciences ``Paolo Fortunati"\\
		University of Bologna\\
		Via Belle Arti 41 - 40126 Bologna, Italy}
	
	\editor{~}
	
	\maketitle
	
	\begin{abstract}
		{\color{black} Mainly motivated by the problem of modelling biological processes underlying the basic functions of a cell -that typically involve complex  interactions between genes- we present a new algorithm, called PC-LPGM,  for learning  the structure of    undirected graphical models over discrete variables. 
			We prove theoretical consistency of PC-LPGM in the limit of infinite observations and discuss its robustness to model misspecification. 
			To evaluate the performance of PC-LPGM in  recovering the true structure of the graphs in situations where relatively moderate  sample sizes are available,  extensive simulation studies are conducted, that also allow  to compare our proposal with its main competitors. A biological validation of the algorithm is presented through the analysis of two real data sets.}
	\end{abstract}
	
	\begin{keywords}
		Graphical models, Undirected graphs, Structure learning, Sparsity, Conditional independence tests.
	\end{keywords}
	
	\section{Introduction}
	
	Current demand for modelling complex interactions between genes, combined with the greater availability of high-dimensional discrete data, possibly showing a large number of zeros and measured on a small number of units, has led to an increased focus on structure learning for discrete data in high dimensional settings.
	
	Various solutions are nowadays available in the literature for learning (sparse) graphical models for discrete data. \citet{hofling2009estimation} consider the problem of estimating the parameters as well as the structure of binary-valued Markov networks; \citet{ravikumar2010} consider the problem of estimating the graph associated with a binary Ising Markov random field; \citet{jalali2011learning} consider learning general discrete graphical models, where each variable can take a multiplicity of possible values, and factors can be of order higher than two and \cite{allen2013local} consider learning graphical models for Poisson counts. To deal with high dimensionality,  most methods resort on penalization, which simultaneously performs parameter estimation and model selection.

	In this paper, we concentrate on count data and introduce a simple algorithm for structure learning of undirected graphical models, called PC-LPGM, particularly useful when sparse graphs are under consideration. The algorithm stems from the conditional approach of   {\cite{allen2013local}}, where the neighbourhood of each node is estimated in turn by solving a lasso penalized regression problem and the resulting local structures stitched together to form the global graph.  We propose to substitute penalized estimation with a testing procedure on the parameters of the local regressions following the lines of the  PC algorithm, see \cite{spirtes2000causation}. This solution is particularly attractive, since it inherits the potential of the PC algorithm  to estimate a sparse graph even if $p,$ the number of variables, is in the hundreds or thousands.  
	
	We give a theoretical proof of convergence of PC-LPGM that shows  the proposed algorithm consistently estimates the edges of the underlying (sparse) undirected graph, as the sample size $n\rightarrow \infty.$ For such proof to be developed, a joint distribution must exist, a condition which might be questionable when relying on a conditional model specification such as the one behind a neighbourhood approach. If one assumes that each variable conditioned on all other variables follows a Poisson distribution, for example,  a unique  joint distribution compatible with the given conditionals exists provided that conditional dependencies are all negative. As this condition,  known as ``competitive relationship" among variables,  highly limits attractiveness of such specification in applications, we have chosen to develop statistical guarantees for PC-LPGM under the assumption that conditional distributions follow a truncated Poisson law. Such choice  admits  dependencies richer than those under competitive relationship; see, however, \cite{yang2013poisson} for a discussion about its limitations. For the truncated Poisson model, under mild assumptions on the expected Fisher information matrix,  and fixing the truncation point $R>0$, convergence is guaranteed for $n\ge O_p(d^3\log p),$  where $d$ is the maximum neighbourhood size. 
	
	To explore whether  it is reasonable to extend the desirable properties of PC-LPGM to the case of conditional Poisson distributions with unrestricted conditional dependencies, a theoretical study and  extensive simulations studies are conducted to  evaluate statistical properties of the algorithm in such cases.
	
	The paper is organized as follows. After reviewing some essential concepts on undirected graphical models and truncated Poisson models in Section \ref{reviewPois}, we introduce  PC-LPGM algorithm   in  Section  \ref{proposedmeth}.  We  then  provide statistical guarantees in Section \ref{statistical}.  A discussion on { consistency of the algorithm is offered in Section \ref{robustness}, with special focus on the model specification, the truncated Poisson distribution, and on properties of the algorithm in the setting of conditional Poisson distributions with unrestricted conditional dependencies. This section also explores performances of the algorithm relative to various alternatives. An excursion into   the intuitive advantages of the learning strategy adopted by PC-LPGM is taken in Section \ref{discussion}.
		A validation of the algorithm on two real cases is given in Section \ref{realanalysis}.  Some  concluding remarks are presented in Section \ref{conclusions}. }
	
	\section{A quick review on truncated Poisson undirected graphical models}\label{reviewPois}
	
	In this section, we review some essential concepts on undirected graphical models and introduce  truncated Poisson undirected graphical models. 
	
	Consider a $p$-dimensional random vector $\bold{X}=(X_1,\ldots,X_p)$ such that each random variable $X_s$ corresponds to a node of a graph $G=(V,E)$ with index set $V=\{1,2,\ldots,p\}$.  An edge between two nodes $s$ and $t$ will be denoted by $(s,t)$. The neighbourhood of a node $s\in V$ is defined to be the set $N(s)=\{t\in V:~ (s,t)\in E\}$ consisting of all nodes connected to $s$.
	The random vector $\bold{X}$ satisfies the pairwise Markov property with respect to $G$ if 
	$$X_s\indep X_t|\mathbf{X}_{V\backslash \{s,t\}}$$
	whenever $(s,t) \notin E.$ When all variables $X_s, s\in V,$ are discrete with positive joint probabilities, as in the case  under consideration, the pairwise Markov property coincides with the local and global Markov property, according to which, respectively, 
	$$X_s\indep \bold{X}_{V\backslash\{N(s)\cup \{s\}\}}|\bold{x}_{N(s)}$$
	for every $s\in V,$ and
	$$\bold{X}_A\indep \bold{X}_B|\bold{x}_C,$$
	for any triple of pairwise disjoint subsets $A,B,C\subset V$ such that $C$ separates $A$ and $B$ in $G$, that is, every path between a node in $A$ and a node in $B$ contains a node in $C.$

	To specify a probabilistic model for $\mathbf{X}$, we  take a conditional approach (see also \cite{arnold2012conditionally}). 	Assume that each conditional distribution of node $X_s$ given other variables $\bold{X}_{V\backslash\{s\}}$ follows a  Poisson distribution truncated at $R$, $R>0,$  {written as} $X_s|\bold{x}_{V\backslash\{s\}}\sim \text{TP}(\exp\{\theta_s+\sum_{t\ne s}\theta_{st}x_t\}),$ 
	with node conditional distribution
	\begin{eqnarray}\label{dijoinprob}
	\mathbb{P}(x_s|\bold{x}_{V\backslash \{s\}})&=&\dfrac{\exp\big\{\theta_sx_s+\sum_{t\ne s}\theta_{st}x_tx_s-\log x_s!\big\}}{\sum_{k=0}^R\exp\big\{\theta_sk+k\sum_{t\ne s}\theta_{st}x_t-\log k!\big\}}\nonumber\\
	&=&\dfrac{\exp\big\{\theta_sx_s+x_s\langle \boldsymbol{\theta}_s,\bold{x}_{V\backslash\{s\}}\rangle-\log x_s!\big\}}{\sum_{k=0}^R\exp\big\{\theta_sk+k\langle \boldsymbol{\theta}_s,\bold{x}_{V\backslash\{s\}}\rangle-\log k!\big\}}\nonumber\\
	&=&\exp\big\{\theta_sx_s+x_s\langle \boldsymbol{\theta}_s,\bold{x}_{V\backslash\{s\}}\rangle-\log x_s!-D(\langle \boldsymbol{\theta}_s,\bold{x}_{V\backslash\{s\}}\rangle)\big\},
	\end{eqnarray}
	where  $\boldsymbol{\theta}_s=\{\theta_{st},~ t\in V, t\ne s\}$ denotes the set of conditional dependence parameters, $\langle .,.\rangle$  denotes the inner product, and $D(\langle \boldsymbol{\theta}_s,\bold{x}_{V\backslash\{s\}}\rangle)=\log\big(\sum_{k=0}^R\exp\big\{\theta_sk+k\langle \boldsymbol{\theta}_s,\bold{x}_{V\backslash\{s\}}\rangle-\log k!\big\}\big)$.
	
	An application of Proposition 1 in \cite{yang2015graphical} shows that a valid joint probability distribution function from the above given set of specified conditional distributions can be constructed. By Assumption 1 and Assumption 2 in Section 4.1 of \cite{besag1974spatial},	such distribution defines an undirected graph $G=(V,E)$ in which a missing edge between node $s$ and node $t$ corresponds to the condition $\theta_{st}=\theta_{ts}= 0.$ On the other side, one edge between node $s$ and node $t$ implies $\theta_{st}\equiv\theta_{ts}\ne 0.$  
	
	The existence of a joint distribution suggests that  the structure of the network might be recovered from observed data within a likelihood approach by mean of a set of statistical tests. Indeed, in an undirected graphical model, the pairwise Markov property infers a collection of full conditional independences encoded in absent edges. For this reason, performing $\binom{|V|}{2}$ pairwise full conditional independence tests yields a method to estimate the graph $G$. 
	However, such an approach might be impractical even for modestly sized graphs.  The existence of the maximum likelihood estimates is, in general, not guaranteed  if the number of observations is small, the basic problem being that the number of parameters in {$\boldsymbol{\theta}=\{\boldsymbol{\theta}_s, \,\, s\in V\}$} is of the order $p^2$.  Hence, the sample size is often not large enough to obtain a good estimator. Moreover, it requires computing complex normalization constants and  combinatorial searches through the space of graph structures.  For this reason, in what follows, we will exploit the local Markov property, according to which every variable is conditionally independent of the remaining ones given its neighbours. This property suggests that each variable $X_s, s\in V$ can be optimally predicted from its neighbour $\bold{X}_{N(s)}.$

	\section{The PC-LPGM algorithm}\label{proposedmeth}
	
	\noindent
	We will work within the neighbourhood selection approach. The analysis of this setting is related to the concept of pseudo-likelihood,
	$$PL(\boldsymbol{\theta}) = \prod_{s \in V} \mathbb{P}(x_s|\bold{x}_{V\backslash \{s\}}),$$
	where $\mathbb{P}(x_s|\bold{x}_{V\backslash \{s\}})$ is the distribution function of each node conditional distribution. Standard model specifications treat different conditional distributions $\mathbb{P}(x_s|\bold{x}_{V\backslash \{s\}})$  as unrelated.  {In other words, the} symmetry of interaction parameters $\theta_{st}$ and $\theta_{ts}$ is usually not explicitly taken  into account (see, however, \cite{peng2009partial} for a solution that takes the natural symmetry of coefficients into account in the Gaussian setting).

	In this setting, structure learning usually proceeds by disjointly maximizing the single factors in $PL(\boldsymbol{\theta}).$ In high-dimensional sparse settings, many up-to-date algorithms are based on solving local convex optimization problems, typically formed by the sum of a loss function, such as the local negative log likelihood, with a sparsity inducing penalty function. Each local penalized estimate ${\boldsymbol{\hat\theta}_s}$ is then combined into a single non-degenerate global estimate, possibly employing consensus operators aimed at solving inconsistencies with respect to parameters shared between factors \citep[see, for example,][]{mizrahi2014distributed}. 
	From empirical studies, it is in most cases easy to check that such algorithms converge, sometimes also reasonably quickly thanks to the possibility of distributing the various maximization tasks. However, it is not immediately clear if convergence can be established theoretically, so that it cannot be given for granted that such algorithms ultimately yield correct graphs.
	
	Our proposal, called PC-LPGM,  is a pseudo-likelihood based algorithm that stems from  current neighbourhood selection methods for count data  \citep[see][]{allen2013local}, but substitutes penalization with hypothesis testing.  { In Section 5, we will pause on the relative merits of hypothesis testing with respect to penalization. But prodromic to such discussion is the proof, developed in Section 4, that  the  sequence  of  tests  does  indeed  converge  to  the true structure in the limit of infinite observations, regardless of the dimension of the problem.}
	
	We consider the same model specification as in (\ref{dijoinprob}). In detail, we assume that each node conditional distribution follows a truncated Poisson distribution. As we are only interested in the structure of graph $G$, without loss of generality we can assume  $\theta_s=0,~s\in V.$ 
	In line with the most common solutions, we also treat the conditional distributions $\mathbb{P}_{\boldsymbol{\theta_s}}(x_s|\bold{x}_{V\backslash \{s\}})$ as unrelated.

	In PC-LPGM, neighbours are identified by mean of conditional independence tests built from the conditional models and aimed at  identifying the set of non-zero conditional dependence parameters.  Tests are based on Wald-type  statistics built on exploiting the asymptotic normality of the local maximum likelihood estimators.  
	To face the high computational complexity related to the testing procedure, we employ the PC algorithm, which relies on controlling the number of variables in the conditional sets, a strategy particularly effective when sparse graphs are under consideration.
	
	In what follows, let $\bold{X}^{(1)},\ldots,\bold{X}^{(n)}$ be $n$ independent p-random vectors drawn from $\bold{X}$, where $\bold{X}^{(i)}= (X_{i1},\ldots,X_{ip})$; and $\mathbb{X}=\{\bold{x}^{(1)},\ldots,\bold{x}^{(n)}\}$ be the collection of $n$ samples drawn from the random vectors $\bold{X}^{(1)},\ldots,\bold{X}^{(n)}$, with $\bold{x}^{(i)}=(x_{i1},\ldots,x_{ip}),\quad i=1,\ldots,n$. For each $U\subset V$, let $\mathbb{X}_U$ be the set of $n$ samples of the $|U|$-random vector $\bold{X}_U=(X_i:~i\in U)$, with $\bold{x}_U^{(i)}=(x_{ij})_{j\in U},~ i=1,\ldots,n$. 
	Starting from the complete graph, for each $s$ and $t\in V\backslash\{s\}$ and  for any set of variables $\bold{S}\subset \{1,\ldots,p\}\backslash \{s,t\}$, we test,  at some pre-specified significance level, the null hypothesis $H_0: {\theta}_{st|\bold{K}}=0$, with $\bold{K}=\bold{S}\cup\{s,t\}$. In other words, we test if data support existence of the conditional independence relation $X_s\indep X_t|\bold{X}_{\bold{S}}$. If the null hypothesis is not rejected, the edge $(s,t)$ is considered to be absent from the graph. A control is operated on the cardinality of the set $\bold{S}$ of conditioning variables, which is progressively increased from 0 to $p-2$ or to $m, \,\, m<(p-2)$.

	Assume 
	\begin{equation}\label{linktoOr-PPGM}
	X_s| {\bold{x}_{\bold{K}\backslash\{s\}}}\sim \text{TP}\big(\exp\big\{\sum_{t\in \bold{K}\backslash\{s\}}\theta_{st|\bold{K}}x_t\big\}\big),\quad \forall s\in \bold{K}\subset \{1,\ldots,p\},
	\end{equation}
	and denote $\boldsymbol{\theta}_{s|\bold{K}}= \{\theta_{st|\bold{K}}:~ t\in \bold{K}\backslash\{s\}\}.$ A rescaled negative node conditional log-likelihood given the conditioning variables $\bold{X}_{\bold{K}\backslash\{s\}}=(X_k:~ k\in \bold{K}\backslash\{s\})$ can be written as 	
	
	\begin{eqnarray}\label{smallloglikelihood}
	l(\boldsymbol{\theta}_{s|\bold{K}}, \, \mathbb{X}_{\{s\}} \, ;  {\mathbb{X}_{\bold{K}\backslash\{s\}}}) &=& -\frac{1}{n}\log \prod_{i=1}^{n}\mathbb{P}_{\boldsymbol{\theta}_{s|\bold{K}}}(x_{is}| {\bold{x}^{(i)}_{\bold{K}\backslash\{s\}}})\\
	&=&\frac{1}{n}\sum_{i=1}^{n}\left[-x_{is}\langle\boldsymbol{\theta}_{s|\bold{K}}, {\bold{x}^{(i)}_{\bold{K}\backslash\{s\}}}\rangle+\log x_{is}! +D(\langle\boldsymbol{\theta}_{s|\bold{K}}, {\bold{x}^{(i)}_{\bold{K}\backslash\{s\}}}\rangle)\right],\nonumber
	\end{eqnarray} 
	where  the scaling factor is taken for later mathematical convenience.
	The estimate $\hat{\boldsymbol{\theta}}_{s|\bold{K}}$ of the parameter $\boldsymbol{\theta}_{s|\bold{K}}$  {is determined  by minimizing the  rescaled negative conditional log-likelihood given in Equation \eqref{smallloglikelihood}}, i.e.,
	\begin{equation*}
	\hat{\boldsymbol{\theta}}_{s|\bold{K}} = \text{argmin}_{\boldsymbol{\theta}_{s|\bold{K} }\in \mathbb{R}^{|\bold{K}|-1} }	\,\,l(\boldsymbol{\theta}_{s|\bold{K}}, \, \mathbb{X}_{\{s\}} \, ;  {\mathbb{X}_{\bold{K}\backslash\{s\}}}).
	\end{equation*}
	
	A Wald-type test statistic for the hypothesis $H_0: {\theta}_{st|\bold{K}}=0$  can be obtained from asymptotic normality of $\hat{\boldsymbol{\theta}}_{s|\bold{K}}$, 
	$$\sqrt{n} (\hat{\boldsymbol{\theta}}_{s|\bold{K}} - {\boldsymbol{\theta}}_{s|\bold{K}}) \xrightarrow{\text{ d }} N(\boldsymbol{0}, I({\boldsymbol{\theta}}_{s|\bold{K}})^{-1}), $$
	where $I({\boldsymbol{\theta}}_{s|\bold{K}})$ denotes the expected Fisher information matrix,
	$$I({\boldsymbol{\theta}}_{s|\bold{K}}) = {\mathbb{E}_{\boldsymbol{\theta}_s}}
	\left[n\dfrac{\partial ^2l(\boldsymbol{\theta}_{s|\bold{K}}, X_s \, ;  {\mathbb{X}_{\bold{K}\backslash\{s\}}})}
	{\partial^2\boldsymbol{\theta}_{s|\bold{K}}}\right],$$
	which holds under fairly general regularity conditions. The test statistic for the null hypothesis  $H_0: {\theta}_{st|\bold{K}}=0$ can be obtained on exploiting the marginal asymptotic normality of the component  ${\hat\theta}_{st|\bold{K}}.$

	In practice, the observed information $J({\boldsymbol{\theta}}_{s|\bold{K}})= n \dfrac{\partial ^2l(\boldsymbol{\theta}_{s|\bold{K}}, \mathbb{X}_{\{s\}} ;  {\mathbb{X}_{\bold{K}\backslash\{s\}}})}{\partial^2\boldsymbol{\theta}_{s|\bold{K}}}$, that is,  the second derivative of the negative log-likelihood function,   is more conveniently used evaluated at $\hat{\boldsymbol{\theta}}_{s|\bold{K}}$ as variance estimate  of maximum likelihood quantities instead of the expected Fisher information matrix, a modification which comes from the use of an appropriately conditioned sampling distribution for the maximum likelihood estimators. Following this line, the test statistic for the hypothesis $H_0: {\theta}_{st|\bold{K}}=0$ is  given by 
	\vspace{-.2em}
	\begin{equation}\label{statisticz}
	Z_{st|\bold{K}}=\displaystyle\dfrac{\sqrt{n}\hat{\theta}_{st|\bold{K}}}{\sqrt{ \left[ J(\hat{\boldsymbol{\theta}}_{s|\bold{K}})^{-1}\right]_{tt}}},
	\vspace{-.18em} 
	\end{equation} 
	where $\left[A\right]_{jj}$ denotes the element in position $(j,j)$ of matrix $A.$ It is readily available that $Z_{st|\bold{K}}$ is asymptotically standard normally distributed under the null hypothesis, provided that some general regularity conditions hold \citep[page 185]{lehmann1986testing}. Possible inconsistencies with respect to parameters shared between local conditional models are solved by removing edge $(s,t)$ if either $H_0: {\theta}_{st|\bold{K}} = 0$ or $H_0: {\theta}_{ts|\bold{K}} = 0$ is not rejected.

	The  conditional independence tests  are  prone  to mistakes.  Moreover, incorrectly  deleting  or  retaining  an  edge  would  result  in   changes  in  the  neighbour  sets  of  other  nodes,  as  the graph is updated dynamically. Therefore, the resulting graph is dependent on the order in which  the conditional  independence  tests are performed.  To avoid this problem, we employ the solution in \cite{colombo2014order}, who developed  a modification of the PC algorithm that removes the order-dependence, called PC-stable. In this modification, the  neighbours   of  all  nodes  are  searched for and kept  unchanged  at  each  particular  cardinality  $l$ of the set $\bold{K}_s$.  As  a  result,  an edge deletion at one level does not affect the conditioning sets of the other nodes, and thus the output is independent on the variable ordering.  
	
	The pseudo-code of our algorithm is illustrated in Algorithm \ref{pseudocode}, where $\text{adj}(\hat{G},s)=\{t\in V:~(s,t)\in \hat{G}\}$ denotes the estimated set of all nodes that are adjacent to $s$ on the graph $\hat{G}$.
	We note that the pseudo-code is identical to Algorithm 4.1 in \cite{colombo2014order}.  Indeed, the difference lies in the statistical procedure used to test the hypothesis at line 15.

	\section{Statistical Guarantees}\label{statistical}
	
	In this section, we address  the property of statistical consistency of our algorithm.  {In detail, we study}  the limiting behaviour of our estimation procedure as the sample size $n$, and the model size $p$ go to infinity. 
	In what follows, we derive uniform consistency of  our distributed estimators explicitly as a function of the sample size, $n$, the number of nodes, $p$, the truncation point $R.$  Moreover, we  prove consistency of the graph estimator as a function of the previous quantities and of the maximum neighbourhood size $d,$ by assuming that the true distribution is faithful to the graph.  We acknowledge that our results are based on the work of \cite{yang2012graphical}  for exponential family models, combined with ideas coming from \cite{kalisch2007estimating}.  In detail, we borrowed some ideas from the proof of consistency of estimators in $l_1$ regularized local models given in \cite{yang2012graphical} and we adapted to our setting the ideas of \cite{kalisch2007estimating} for proving consistency of the graph estimator.

	For the readers' convenience, before stating the main result, we summarize  some notation that will be used through out this proof. Given a vector $v\in \mathbb{R}^p$, and a parameter $q\in[0,\infty]$, we write $\|u\|_q$ to denote the usual $l_q$ norm. Given a matrix $A\in \mathbb{R}^{p\times p}$, denote the largest and smallest eigenvalues as $\Lambda_{\max}(A)$, $\Lambda_{\min}(A)$, respectively. We use $|||A|||_2= \sqrt{\Lambda_{\max}(A^TA)}$ to denote the spectral norm, corresponding to the largest singular value of $A$,
	and the $l_\infty$ matrix norm is defined as
	$|||A|||_\infty=\max_{i=1,\ldots,a}\sum_{j=1}^{a}|A_{i,j}|.$
	
	\begin{algorithm}\label{algorithmm}
		\caption{ \label{pseudocode}The PC-LPGM algorithm. }
		\begin{algorithmic}[1]
			\hrule
			\STATE{\textbf{Input}:} $n$ independent realizations of the $p$-random vector $\bold{X}$;  $\bold{x}^{(1)},\bold{x}^{(2)},\ldots,\bold{x}^{(n)}$;  \\ an ordering $order(V)$ on the variables, (and a stopping level $m$).
			\STATE{\textbf{Output}:} An estimated undirected graph $\hat{G}$.
			\STATE{} Form the complete undirected graph $\tilde{G}$ on the vertex set $V$.
			\STATE{} $l=-1$;  $\quad \hat{G}=\tilde{G}$
			\STATE{} \textbf{repeat}
			\STATE{} \quad $l=l+1$
			\STATE{} \quad \textbf{for} all vertices $s\in V$, \textbf{do}
			\STATE{} \quad\quad let $\bold{K}_s = \text{adj}(\hat{G},s)$
			\STATE{} \quad \textbf{end for}
			\STATE{} \quad \textbf{repeat}
			\STATE{} \quad\quad
			Select a (new) ordered pair of nodes $s,t$ that are adjacent in $\hat{G}$ such that 
			\STATE{} \quad\quad$|\bold{K}_s\backslash\{t\}|\ge l$, using $order(V)$.
			\STATE{} \quad\quad\textbf{repeat}
			\STATE{} \quad\quad\quad choose a (new) set $\bold{S}\subset \bold{K}_s\backslash\{t\}$ with $|\bold{S}|=l$, using $order(V)$.
			\STATE{} \quad\quad\quad\textbf{if} $H_0: {\theta}_{st|\bold{S}}=0$ not rejected
			\STATE{} \quad\quad\quad\quad delete edge $(s,t)$ from $\hat{G}$
			\STATE{} \quad\quad\quad \textbf{end if}
			\STATE{} \quad\quad\textbf{until} edge $(s,t)$ is deleted or all $\bold{S}\subset \bold{K}_s\backslash\{t\}$ with $|\bold{S}|=l$ have been considered.
			\STATE{} \quad\textbf{until} all ordered pair of adjacent variables $s$ and $t$ such that $|\bold{K}_s\backslash\{t\}|\ge l$ and 
			\STATE{} \quad\quad $\bold{S}\subset \bold{K}_s\backslash\{t\}$ with $|\bold{S}|=l$ have been tested for conditional independence.
			\STATE{} \textbf{until} $l=m$ or for each ordered pair of adjacent nodes $s,t$: $|\text{adj}(\hat{G},s)\backslash\{t\}|< l$.
		\end{algorithmic}
		\hrule
	\end{algorithm}	
{\begin{remark}
	Let $m^*$ be the maximum value reached by $l$ in Algorithm \ref{pseudocode}. When the maximum number of neighbours that one node is allowed to have is fixed to $d$,  then $m^*\in \{d-1,d\}$ (see \cite{kalisch2007estimating}). Moreover, as the stopping level $m$ in Algorithm \ref{pseudocode} satisfies $m\le m^*$, it holds $m\le d.$
\end{remark} }
	\subsection{Assumptions}
	We will  begin by stating the assumptions that underlie our analysis, and then give a precise statement of the main result. 
	
	Denote the population Fisher information and the sample Fisher information matrix corresponding to the  covariates in model {\color{ black} \eqref{linktoOr-PPGM} with $\bold{K}= V$ } as follows
	$$I_s(\boldsymbol{\theta}_s)=- \mathbb{E}_{\boldsymbol{\theta}}\left(\nabla^2 \log\left(\mathbb{P}_{\boldsymbol{\theta}_s}(X_s|\bold{X}_{V\backslash\{s\}})\right)\right),$$
	and
	$$Q_s(\boldsymbol{\theta}_s)= \nabla^2l(\boldsymbol{\theta}_s,\bold{X}_s;\bold{X}_{V\backslash \{s\}}).$$
	We note that we will consider the problem of maximum likelihood on a closed and bounded dish $\boldsymbol{\Theta}\subset \mathbb{R}^{(p-1)}$.  For $\boldsymbol{\theta}_{s|\bold{K}}\in \mathbb{R}^{|K|-1}$, we can immerse $\boldsymbol{\theta}_{s|\bold{K}}$ into $\boldsymbol{\Theta}\subset \mathbb{R}^{(p-1)}$ by zero-pad $\boldsymbol{\theta}_{s|\bold{K}}$ to include zero weights over $\{V\backslash\bold{K}\}$. 
	
	\begin{assump}\label{assum1}
		The coefficients $\boldsymbol{\theta}_{s|\bold{K}}\in \boldsymbol{\Theta}$ for  all sets $\bold{K}\subset V$ and all $s\in K$ have an upper bound norm, 
		$\max_{s,t,\bold{K}}|\theta_{st|\bold{K}}|\le M,~ \forall~ \theta_{st|\bold{K}}\ne 0,$ and  a lower bound norm, 
		$\min_{s,t,\bold{K}}|\theta_{st|\bold{K}}|\ge c,~ \forall~ \theta_{st|\bold{K}}\ne 0,$
		where $t\in \bold{K}$.
	\end{assump}
	
	\begin{assump}\label{assum2} The Fisher information  matrix corresponding to the  covariates in model {\color{ black} \eqref{linktoOr-PPGM} with $\bold{K}= V$} has bounded eigenvalues; that is, there exists a constant $\lambda_{\min}>0$ such that
		$$\Lambda_{\min}(I_s(\boldsymbol{\theta}_s))\ge \lambda_{\min}, ~\forall~\boldsymbol{\theta}_s\in \boldsymbol{\Theta}.$$
		Moreover, we require that
		$$\Lambda_{\max}\bigg(\mathbb{E}_{\boldsymbol{\theta}}\left( \bold{X}_{V\backslash \{s\}}^T\bold{X}_{V\backslash \{s\}}\right)\bigg)\le \lambda_{\max},\quad \forall 	s\in V,\forall~\boldsymbol{\theta}\in \boldsymbol{\Theta},$$
		where  $\lambda_{\max}$ is some constant such that $\lambda_{\max} <\infty$.
	\end{assump}
	
	The first assumption simply bounds the effects of covariates in all local models. In other words, we consider parameters $\theta_{st|\bold{K}}$ belong to a compact set bounded by $M$. Being the expected value of the rescaled negative log-likelihood twice differentiable, the lower
	bound on the eigenvalues of the Fisher information matrix in the second assumption guarantees strong convexity in all partial models. 
	Condition on the upper eigenvalue of the covariance matrix guarantees that  the relevant covariates do not become overly dependent, a requirement which is commonly adopted in these settings.
	
	\subsection{Convergence guarantees of local estimators}
	We are now ready to consider the question of whether convergence guarantees can be proved in the setting of our interest. Before proving our main
	theorem, we show some intermediate results of independent interest (see Appendix~\ref{suppA} for related proofs). 
	
	\begin{md}\label{pro11}
		Assume \ref{assum1}- \ref{assum2} and let $\bold{K}\subset V$. Then, for all $s\in \bold{K}$ and any $\delta>0$
		$$\mathbb{P}_{\boldsymbol{\theta}}(\|\nabla l(\boldsymbol{\theta}_{s|\bold{K}},\bold{X}_{\{s\}};{\bold{X}_{\bold{K}\backslash\{s\}}})\|_{\infty}\ge \delta)\le \exp\{-c_1n\delta^2+c_0\log d\},$$
		$~\forall~ \boldsymbol{\theta_{s|\bold{K}}}\in \boldsymbol{\Theta},$ when  $n\rightarrow\infty$.
\end{md}

	\begin{dl}\label{dl2}
		Assume \ref{assum1}- \ref{assum2} and let $\bold{K}\subset V$. Then there exists a non-negative decreasing sequence $\epsilon_n\rightarrow 0$, such that
		$$\mathbb{P}_{\boldsymbol{\theta}}(\|\hat{\boldsymbol{\theta}}_{s|\bold{K}}-\boldsymbol{\theta}_{s|\bold{K}}\|_2\le \epsilon_n)\ge 1-\exp\big\{ -c_1n\epsilon_n^2+c_0\log d\big\}-\exp\left\{-c_2\dfrac{n}{d^2}+c_3\log d\right\},$$
		$~\forall~ s\in \bold{K}, \boldsymbol{\theta}\in \boldsymbol{\Theta},$ when $n\rightarrow \infty$.	
	\end{dl}
	
	We now proceed to consider uniform consistency of the local estimators. Let $\hat{\boldsymbol{\theta}}=(\hat{\boldsymbol{\theta}}_1, \hat{\boldsymbol{\theta}}_2,\cdots,\hat{\boldsymbol{\theta}}_p)^\top$ be the array of rowwise  local estimators $\hat{\boldsymbol{\theta}}_s = \hat{\boldsymbol{\theta}}_{s|\bold{K}}$ with $\bold{K}=V$. 
	We can state the following theorem, which extends Theorem~\ref{dl2} without any additional conditions.
	\begin{dl}[uniform consistency]
		Assume \ref{assum1}- \ref{assum2}. Then, $\hat{\boldsymbol{\theta}}$ converges in probability to $\boldsymbol{\theta}$, the true value, as $n$ increases,
		uniformly in $\boldsymbol{\theta}$.
	\end{dl}
	\begin{proof}
		We have to show that given $\epsilon > 0,~ \mu > 0$, there exists an integer $n_0$ dependent on $\epsilon$ and $\mu$ but not on $\boldsymbol{\theta}$, such that for all $n > n_0$,
		$$\mathbb{P}_{\boldsymbol{\theta}}(|||\boldsymbol{\theta}-\hat{\boldsymbol{\theta}}|||_2\le \epsilon)\ge 1-\mu.$$ 
		Take $\dfrac{\epsilon}{p}$ as the number $\delta_n$, and the $\dfrac{\mu}{p}$ to be the $\exp\{-cn\}$ in Theorem \ref{dl1}. Then, for each $s\in V$, there exist $n_s$, such that for all $n>n_s$
		$$\mathbb{P}_{\boldsymbol{\theta}}\left(\|\hat{\boldsymbol{\theta}}_s-\boldsymbol{\theta}_s\|_2\le \dfrac{\epsilon}{p}\right)\ge 1-\dfrac{\mu}{p}.$$
		Let $\Omega_s$ be the space such that for all $\mathbb{X}\in \Omega_s$,
		$$\|\hat{\boldsymbol{\theta}}_s-\boldsymbol{\theta}_s\|_2\le \dfrac{\epsilon}{p},$$
		and $\mathbb{P}_{\boldsymbol{\theta}}(\Omega_s)\ge 1-\dfrac{\mu}{p}.$
		Define $n_0=\max_{s\in V}\{n_s\}$ and $\Omega=\cap_{s\in V}\Omega_s$. Then, for all $\mathbb{X}\in \Omega$, 
		
		\begin{eqnarray*}
			|||\boldsymbol{\theta}-\hat{\boldsymbol{\theta}}|||_2&\le&\sqrt{\sum_{s=1}^p\sum_{t\ne s}|\theta_{st}-\hat{\theta}_{st}|^2}\\
			&=&\sqrt{\sum_{s=1}^p\|\boldsymbol{\theta_{s}}-\boldsymbol{\hat{\theta}_{s}}\|_2^2}\\
			&\le&\sqrt{p\left(\dfrac{\epsilon}{p}\right)^2}\\
			&\le&\epsilon.
		\end{eqnarray*}
		Moreover, it is easy to prove by induction that
		\begin{equation*}\label{Lmeasure}
		\mathbb{P}_{\boldsymbol{\theta}}(\Omega)\ge 1-p ~\dfrac{\mu}{p}= 1-\mu.
		\end{equation*}
		Hence, for all $\mathbb{X}\in \Omega$, we have
		$|||\boldsymbol{\theta}-\hat{\boldsymbol{\theta}}|||_2\le\epsilon,$
		and $\mathbb{P}_{\boldsymbol{\theta}}(\Omega)\ge 1-\mu$.  {In other words, we have}
		$$\mathbb{P}_{\boldsymbol{\theta}}(|||\boldsymbol{\theta}-\hat{\boldsymbol{\theta}}|||_2\le \epsilon)\ge 1-\mu.$$
	\end{proof}

	\noindent

	\noindent

	\begin{remark}
		With suitable modifications,  uniform consistency can be proved in the case of Poisson node conditional distributions with ``competitive relationships" between variables, that is, with only negative conditional interaction parameters. Analogously, it can be extended to other distributions for count data belonging to the exponential family, such as the Negative Binomial distribution, provided that a joint distribution compatible with the conditional specifications can be constructed. 
	\end{remark}
	
	\begin{remark}
		Convergence of the pseudo likelihood estimator $\hat{\boldsymbol{\theta}}$ might also have been proved by characterizing its asymptotic behaviour in terms of law of large numbers. Indeed, the pseudo likelihood estimator $\hat{\boldsymbol{\theta}}$ can be proved to converge to the true parameter value when some conditions on the parameter space $\boldsymbol{\theta}$ and moments of the variables $\boldsymbol{X}$ are satisfied  \citep[see, for example, Theorem 5.7 from][]{van2000asymptotic}. It is worth noting that our proof allows to highlight the relative scaling of $n$, $p$ and $R$ needed to reach convergence.
	\end{remark}

	\subsection{Consistency of the graph estimator}\label{conv}
	
	In what follows, we assume faithfulness of the truncated Poisson node conditional distributions to the graph $G$.  {We restrict} the parameter space $\boldsymbol{\Theta}$ to the subspace, $\Omega(\boldsymbol{\Theta})$ say,  on which the faithfulness condition is guaranteed. We recall that a distribution  $P_\bold{X}$ is said to be faithful to the graph $G$ if $$\bold{X}_A\indep \bold{X}_B|\bold{X}_C\Rightarrow A\indep_G B|C,$$ for all disjoint vertex sets $A,B,C.$
	It is worth noting that faithfulness of the local distributions guarantees faithfulness of the joint distributions, thanks to the equivalence between local and global Markov property.

	Now we state the main result of this work for the consistency of the graph estimate.  We note that PC-LPGM employs a modification of the PC algorithm,  { PC-stable}. However, the proof of consistency of the algorithm in \cite{kalisch2007estimating} is unchanged.
	
	\begin{dl}\label{mainresult}
		Assume \ref{assum1}- \ref{assum2}. Denote by $\hat{G}(\alpha_n)$ the estimator resulting from Algorithm 1, and by $G$ the true graph. Then, there exists a numerical sequence $\alpha_n\longrightarrow 0$, such that
		$$\mathbb{P}_{\boldsymbol{\theta}}(\hat{G}(\alpha_n)=G)=1, ~\forall~\boldsymbol{\theta}\in {\Omega(\boldsymbol{\Theta})},$$
		when $n\longrightarrow \infty$.
	\end{dl}
	
	\begin{proof}
		Let $\hat{\theta}_{st|\bold{K}}$, and $\theta^*_{st|\bold{K}}$ denote the estimated and true partial weights between $X_s$ and $X_t$ given $X_r, r\in \bold{S}$, where $\bold{S}=\bold{K}\backslash\{s,t\}\subset \{1,\ldots,p\}\backslash \{s,t\}$. Many partial weights are tested for being zero during the run of the PC-procedure. For a fixed ordered pair of nodes $s,t$, the conditioning sets are elements of 
		$$K_{st}^m=\left\{\bold{S}\subset \{1,\ldots,p\}\backslash \{s,t\}: |\bold{S}|\le m\right\}.$$
		The cardinality is bounded by 
		$$|K_{st}^m|\le C p^m,\quad\text{ for some } 0<C<\infty.$$
		Let $E_{st|\bold{K}}$ denote type I or type II errors occurring when testing $H_0:~ \theta_{st|\bold{K}}=0$. Thus
		\begin{equation}\label{error}
		E_{st|\bold{K}} =E_{st|\bold{K}}^I\cup E_{st|\bold{K}}^{II},
		\end{equation}
		in which, for $n$ large enough 
		\begin{itemize}
			\item type I error $E_{st|\bold{K}}^I$: $Z_{st|\bold{K}}> \Phi^{-1}(1-\alpha/2)$ and ${\theta^*_{st|\bold{K}}}=0$;
			\item type II error $E_{st|\bold{K}}^{II}$: $Z_{st|\bold{K}}\le \Phi^{-1}(1-\alpha/2)$ and ${\theta^*_{st|\bold{K}}}\ne 0$;
		\end{itemize}
		where $Z_{st|\bold{K}}$ was defined in \eqref{statisticz}, and $\alpha$ is a chosen significance level. 
		Consider  an arbitrary matrix $\boldsymbol{\theta}_{|\bold{K}}=\{\boldsymbol{\theta}_{s|\bold{K}}\}^T_{s\in \bold{K}}\in \Omega(\boldsymbol{\Theta})$, such that $|\theta_{st|\bold{K}}|\ge \delta$, for some $\delta>0$. Let $\boldsymbol{\theta}^0_{|\bold{K}}$ be the matrix that has the same elements as $\boldsymbol{\theta}_{|\bold{K}}$ except $\theta_{st|\bold{K}}=\theta^0_{st|\bold{K}}=0$.
		Choose $\alpha_n=2(1-\Phi(n^b))$, { where $0<b<1/2$ will be chosen later}, then
		\begin{eqnarray}\label{error1}
		\sup_{s,t,\bold{K}\in K_{ij}^m}\mathbb{P}_{\boldsymbol{\theta}^0_{|\bold{K}}}(E^I_{st|\bold{K}}) &=& \sup_{s,t,\bold{K}\in K_{st}^m}\mathbb{P}_{\boldsymbol{\theta}^0_{|\bold{K}}}\bigg(|\hat{\theta}_{st|\bold{K}}|>n^{b-1/2}\sqrt{ \left[ J(\hat{\boldsymbol{\theta}}_{s|\bold{K}})^{-1}\right]_{tt}}\bigg)\nonumber\\
		&=& \sup_{s,t,\bold{K}\in K_{st}^m}\mathbb{P}_{\boldsymbol{\theta}^0_{|\bold{K}}}\bigg(|\hat{\theta}_{st|\bold{K}}-\theta^0_{st|\bold{K}}|>n^{b-1/2}\sqrt{ \left[ J(\hat{\boldsymbol{\theta}}_{s|\bold{K}})^{-1}\right]_{tt}}\bigg)\nonumber\\
		&\le& \exp\big\{ -c_1n^{2b}+c_0\log d\big\}+\exp\left\{-c_2\dfrac{n}{d^2}+c_3\log d\right\},
		\end{eqnarray}
		using Theorem \ref{dl1} and the fact that 
		$n^{b-1/2}\sqrt{ \left[ J(\hat{\boldsymbol{\theta}}_{s|\bold{K}})^{-1}\right]_{tt}}\longrightarrow 0$
		as $n\longrightarrow \infty.$ Furthermore, with the choice of $\alpha_n$ above, and $\delta\ge 2n^{b-1/2}\sqrt{ \left[ J(\hat{\boldsymbol{\theta}}_{s|\bold{K}})^{-1}\right]_{tt}}$,
		\begin{eqnarray*}
			\sup_{s,t,\bold{K}\in K_{st}^m}\mathbb{P}_{\boldsymbol{\theta}_{|\bold{K}}}(E^{II}_{st|\bold{K}}) &=& \sup_{s,t,\bold{K}\in K_{st}^m}\mathbb{P}_{\boldsymbol{\theta}_{|\bold{K}}}\bigg(|\hat{\theta}_{st|\bold{K}}|\le n^{b-1/2}\sqrt{ \left[ J(\hat{\boldsymbol{\theta}}_{s|\bold{K}})^{-1}\right]_{tt}}\bigg)\\
			&=& \sup_{s,t,\bold{K}\in K_{st}^m}\mathbb{P}_{\boldsymbol{\theta}_{|\bold{K}}}\bigg(|\theta_{st|\bold{K}}|-|\hat{\theta}_{st|\bold{K}}|\ge |\theta_{st|\bold{K}}|-n^{b-1/2}\sqrt{ \left[ J(\hat{\boldsymbol{\theta}}_{s|\bold{K}})^{-1}\right]_{tt}}\bigg)\\
			&\le& \sup_{s,t,\bold{K}\in K_{st}^m}\mathbb{P}_{\boldsymbol{\theta}_{|\bold{K}}}\bigg(|\theta_{st|\bold{K}}-\hat{\theta}_{st|\bold{K}}|\ge |\theta_{st|\bold{K}}|-n^{b-1/2}\sqrt{ \left[ J(\hat{\boldsymbol{\theta}}_{s|\bold{K}})^{-1}\right]_{tt}}\bigg)\\
			&\le& \sup_{s,t,\bold{K}\in K_{st}^m}\mathbb{P}_{\boldsymbol{\theta}_{|\bold{K}}}\bigg(|\hat{\theta}_{st|\bold{K}}-\theta_{st|\bold{K}}|\ge n^{b-1/2}\sqrt{ \left[ J(\hat{\boldsymbol{\theta}}_{s|\bold{K}})^{-1}\right]_{tt}}\bigg),
		\end{eqnarray*}
		Finally, by Theorem \ref{dl2}, we then obtain
		\begin{equation}\label{error2}
		\sup_{s,t,\bold{K}\in K_{st}^m}\mathbb{P}_{\boldsymbol{\theta}_{|\bold{K}}}(E^{II}_{st|\bold{K}})
		\le \exp\big\{ -c_1n^{2b}+c_0\log d\big\}+\exp\left\{-c_2\dfrac{n}{d^2}+c_3\log d\right\},
		\end{equation}
		as $n\longrightarrow\infty$. 	Now, by \eqref{error}-\eqref{error2}, we get
		\begin{eqnarray}\label{totalerror}
		&&\mathbb{P}_{\boldsymbol{\theta}}(\text{ a type I or II error occurs in testing procedure})\\
		&&~\le \mathbb{P}_{\boldsymbol{\theta}_{|\bold{K}}}(\cup_{s,t,\bold{K}\in K_{st}^m}E_{st|\bold{K}})\nonumber\\
		&&~\le O_p(p^{m+2})\sup_{s,t,\bold{K}\in K_{st}^m}\mathbb{P}_{\boldsymbol{\theta}_{|\bold{K}}}(E_{st|\bold{K}})\nonumber\\
		&&~\le O_p(p^{m+2})\bigg[\exp\big\{ -c_1n^{2b}+c_0\log d\big\}-\exp\left\{-c_2\dfrac{n}{d^2}+c_3\log d\right\}\bigg]\nonumber\\
		&&~\le O_p\bigg(\exp\big\{ -c_1n^{2b}+c_0'd\log p\big\}-\exp\left\{-c_2\dfrac{n}{d^2}+c_0'd\log p\right\}\bigg)\nonumber\\
		&&~\rightarrow 0\nonumber, 
		\end{eqnarray}
		as $n\longrightarrow\infty$ provided that $n\ge O_p(d^3\log p)$ and $b$ is chosen such that $n^{2b}\ge O_p(d\log p)$. 
	\end{proof}
	
	\noindent
	{	
		\begin{remark}\label{rr}
		With the appropriately defined Wald-type  statistics, consistency of the graph estimator can be proved in the case of Poisson node conditional distributions with ``competitive relationships" between variables, that is, with only negative conditional interaction parameters. Analogously, it can be extended to other distributions for count data belonging to the exponential family, such as the Negative Binomial distribution, provided that a joint distribution compatible with the conditional specifications can be constructed. 
	\end{remark}

{
	\section{On consistency of PC-LPGM}\label{robustness}
Previously derived statistical guarantees are based on the assumption that the node-wise data generating process belongs to the truncated Poisson family of models. Such assumption guarantees  the existence of a joint distribution,  an  ingredient essential to the proof of consistency. Two  questions naturally emerge with respect to consistency of the algorithm. The first  has to do with the choice of the truncated Poisson distribution instead of other unrestricted alternatives (see \citet{NIPS2013_5153}) that also guarantee a valid joint distribution, and therefore would make the search for a formal proof plausible. The second question has to do with consistency of PC-LPGM when the joint distribution does not exists, as it happens, for example, when  conditional Poisson distributions are assumed, but no restrictions are imposed on the conditional interaction parameters. This section focusses on such two issues.

\subsection{About the choice of the truncated Poisson distribution}
A key feature of the truncated Poisson family of models, beside guaranteeing  existence of a joint distribution,  is the inclusion of the Poisson distribution as limiting case, reached  when  the truncation point $R$ grows to infinity.  In fact, other families, such as, for example, the quadratic and the sublinear Poisson ones \citep{NIPS2013_5153},  } guarantee existence of a joint distribution. But they do not provide the inclusion property, which is crucial if one wants to explore the effects of a possible model misspecification. Indeed, one could argue that the data generating process is truly Poisson, and that truncation represents an element of model misspecification. In what follows, we will explore the impact of such misspecification.

In standard settings, it is well known that, in presence of model misspecification, the maximum likelihood estimator, instead of converging to the true parameter value, converges to a Kullback-Leibler projection of the data-generating distribution onto the fitted model class. Thus, the maximum likelihood estimator still exhibits a desirable form of robustness to model misspecification. 
In this section, we  theoretically prove that, under suitable assumptions,
PC-LPGM based on the truncated Poisson assumption still converges to the true graph even if the data generating process is  Poisson. To this aim, some new notation is introduced, which is defined  on top of previously defined notation by conveniently adding a superscript $P$ for Poisson distributions and $TP$ for truncated Poisson distributions, when needed. 
	
	In what follows, the true conditional  distributions are Poisson, i.e.,
		\begin{equation}\label{Poison model}
	X_s|\boldsymbol{x}_{V\backslash\{s\}}\sim \text{Pois}(\lambda^*),~ \text{where } \lambda^*=\exp\{\sum_{t\ne s}\theta_{st}^*x_t\}.
	\end{equation}
	For a generic sample $\mathbb{X},$ the  rescaled log-likelihood under the true model is given by
	\begin{eqnarray*}
		\ell_n^{P}(\boldsymbol{\theta}_{s},\mathbb{X}_s;\mathbb{X}_{V\backslash\{s\}})&=&\frac{1}{n}\sum_{i=1}^n\log \mathbb{P}^{P}_{\boldsymbol{\theta}_s}(x_{is}|\bold{x}^{(i)}_{V\backslash\{s\}})\\
		&=& \frac{1}{n}\sum_{i=1}^n\exp\big\{\sum_{t\ne s}\theta_{st}x_{is}x_{it}-\log x_{is}!-\exp\{\sum_{t\ne s}\theta_{st}x_{it}\}\big\}.
	\end{eqnarray*}
	Let $\hat{\boldsymbol{\theta}}^{P}_{ns}$ be a maximum likelihood estimator of the rescaled log-likelihood  $\ell_n^{P}(\boldsymbol{\theta}_{s},\mathbb{X}_s;\mathbb{X}_{V\backslash\{s\}})$. Under {standard regularity} conditions, $\hat{\boldsymbol{\theta}}^{P}_{ns}$ converges to the true parameter $\boldsymbol{\theta}_s^*$.
	
	We now apply PC-LPGM assuming that the  conditional  distributions are truncated Poisson, that is 
	\begin{equation}\label{TPoisson model}
	X_s|\boldsymbol{x}_{V\backslash\{s\}}\sim \text{TPois}(\lambda),~ \text{where } \lambda=\exp\{\sum_{t\ne s}\theta_{st}x_t\}.
	\end{equation}
	In this way,  PC-LPGM is working with a misspecified model. The rescaled log-likelihood under the misspecified model  is given by
	\begin{eqnarray*}
		\ell_n^{TP}(\boldsymbol{\theta}_{s},\mathbb{X}_s;\mathbb{X}_{V\backslash\{s\}})&=&\frac{1}{n}\sum_{i=1}^n\log \frac{\mathbb{P}^{P}_{\boldsymbol{\theta}_s}(x_{is}|\bold{x}^{(i)}_{V\backslash\{s\}})}{F^P_i(R,\boldsymbol{\theta}_s)}\\
		&=& \frac{1}{n}\sum_{i=1}^n\log \mathbb{P}^{P}_{\boldsymbol{\theta}_s}(x_{is}|\bold{x}^{(i)}_{V\backslash\{s\}})-\frac{1}{n}\sum_{i=1}^n\log{F^P_{i}(R,\boldsymbol{\theta}_s)},
	\end{eqnarray*}
where $F^P_i(R,\boldsymbol{\theta}_s)=\mathbb{P}^{P}_{\boldsymbol{\theta}_s}(x_{is}\le R|\bold{x}^{(i)}_{V\backslash\{s\}})$.
	{Thanks to the inclusion property,} it is easy to see that for all $\epsilon_1>0,$ there exists a $R_0>0$ such that, for all $R>R_0$, it holds
	\begin{equation}\label{Rconvergence}
	\|\ell_n^{TP}(\boldsymbol{\theta}_{s},\mathbb{X}_s;\mathbb{X}_{V\backslash\{s\}})-\ell_n^{P}(\boldsymbol{\theta}_{s},\mathbb{X}_s;\mathbb{X}_{V\backslash\{s\}})\|_2< \epsilon_1,~ \forall\,\boldsymbol{\theta}_{s}\in\boldsymbol{\Theta} .
	\end{equation}
	
	Fix a $R$ that satisfies Equation \eqref{Rconvergence}, where $\epsilon_1$ is given in Appendix (see proof of Theorem~\ref{robust}), and
	let $\hat{\boldsymbol{\theta}}_{ns}^{TP}$ be a maximum likelihood estimator  of the rescaled log-likelihood  $\ell_n^{TP}(\boldsymbol{\theta}_{s},\mathbb{X}_s;\mathbb{X}_{V\backslash\{s\}})$. In the following theorem, we prove that $\hat{\boldsymbol{\theta}}^{TP}_{ns},$ under suitable conditions, converges to the true parameter $\boldsymbol{\theta}_s^*.$

	\begin{dl}\label{robust}
		Assume that the log-likelihood function of models \eqref{Poison model} and \eqref{TPoisson model} have { a} unique optimal solution on $\boldsymbol{\Theta}$. Then, $\hat{\boldsymbol{\theta}}_{ns}^{TP}$ converges to the true parameter $\boldsymbol{\theta}_s^*$ when $n$ tends to infinity provided that  $|\Lambda_{min}[Q_s^P({\boldsymbol{\theta}}_{s})]|=|\Lambda_{min}[\nabla^2\ell_{n}^{P}({\boldsymbol{\theta}}_{s},\mathbb{X}_s;\mathbb{X}_{V\backslash\{s\}})]|>\lambda_{min}>0,$ for all ${\boldsymbol{\theta}}_{s}\in \boldsymbol{\Theta}.$ 
	\end{dl}

\begin{remark}
		By following the same lines, it is easy to show that Theorem \ref{robust}  also holds for $\bold{K}\subset V$, i.e., under the same conditions  of Theorem \ref{robust} $\hat{\boldsymbol{\theta}}_{ns|\bold{K}}^{TP}$ converges to the true parameter $\boldsymbol{\theta}_{s|\bold{K}}^*$ when $n$ tends to infinity.
\end{remark}
The previous theorem shows that, provided that the truncation point $R$ is large enough, the maximum likelihood estimators derived from the misspecified conditional models are still node-wise consistent. Statistical properties of PC-LPGM, however, hinge on statistical properties  of the Wald-type statistic in the conditional models. The following theorem derives the conditions under which the Wald-type statistic derived under the true and the misspecified models are asymptotically equivalent under the null hypothesis, and, therefore, provide the same test results.
	\begin{dl}\label{robust1}
		Assume that the log-likelihood function of models \eqref{Poison model} and \eqref{TPoisson model} have { a} unique optimal solution on $\boldsymbol{\Theta}$. Then, the Z statistic $Z^{TP}_{st|\bold{K}}$ converges to  $Z^{P}_{st|\bold{K}}$ when $n$ tends to infinity provided that there exist positive constants $\lambda_{min}$, and $\lambda_{max}$ such that $$\lambda_{max}>\Lambda_{max}[Q^P_s({\boldsymbol{\theta}}_{s})]\ge\Lambda_{min}[Q^P_s({\boldsymbol{\theta}}_{s})]>\lambda_{min}>0,$$ 
		and $$\lambda_{max}>\Lambda_{max}[Q^{TP}_s({\boldsymbol{\theta}}_{s})]\ge\Lambda_{min}[Q^{TP}_s({\boldsymbol{\theta}}_{s})]>\lambda_{min}>0,$$ 
		for all ${\boldsymbol{\theta}}_{s}\in \boldsymbol{\Theta}.$
	\end{dl}
	\begin{proof}
		By applying the definition of Z statistic under the null hypothesis $H_0: \theta_{st}=0$, we get $$Z_{st|\bold{K}}^P=\dfrac{\sqrt{n}\hat{\theta}_{st|\bold{K}}^{P}}{\sqrt{\left[J^P(\hat{\boldsymbol{\theta}}_{s|\bold{K}}^{P})^{-1}\right]_{tt}}}=\dfrac{\hat{\theta}_{st|\bold{K}}^{P}}{\sqrt{\left[Q^P_s(\hat{\theta}_{st|\bold{K}}^{P})^{-1}\right]_{tt}}},$$ and $$Z_{st|\bold{K}}^{TP}=\dfrac{\hat{\theta}_{st|\bold{K}}^{TP}}{\sqrt{\left[Q^{TP}_s(\hat{\boldsymbol{\theta}}_{s|\bold{K}}^{TP})^{-1}\right]_{tt}}}.$$ It holds
		\begin{eqnarray*}
			\left|Z_{st|\bold{K}}^P-Z_{st|\bold{K}}^{TP}\right|&=& \left|\dfrac{\hat{\theta}_{st|\bold{K}}^{P}}{\sqrt{\left[Q^P_s(\hat{\theta}_{st|\bold{K}}^{P})^{-1}\right]_{tt}}}-\dfrac{\hat{\theta}_{st|\bold{K}}^{TP}}{\sqrt{\left[Q^{TP}_s(\hat{\boldsymbol{\theta}}_{s|\bold{K}}^{TP})^{-1}\right]_{tt}}}\right|\\
			&\le&	\left|\dfrac{\hat{\theta}_{st|\bold{K}}^{P}}{\sqrt{\left[Q^P_s(\hat{\theta}_{st|\bold{K}}^{P})^{-1}\right]_{tt}}}\right|+\left|\dfrac{\hat{\theta}_{st|\bold{K}}^{TP}}{\sqrt{\left[Q^{TP}_s(\hat{\boldsymbol{\theta}}_{s|\bold{K}}^{TP})^{-1}\right]_{tt}}}\right|\\
			&\le& \sqrt{\dfrac{\lambda_{max}}{p}}
			\left(\left|\hat{\theta}_{st|\bold{K}}^{P}\right|+\left|\hat{\theta}_{st|\bold{K}}^{TP}\right|\right)\\
			&\rightarrow& 0 \text{ when } n\rightarrow\infty,
		\end{eqnarray*}
		since $\hat{\theta}_{st|\bold{K}}^{P}$, $\hat{\theta}_{st|\bold{K}}^{TP}$ tend to 0 when $n\rightarrow\infty$; and the truncation point $R$ is large enough (where line 2 to line 3 due to the singular value decomposition of matrices).
	\end{proof}
As $Z_{st|\bold{K}}^P$  and $Z_{st|\bold{K}}^{TP}$ are asymptotically equivalent, so are results of the tests that they provide when used in PC-LPGM. In other words, when the true model for the conditional distributions is Poisson, under the previously stated conditions, PC-LPGM based on $Z_{st|\bold{K}}^{TP}$ leads to the same rejections as the correctly specified PC-LPGM, which is based on the proper test statistics $Z_{st|\bold{K}}^P.$ In practice, 
extensive simulation studies have shown that fixing a truncation point just equal to  the largest observed value guarantees that PC-LPGM under the two specifications (the true Poisson and the misspecified truncated Poisson) leads to the same results {(an excerpt of the results produced by such studies is given in Table \ref{table10-TPandP}, Appendix \ref{suppD}).}

	\subsection{ Unrestricted Poisson conditional models}\label{empirical}
	 Remark~\ref{rr} in Section~\ref{rr}  guarantees that, in the case of Poisson node conditional distributions, a proof of consistency of PC-LPGM with proper test statistic can be provided in the situation of ``competitive relationships" between variables.  It is interesting to explore if consistency also holds   with unrestricted conditional interaction parameters, a situation for which a theoretical proof is still an unsolved question. 
	
	We devote this section to an empirical study of consistency of our proposed algorithm in this setting. We clarify that, in what follows, PC-LPGM works under a correct model specification, i.e., tests are based on the proper  $Z_{st|\bold{K}}^P$ statistic. }
	We aim to measure the ability of  PC-LPGM  to recover the true structure of the graphs, also in situations where relatively moderate sample sizes are available.   As measure of ability, we  adopt two measures: PPV that stands for Positive Predictive Value and is defined as TP/(TP+FP); and Sensitivity (Se), defined as TP/(TP+FN), where TP (true positive), FP (false positive), and FN (false negative) refer to the { number of} inferred edges.

	In doing these studies, we also aim to compare  PC-LPGM to a number of popular structure learning algorithms. We therefore consider the Local Poisson Graphical Models (LPGM) approach \citep{allen2013local},  as implemented in the {\tt R} package {\tt XMRF},  and Poisson dependency networks (PDNs) \citep{hadiji2015poisson}, 
 implemented in the {\tt R} function {\tt learnPDN} (see \url{https://sfb876.tu-dortmund.de/auto?self=%24eon9ai8e80}). 
	It is worth remembering that structure learning for discrete undirected graphical models is usually performed by employing methods for continuous data after proper data transformation.  We therefore consider two representatives of  approaches based on the Gaussian assumption, {that is,}   variable selection with lasso (VSL) \citep{meinshausen2006high}, and  the graphical lasso algorithm (GLASSO) \citep{friedman2008sparse}. 
Moreover, we consider two structure learning methods dealing with the class of nonparanormal distributions,  the nonparanormal-Copula algorithm (NPN-Copula) \citep{liu2009nonparanormal}, and the nonparanormal-SKEPTIC algorithm (NPN-Skeptic) \citep{liu2012nonparanormal}.  {These last four algorithms are all available in the {\tt R} package {\tt huge}}.  

	\subsubsection{Data generation}
	For two different cardinalities,  $p=10$ and $p=100$,  we consider three graphs of different structure: (i) a scale-free graph, in which the node degree distribution follows a powerlaw; (ii) a hub graph, where each node is connected to one of the hub nodes; (iii) a random graph, where presence of edges are   {independent and identically distributed} Bernoulli random variables.
	To construct the scale-free and hub networks, we employed the \texttt{R} package \texttt{XMRF}. For the scale-free network, we assumed a power law with parameter 0.01 for the node degree distribution. For the hub network, we assumed two hub nodes for $p= 10$, and 5 hub nodes for $p=100$.  To construct the random network,  we employed the \texttt{R} package \texttt{igraph} with edge probability $0.2$ for $p=10$, and $0.02$ for $p=100$.   See Figure~\ref{graphtypes} and~\ref{graphtypes100} for a plot of the three chosen graphs for $p=10$ and $p=100$, respectively.  
	
	\begin{figure}[htbp]
		\begin{center}
			\includegraphics[width = 0.9\linewidth, height=0.32\textheight]{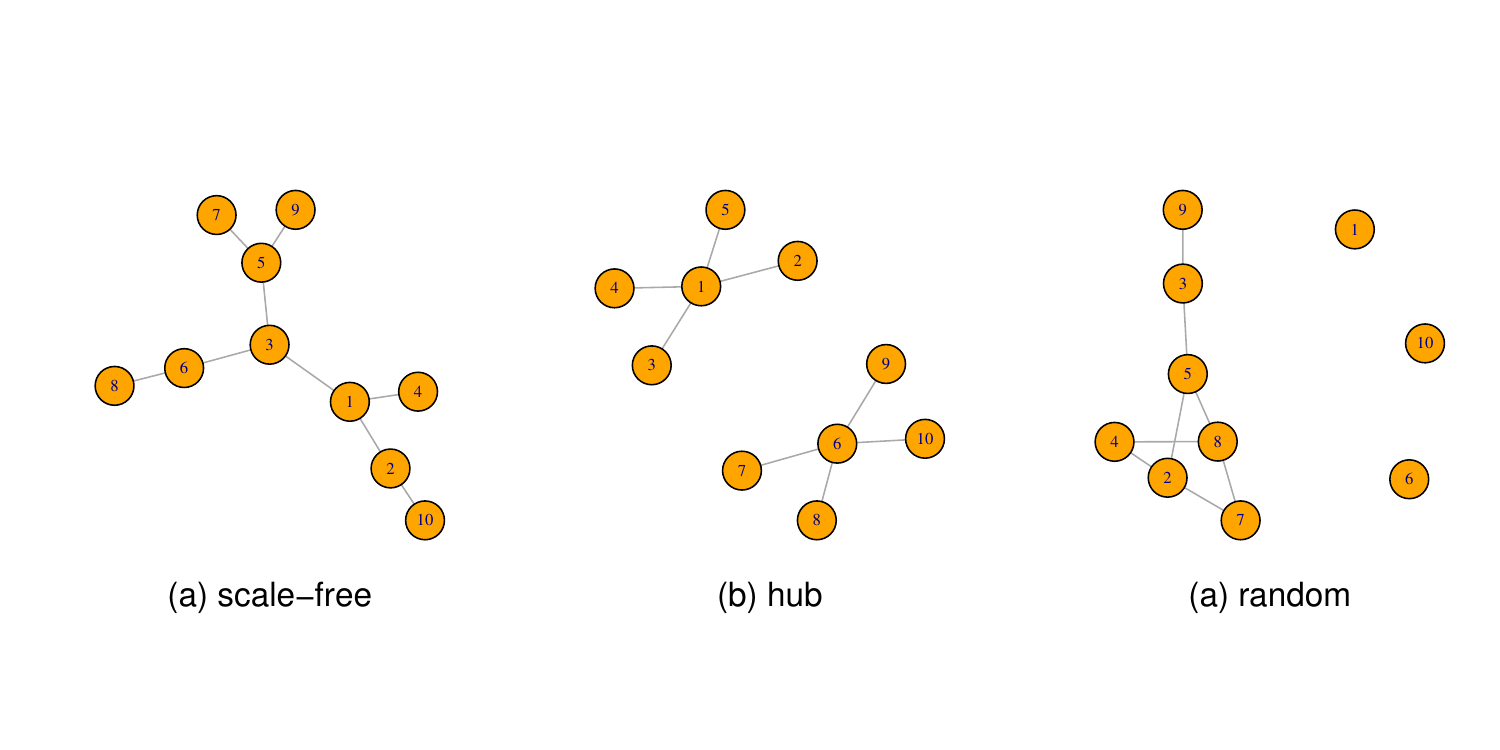}
			\vspace{-3em}
			\caption{ \small The graph structures for  $p=10$ employed in the simulation studies: (a) scale-free; (b)  hub; (c)  random graph.}
			\vspace{-2em}
			\label{graphtypes}
		\end{center}
	\end{figure}

	For each graph, 500 datasets were sampled for three sample sizes, $n=200,1000,2000$.  
	To generate the data, we followed the approach in \cite{allen2013local}.
	{Let $\mathbb{X}\in\mathbb{R}^{n\times p}$ be the set of $n$ independent observations of random vector $\bold{X}$.
		Then, $\mathbb{X}$ is obtained from the following model}
	$\mathbb{X}=\mathbb{Y}W+\mathbb{\epsilon},$
	where $\mathbb{Y}=(y_{st})$ is an $n\times (p+p(p-2)/2)$ matrix whose entries $y_{st}$ are realizations of independent random variables $Y_{st}\sim \text{Pois}(\lambda_{true})$ and $\mathbb{\epsilon}=(e_{st})$ is an $n\times p$ matrix with  entries $e_{st}$ which are realizations of random variables $E_{st}\sim \text{Pois}(\lambda_{noise})$. 
	Let $W$ be the adjacency matrix of a given true graph, then the adjacency matrix is encoded by matrix $W$ as $W=[I_p;P\odot(1_p tri(W)^T)]^T$. Here, $P$ is a $p\times (p(p-1)/2)$ pairwise permutation matrix, $\odot$ denotes the element-wise product, and $tri(W)$ is the $(p(p-1)/2)\times 1$ vectorized upper triangular part of $W$. As in \cite{allen2013local}, we simulated data at two signal-to-noise ratio (SNR) levels. We set $\lambda_{true}=1$ with $\lambda_{noise}=5$ for the low SNR level, and $\lambda_{noise}=0.5$ for the high SNR level.
	\vspace{-2.5em}
	\begin{figure}[htbp]
		\begin{center}
			\includegraphics[width = 0.9\linewidth, height=0.32\textheight]{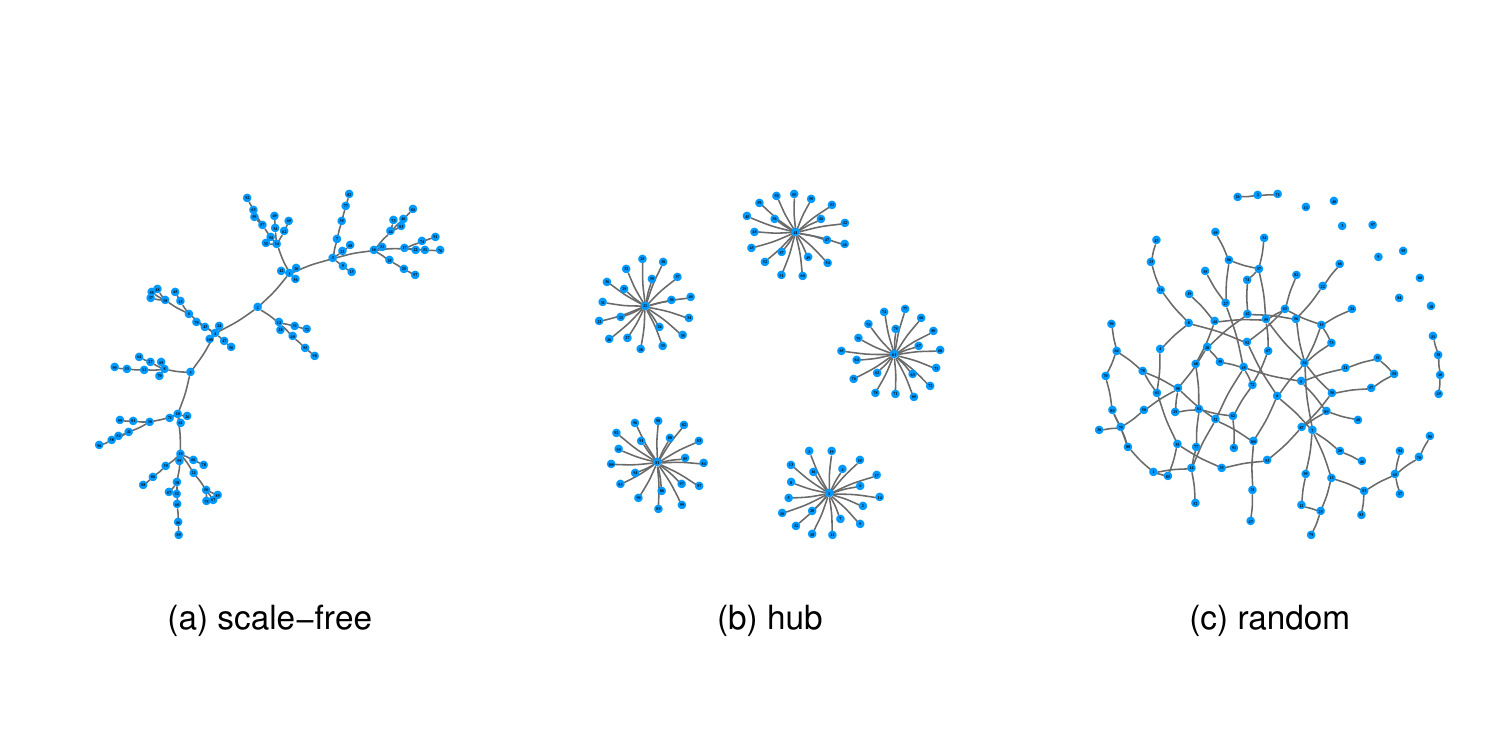}
			\vspace{-3em}
			\caption{\small The graph structures for  $p=100$ employed in the simulation studies: (a) scale-free; (b)  hub; (c)  random graph.}
			\label{graphtypes100}
			\vspace{-2em}
		\end{center}
	\end{figure}
	\vspace{-1.5em}
	
	\subsubsection{Results}
	The considered algorithms  are listed below, along with specifications, if needed, of tuning parameters. 
	Algorithms for Gaussian data have been used on log transformed data shifted by 1. Whenever a regularization parameter $\lambda$ had to be chosen, the StARS algorithm \citep{liu2010stability} was employed, which aims to seek the value of $\lambda \in (\lambda_{min},\lambda_{max})$, $\lambda_{opt} $ say, leading to the most stable set of edges. We refer the reader to Appendix \ref{suppC}, for details on the StARS algorithm and its tuning parameters, in particular the variability threshold $\beta$, the number of parameters $\lambda$, i.e., $nlambda$, and the number of subsamplings $B$. It is worth noting that, whenever the graph  corresponding to  $\lambda_{opt}$ was empty,  we shifted to the first nonempty graph (if it existed) in the decreasing regularization path. We therefore considered:
	\begin{itemize}
		\item[-] {\bf PC-LPGM:}    level of significance of tests $1\%$;
		\item[-] {\bf LPGM:}  $\beta=0.05$, $nlambda=10$, $B=20;$ $\frac{\lambda_{min}}{\lambda_{max}}=0.01$; $\gamma = 0.001$, $sth=0.9$; 
		\item[-] {\bf VSL:} $\beta=0.1$, $nlambda=10$, $B=20;$
		\item[-] {\bf GLASSO:}  $\beta=0.1$, $nlambda=10$, $B=20;$
		\item[-] {\bf NPN-Copula:} $\beta=0.1$, $nlambda=10$, $B=20;$
		\item[-] {\bf NPN-Skeptic:} 	$\beta=0.1$, $nlambda=10$, $B=20.$
	\end{itemize}
	
{	For the two considered vertex cardinalities,  $p=10, 100$, and for the chosen sample sizes $n=200, 1000, 2000$, Figure \ref{10-5} and \ref{100-5} plot  Monte Carlo means of TP, PPV and Se for each of considered method at low ($\lambda_{noise}=5$) and high ($\lambda_{noise}=0.5$) SNR levels.}
	Each value is computed as an average of  the 1500 values obtained by simulating 500 samples for each of the three networks. 
	Monte Carlo means (and standard deviations) of the same quantities disaggregated by network type are given in Appendix \ref{suppD},  Tables \ref{table1-chap1} -- \ref{table4-chap1}. 
	These results indicate that the PC-LPGM algorithm is consistent and  outperforms, on average, Gaussian-based competitors (VSL, GLASSO), nonparanormal-based competitors (NPN-Copula, NPN-Skeptic) as well as the state-of-the-art algorithms that are designed specifically for Poisson graphical models (LPGM, PDN) on average in terms of reconstructing the structure from given data. 
	
	
\begin{figure}[htbp]
	\centering
		\begin{subfigure}{\textwidth}
		\caption{$\lambda_{noise}=0.5$} \label{s10-05}
		\vspace{-2em}
		\includegraphics[width = 1\linewidth, height=0.35\textheight]{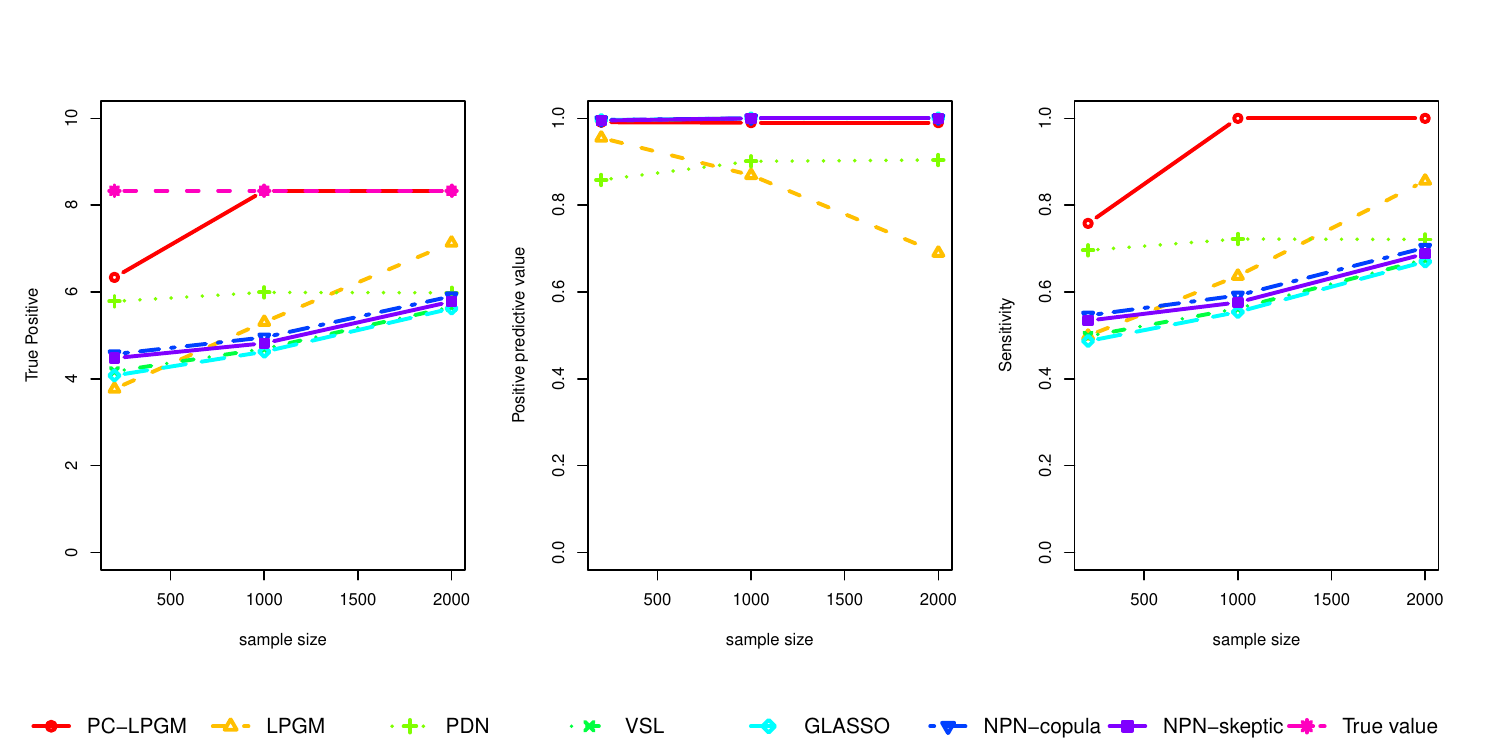}
		\end{subfigure}
\vspace{2em}
\newline
		\begin{subfigure}{\textwidth}
		\caption{$\lambda_{noise}=5$} \label{s10-5}
		\vspace{-2em}
		\includegraphics[width = 1\linewidth, height=0.35\textheight]{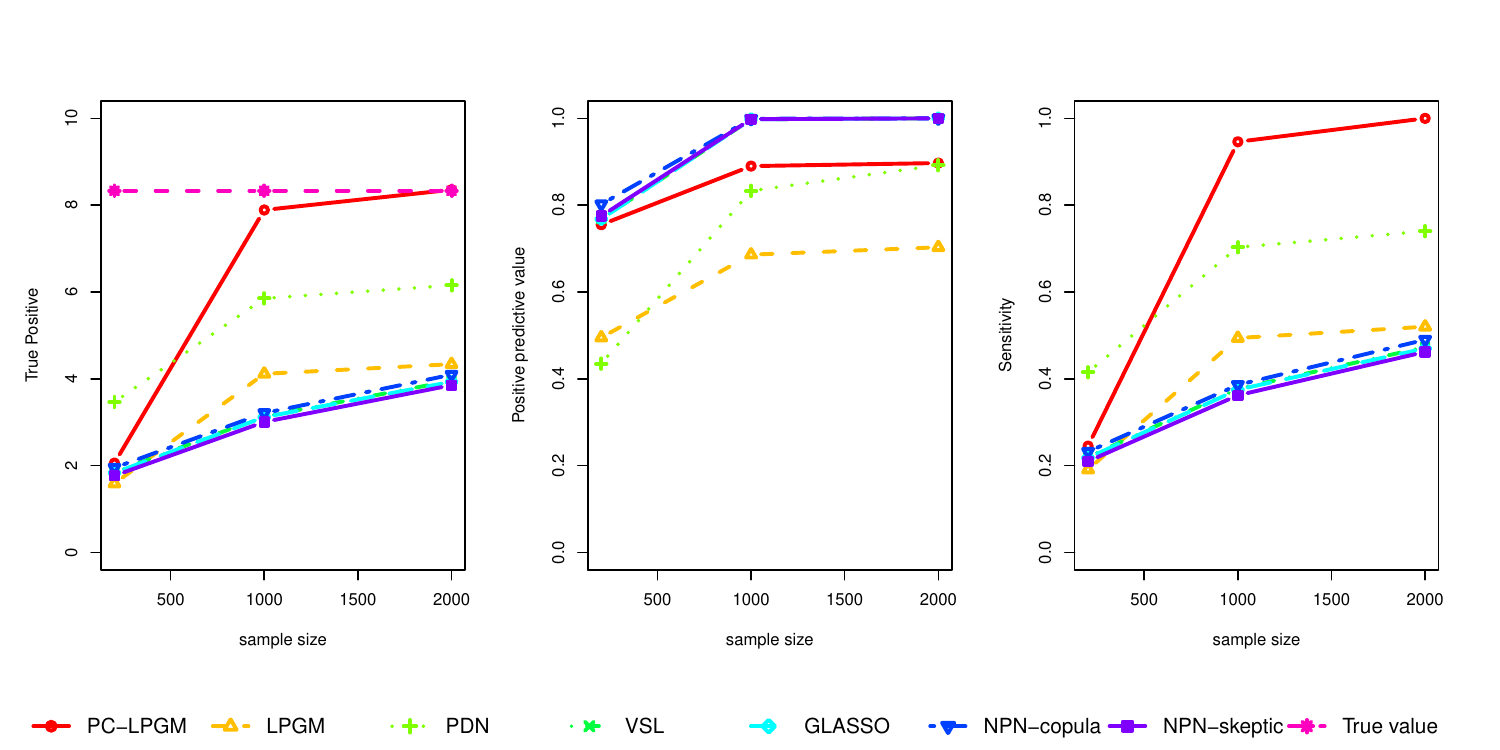}
		\end{subfigure}

\caption{\scriptsize Monte Carlo means of TP, PPV and Se for PC-LPGM; LPGM; PDN; VSL; GLASSO; NPN-Copula; NPN-Skeptic for networks in Figure~\ref{graphtypes} ($p=10$), sample sizes $n=200,\,1000,\, 2000,$  and SNR level $\lambda_{noise}=0.5$ (\subref{s10-05}) and $\lambda_{noise}=5$ (\subref{s10-5}).}
			\label{10-5}
\end{figure}

	When $p=10,$ the PC-LPGM algorithm reaches the highest TP value, followed by the PDN and the LPGM algorithms. When $n\ge 1000,$ PC-LPGM recovers almost all edges for both low and high SNR levels,  see { Figure \ref{10-5}. } 
A closer look at the PPV and Se plot  provides further insight of the behaviour of considered methods. Among the algorithms with highest PPV, PC-LPGM shows a sensitivity approaching 1 already at the sample size $n=1000$ for both a  high and a low SNR level.  It is worth noting that, LPGM algorithm was successful only for a high SNR level ($\lambda_{nois}=0.5$). 
	
	It is interesting to note that the performance of the PC-LPGM algorithm is far better than that of the competing algorithms employing the Poisson assumption,  {PDN and LPGM}. This might be explained in terms of difference between penalization and restriction of the conditional sets. In the LPGM  algorithm, as well as in the PDN algorithm, a prediction model is fitted locally on all other variables, by mean of a series of independent penalized regressions. In the PC-LPGM algorithm, the number of variables in the conditional sets is controlled and progressively increased from 0 to $p-2$ (or to the maximum number of neighbours $m$). In our simulations, this second strategy appears to be more powerful in the network reconstruction.

	The Gaussian based methods (VSL, GLASSO) perform reasonably well, with an inferior score with respect to the leading threesome only for the hub graph at high SNR level (see Table \ref{table1-chap1}, Appendix \ref{suppD}). 
	It is worth noting that sophisticated techniques that replace the Gaussian distribution with a more flexible continuous distribution such as the nonparanormal distribution, for example, NPN-Copula, NPN-Skeptic  show slight gains in accuracy over the naive analysis. 
%

\begin{figure}[htbp]
	\centering
		\begin{subfigure}{\textwidth}
		\caption{$\lambda_{noise}=0.5$} \label{s100-05}
		\vspace{-2em}
		\includegraphics[width = 1\linewidth, height=0.35\textheight]{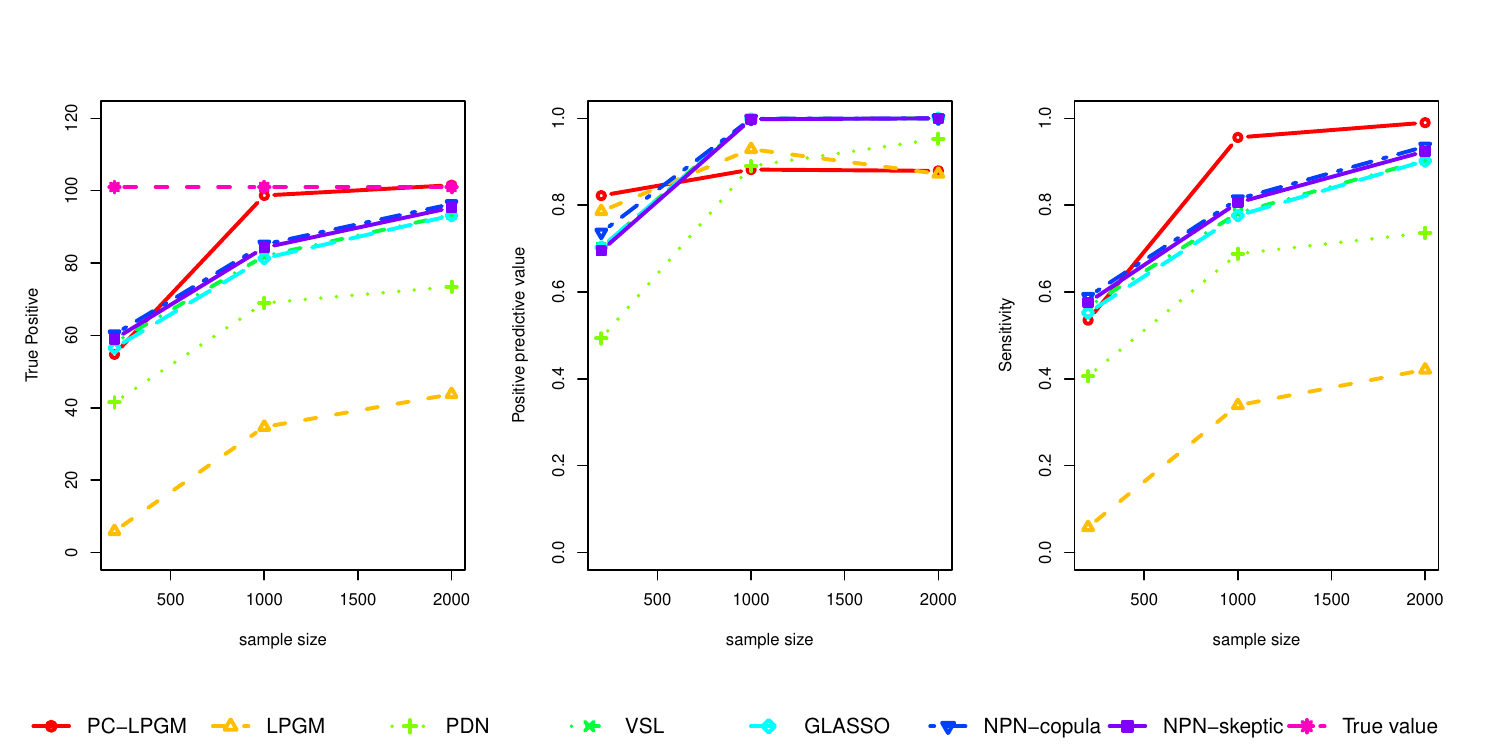}
		\end{subfigure}
\vspace{2em}
\newline
		\begin{subfigure}{\textwidth}
		\caption{$\lambda_{noise}=5$} \label{s100-5}
		\vspace{-2em}
		\includegraphics[width = 1\linewidth, height=0.35\textheight]{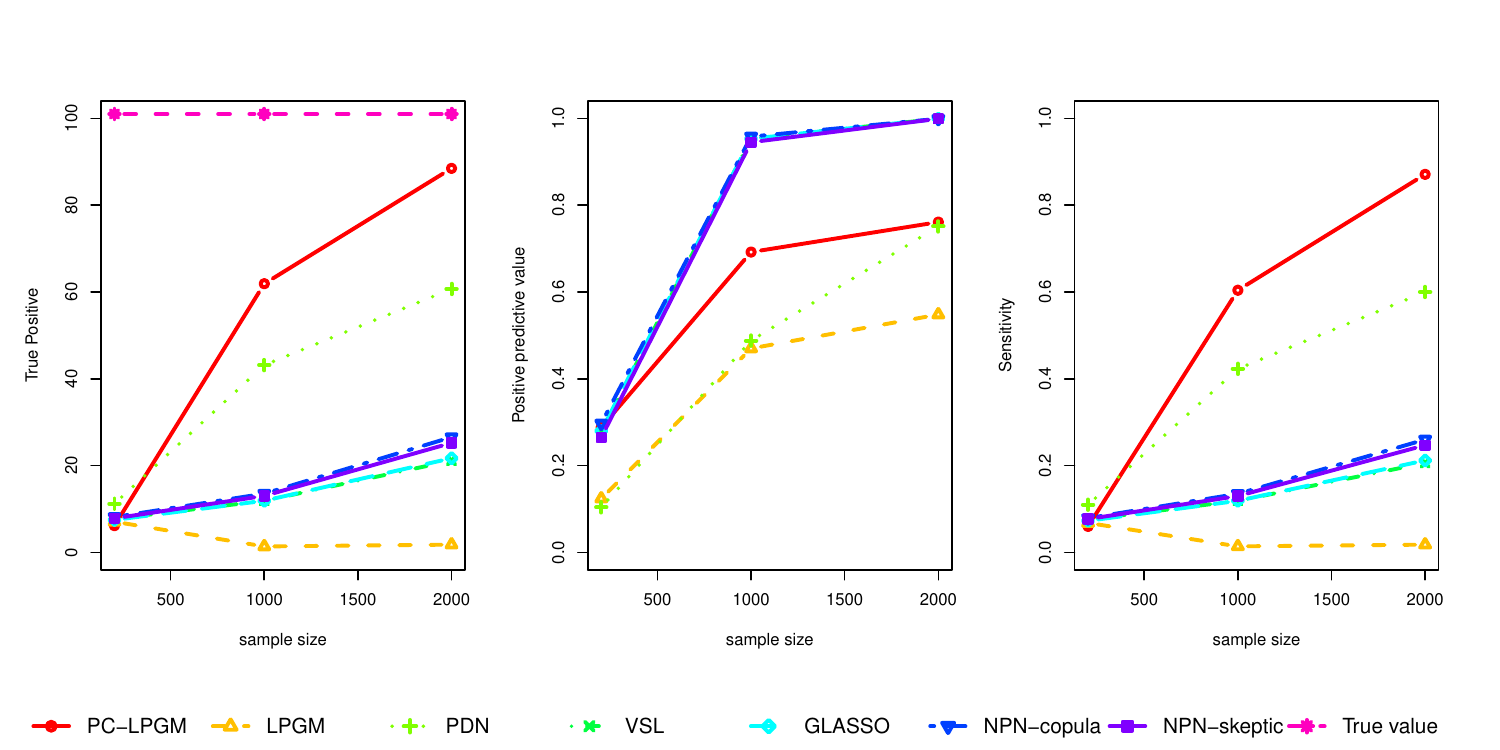}
		\end{subfigure}

\caption{\scriptsize Monte Carlo means of TP, PPV and Se for PC-LPGM; LPGM; PDN; VSL; GLASSO; NPN-Copula; NPN-Skeptic for networks in Figure~\ref{graphtypes} ($p=10$), sample sizes $n=200,\,1000,\, 2000,$  and SNR level $\lambda_{noise}=0.5$ (\subref{s100-05}) and $\lambda_{noise}=5$ (\subref{s100-5}).}
			\label{100-5}
\end{figure}
	
	Results for the high dimensional setting ($p=100$) are somehow comparable, as it can be seen in { Figure~\ref{100-5}.} 
The PC-LPGM outperforms all competing methods, and differences among algorithms are more evident. The TP score of PC-LPGM becomes  already reasonable when $n$ approaches 2000 observations. It is worth noting that performances of methods based on $l_1$-regularized regression, {LPGM in particular}, {} is, overall, less accurate and more variable in this scenario.  An  extensive analysis of the results revealed that  the graph recovered by LPGM is almost comparable to an empty graph in a number of cases.  These poor performances attracted the attention of one Reviewer and of the Action Editor, who asked for ``more challenging simulations, and extremely honest comparisons of the results".  Table \ref{table10-poisnew} and Table \ref{table100-poisnew} address this request. For the two vertex cardinalities, $p = 10$ and $p = 100,$ the tables report Monte Carlo  means of TP, PPV and Se obtained by simulating 500 samples from graphs in Figure \ref{graphtypes} for variables with Poisson node conditional distribution with mean $\lambda= 1$ and levels of noise $ \lambda_{noise} = 0.5,5$.  For both values of $p,$ sample sizes have been considered also below the limits for convergence. As performances of LPGM are highly dependent on the tuning of its parameters ($\beta$, $\gamma$, $sth$, etc), for this specific algorithm figures in Table \ref{table10-poisnew} and Table \ref{table100-poisnew} refer to the best combination of parameters that we managed to find ($B=50,$ $nlambda=20$, $\frac{\lambda_{min}}{\lambda_{max}}=0.01$, $\gamma = 10^{-6}$, $sth=0.6$, $\beta=0.1$ for $p=10$ and $\beta=0.05$ for $p=100$).  In other words,  we let, to the best of our abilities, LPGM work under a favourable tuning. Performance indicators overall show that competing algorithms still compare unfavourably with PC-LPGM algorithm in most scenarios. 

	Overall, results seem to demonstrate the good performances of PC-LPGM algorithm in all considered situations. 
{

\section{On the learning strategy of PC-LPGM}\label{discussion}
\noindent

	In the previous sections we have presented  PC-LPGM, given statistical guarantees, and discussed its consistency with respect to the model specification. 
	In this section, we will  further  discuss its learning strategy, i.e., hypothesis testing.  
	
	
	The} main ingredient that distinguishes PC-LPGM from its potential competitors is the use iterative hypothesis testing instead of penalized estimation.  This substitution offers, in our view, a number of possible advantages. Firstly, inheriting the advantages of the PC algorithm, it allows to easily implement sparsity by a control on the number of variables, $m,$ in the conditional sets, avoiding at the same time over-shrinking of small but significant covariate effects. Secondly, it offers computational advantages, especially when sparse networks are the target of inference.
	Finally, hypothesis testing is scale-invariant, i.e., is not affected by scale transformations of regressors.
	
	The following two empirical studies shed more light on  the above mentioned advantages by comparing PC-LPGM with its most natural counterpart, LPGM.  \\
	
	\noindent		{\bf Sparsity and computational costs.}
	{ In the first study, to guarantee a fair comparison between the two algorithms, we adopted the playground of Section 6, i.e., we worked with Poisson node-conditional distributions and unrestricted dependencies among variables.}
	For  two vertex cardinalities,  $p =10, 100$, and one sample size, $n=1000,$ we  compared PC-LPGM with proper test statistic $Z_{st|\bold{K}}^P$ and LPGM algorithm in learning  a number of random graphs having different edge probability $\pi$. In detail, we fixed nine values for $\pi,$ running from $0.1$ to $0.9$ for $p=10,$ and  from 0.01 to 0.09 for $p=100.$ {Considering that, in this setting, the total number of edges is a random variable with expected value $\pi \frac{p(p-1)}{2},$ we allow, on average,  up to 45 edges when $p=10$ and up to 445 when $p=100$.  For each value of $\pi$, a network was generated and 500 samples were simulated from it as in Section 6.1 at both the { high }  ($\lambda_{noise}=0.5)$ and the { low }  ($\lambda_{noise}=5)$ SNR level.  }The two algorithms were tuned as follows
	\begin{itemize}
		\item[-] {\bf PC-LPGM:}    level of significance of tests $1\%$; $m=8$ for $p=10$; $m=3$ for $p=100;$
		\item[-] {\bf LPGM:}  $\beta=0.	05$; $B=50;$ $\frac{\lambda_{min}}{\lambda_{max}}=0.01$; $\gamma = 10^{-6}$, $sth=0.6$, $nlambda=10$.
	\end{itemize}
	It is worth noting that, with $p=10$, we avoided limiting the cardinality of the conditional sets, that is, we did not impose any prior knowledge on sparsity of the graph.
	
	Table \ref{table10-time} 
	reports  Monte Carlo means of TP,  PPV, Se, and running time for the two algorithms in the two cardinality scenarios.
	The runtime analysis (second) was done on an CPU: Intel(R) Xeon(R) CPU E5-4650 v3 @ 2.10GHz on Linux and using R 3.5.1 and 20 cores.

		When $p=10$, PC-LPGM performs better than LPGM for almost all values of probability $\pi\le 0.5,$  highlighting the efficiency of PC-LPGM when dealing with sparse graphs. The average of the ratio of the runtime of PC-LPGM over that of LPGM is around 0.33, showing that PC-LPGM, on average, needs about one third of the time needed to LPGM.  
		When $p=100$, PC-LPGM reaches the highest PPV and Se in almost all cases with {comparable} runtime when the complexity bound is achieved, say $m\le 3$. 		
		It is worth noting that the computational complexity of PC-LPGM is an exponential function in $m$. In the worst case, the algorithm is infeasible if  $p$ and $m$ are both large. We refer the reader to \cite{kalisch2007estimating} for more comments about complexity of PC algorithm.
		
{
With the second study, we moved to the truncated Poisson model specification setting. In what follows, we denote by PC-TPGM  and TPGM, respectively, our algorithm for the truncated Poisson ($Z_{st|\bold{K}}^{TP}$ test statistic) and the \citet{allen2013local} algorithm under a truncated Poisson  model specification. The key difference in working with truncated models instead of their untruncated counterparts is the presence, in the conditional distributions, of the normalization term $\exp D(\langle \boldsymbol{\theta}_s,\bold{x}_{V\backslash\{s\}}\rangle)$, whose estimation impacts on the computational cost. To guarantee a fair comparison, we  re-implemented {\it ex-novo} the TPGM algorithm   and both PC-TPGM and TPGM were based on the same Nelder--Mead optimization algorithm, implemented in the {\tt R}  function \texttt{optim}. 

For  three vertex cardinalities, i.e., $p =10, 50, 100$, and one sample size, $n=1000,$ three random graphs were generated, each for a given  probability $\pi$ of edge inclusion. In detail, we chose $\pi=0.1,0.2,0.3$  for $p=10;$  $\pi=0.02,0.04,0.06$  for $p=50;$  and $\pi=0.01,0.02,0.03$  for $p=100.$  500 samples were simulated from each network as in Section 6.1 at both the  high ($\lambda_{noise}=0.5),$ and the low  ($\lambda_{noise}=5) $ SNR level. The tuning parameters were selected as previously specified, when needed.  {As expected, for both algorithms, the ability to reconstruct the true networks -as measured by the previously defined metrics- was confirmed to be unaltered over the model specification} (Poisson or truncated Poisson),  when a sufficiently large truncation point $R$ was fixed. For this reason, in what follows, we report only results on  runtime. 

Figure \ref{runtime} compares the runtime of the two algorithms.  For completeness,  runtimes under the Poisson model specification are also reported, with the familiar acronyms for the algorithms.
The figure  plots  Monte Carlo means (over the 1500 replications) of runtime for each of considered method at  high  and the low  SNR level.  As expected, PC-TPGM and TPGM algorithms showed a computational cost  higher than the one observed under the Poisson assumption, but superiority of PC-TPGM over TPGM is preserved.  Overall, PC-LPGM is the most efficient  algorithm,  followed by LPGM. For this reason, we advise the user to use, when the working conditions allow ($p$ large), PC-LPGM (or LPGM).}

			\begin{table}[ht]
				\centering
				{\scriptsize
					\caption{\label{table10-time}\scriptsize{ Monte Carlo  means of TP, PPV, Se and runtime obtained by simulating 500 samples 
							from two random graphs with $p = 10$ and $p=100$ variables with Poisson node conditional distribution and levels of noise $\lambda_{noise} = 0.5, 5$. The probability  of 
							edge inclusion $\pi$ runs from 0.1 to 0.9 for $p=10$, and from 0.01 to 0.09 for $p=100$.}}
					\begin{tabular}{c|c|c|rrrr|rrrr}
						\hline
				~&~&~&\multicolumn{4}{c|} {\bf LPGM}&\multicolumn{4}{c}{\bf PC-LPGM }\\
				$\lambda_{noise}$	&	$p$ & $\pi$ & TP  & PPV & Se & time & TP  & PPV & Se & time \\ 
						\hline
				~&		&0.1 & 7.058 & 0.649 & 0.784 & 8.372 & 8.401&0.976&  0.934 & 1.577 \\ 
				~&		&0.2 & 9.390 & 0.475 & 0.939 & 7.823 &10.000&0.983&  1.000 & 1.413 \\ 
				~&		&0.3 & 12.296 & 0.523 & 0.946 & 8.257 & 12.767 & 0.987&  0.982 & 1.595 \\ 
				~&		&0.4 & 21.006 & 0.873 & 0.955 & 7.568 & 17.490 & 0.996 & 0.795 & 2.465 \\ 
				~&		10&0.5 & 22.298 & 0.906 & 0.969 & 7.332 & 18.282 & 0.997 & 0.795 & 2.877  \\ 
				~&		&0.6 & 26.626 & 0.859 & 0.918 & 7.038 & 17.626 & 0.997 & 0.608 & 2.983 \\ 
				~&		&0.7 & 28.772 & 0.848 & 0.899 & 7.024 & 16.850 & 0.998 & 0.527 & 2.889 \\ 
				~&		&0.8 & 30.800 & 0.905 & 0.880 & 6.096 & 15.950 & 0.998 & 0.456 & 2.825  \\ 
				0.5	&	&0.9 & 31.432 & 0.928 & 0.827 & 6.331 & 14.626 & 0.998 & 0.385 & 2.785  \\ 
				&&&&&&&&&&\\
				~&		&0.01 & 34.064 & 0.748 & 0.946 & 81.192 & 36.000& 0.927&  1.000 & 8.908\\ 
				~&		&0.02 & 9.842 & 0.934 & 0.096 & 53.017 & 97.444 & 0.910 & 0.946 & 14.065 \\ 
				~&		&0.03 & 24.700 & 0.843 & 0.179 & 35.820 & 134.150 & 0.939 & 0.972 & 10.538  \\ 
				~&		&0.04 & 152.303 & 0.821 & 0.841 & 35.460 & 173.363 & 0.954 & 0.958 & 13.459 \\ 
				~&		100&0.05 & 188.227 & 0.337 & 0.798 & 39.953 & 214.367 & 0.961 & 0.908 & 19.874\\ 
				~&		&0.06 & 254.077 & 0.144 & 0.895 & 40.525 & 223.370 & 0.960 & 0.787 & 19.211 \\ 
				~&		&0.07 & 290.883 & 0.164 & 0.871 & 38.545 & 220.047 & 0.954 & 0.659 & 21.600 \\ 
				~&		&0.08 & 306.810 & 0.146 & 0.859 & 37.689 & 228.437 & 0.955 & 0.640 & 19.789 \\ 
				~&		&0.09 & 358.010 & 0.179 & 0.833 & 34.979 & 217.850 & 0.947 & 0.507 & 17.458 \\ 
				&&&&&&&&&&\\
						\hline
				&&&&&&&&&&\\
				~&		&0.1 & 8.857 & 0.556 & 0.984 & 11.351 & 6.807 & 0.958& 0.756& 1.368 \\ 
				~&		&0.2 & 8.730 & 0.897 & 0.873 & 10.989 & 8.500 & 0.967& 0.850& 1.254 \\ 
				~&		&0.3 & 12.593 & 0.902 & 0.969 & 11.032 & 11.757 & 0.992 & 0.904 & 1.347 \\ 
				~&		&0.4 & 20.357 & 0.559 & 0.925 & 10.626 & 16.027 & 0.998 & 0.728 & 2.221 \\ 
				~&		10&0.5&21.040 & 0.566 & 0.915 & 10.833 & 17.667 & 0.998 & 0.768 & 2.518  \\ 
				~&		&0.6 & 27.000 & 0.687 & 0.931 & 37.942 & 16.667 & 0.982 & 0.575 & 8.838 \\ 
				~&		&0.7 & 28.837 & 0.757 & 0.901 & 34.501 & 16.127 & 0.997 & 0.504 & 8.287  \\ 
				~&		&0.8 & 30.893 & 0.828 & 0.883& 26.211& 15.120 & 0.997 & 0.432&  7.734  \\ 
				5	&	&0.9 &32.117 & 0.908 & 0.845& 34.305& 13.853&  0.999&  0.365&  6.720  \\ 
				&&&&&&&&&&\\
				~&		&0.01 & 36.000 & 0.792 & 1.000 & 42.504 & 35.900&  0.918&  1.000 & 10.127 \\ 
				~&		&0.02 &103.000 & 0.895 & 1.000 & 42.199 & 96.400 & 0.902 & 0.936 & 11.372  \\ 
				~&		&0.03 &137.323 & 0.892 & 0.995 & 40.119 & 132.870 & 0.924 & 0.963 & 15.444  \\ 
				~&		&0.04&177.193 & 0.216 & 0.979 & 31.798 & 170.020 & 0.944 & 0.939 & 17.008 \\ 
				~&		100&0.05 & 224.397 & 0.399 & 0.951 & 29.219 & 207.693 & 0.954 & 0.880 & 17.998 \\ 
				~&		&0.06 & 244.570 & 0.474 & 0.861 & 29.790 & 212.047 & 0.952 & 0.747 & 14.699 \\ 
				~&		&0.07 &232.533 & 0.488 & 0.696 & 32.580 & 206.760 & 0.944 & 0.619 & 16.127  \\ 
				~&		&0.08 &252.373 & 0.529 & 0.707 & 66.931 & 213.083 & 0.945 & 0.597 & 33.321 \\ 
				~&		&0.09 &228.670 & 0.452 & 0.532 & 69.960 & 201.286 & 0.938 & 0.468 & 48.612 \\ 
						\hline
					\end{tabular}
				}
			\end{table}
		
\begin{figure}[htbp]
\centering
		\includegraphics[width = 0.9\linewidth, height=0.42\textheight]{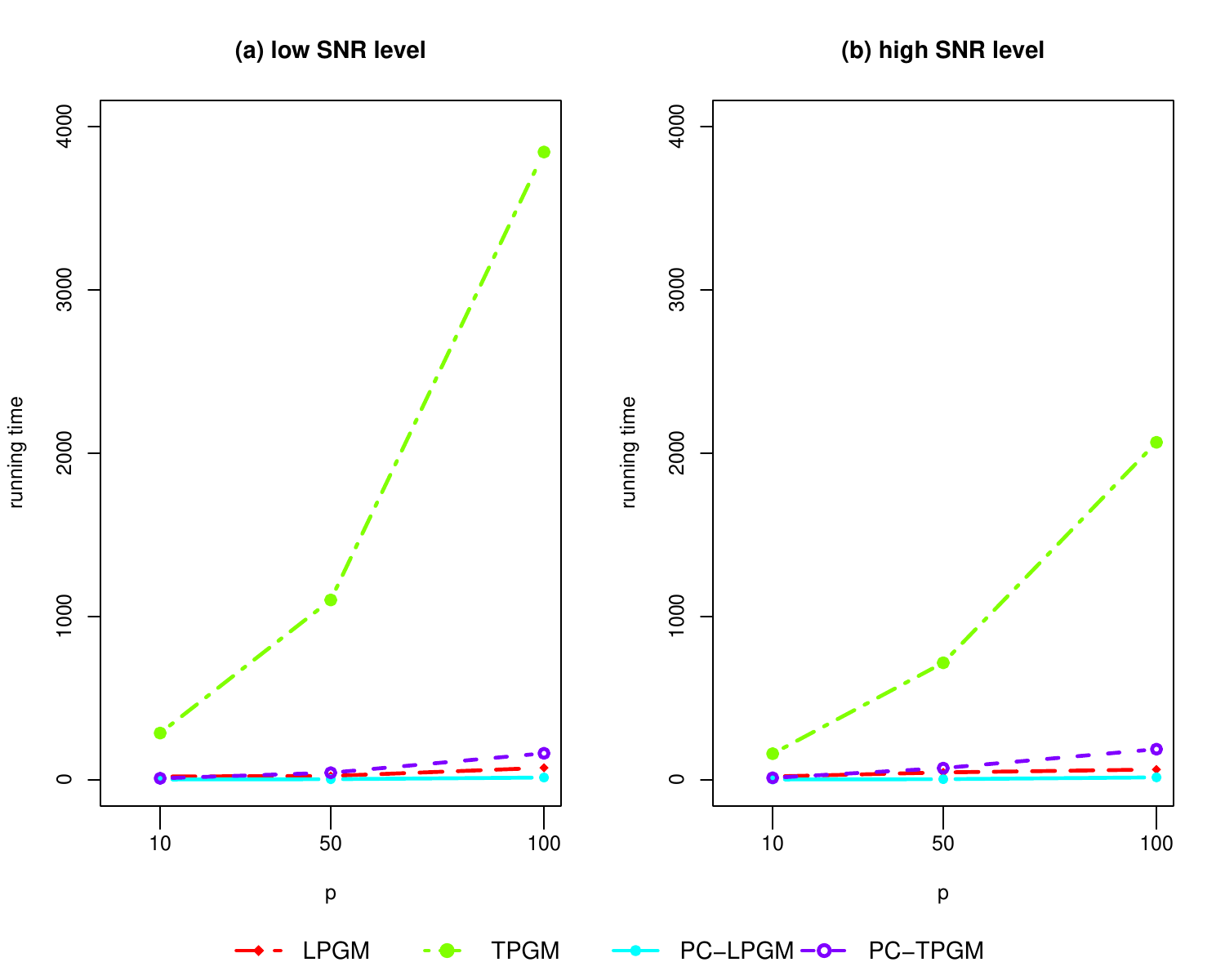}

\caption{\scriptsize Monte Carlo means of runtime (in seconds) obtained by simulating  500 samples from one random graph for each probability  of edge inclusion $\pi$ associated to each vertex cardinality $p=10, 50,100$ ($\pi=0.1, 0.2, 0.3$  for $p=10$; $\pi=0.02, 0.04, 0.06$ for $p=50$; and $\pi=0.01, 0.02, 0.03$ for $p=100$).   Sample size $n=1000$. 
}
\label{runtime}
\end{figure}	

	\noindent
	\\ {\bf Scaling of variables.}
	Feature scaling is one of the most typical, sometimes mandatory, data preprocessing steps performed prior to downstream analysis. Some form of scaling is typical, for example, in the preprocessing of next generation sequencing data, as shown also in Section~\ref{realanalysis}.  One advantage of iterative hypothesis testing versus penalized regression is invariance with respect to scaling of regressors in conditional models. As regularized methods have different levels of sensitivity to feature scaling,  in this section, we will look at one simulated scenario and compare the performance of PC-LPGM and LPGM with respect to scaling. A random network with  edge  probability equal to 0.3 was generated for $p=10$ nodes. For such network, 500 samples of size  $n=1000$  were generated  { as in Section 6.1 at the high  SNR level ($\lambda_{noise}=0.5$). Two variables, $X_1$ and $X_6$, were then scaled in turn. The choice of these variables  was driven by the size of their neighbourhood. Indeed, we selected the two nodes having, respectively, the smallest and the largest neighbourhood's  size ($N(1)=1$ and  $N(6)=4$), where $N(i)$ denotes the number of neighbours of node $i$. }  To scale the variables, we multiplied them for a factor  $k,$ for different values of $k.$   We then applied PC-LPGM and LPGM algorithm on the scaled data.   Results are reported in Table \ref{scaleproblem}. 
		\begin{table}[ht]
		
		\centering
		\begin{scriptsize}
			
			\caption{	\label{scaleproblem}\scriptsize { Monte Carlo means of TP, FP, FN, PPV, and Se obtained by simulating 500 samples 
					from one random graph on $p = 10$ variables with Poisson node conditional distribution and level of noise $\lambda_{noise} = 0.5$. The probability  of 
					edge inclusion is $\pi=0.3.$}}
			\begin{tabular}{c|c|rrrrr|rrrrr}
				\hline
				&&\multicolumn{5}{c|} {\bf LPGM}&\multicolumn{5}{c}{\bf PC-LPGM }\\
				$i$&$k$& TP & FP & FN & PPV & Se &  TP & FP & FN & PPV & Se  \\ 
				\hline
				&1/5 & 8.490 & 9.494 & 4.510 & 0.516 & 0.653 & 11.690 & 0.336 & 1.310 & 0.974 & 0.899 \\ 
				6&1/2 & 10.064 & 8.986 & 2.936 & 0.563 & 0.774 & 12.098 & 0.402 & 0.902 & 0.970 & 0.931 \\
				& 1 & 12.296 & 11.832 & 0.704 & 0.523 & 0.946 &  11.846 & 0.056 & 1.154 & 0.996 & 0.911  \\  
				&5 & 4.022 & 0.032 & 8.978 & 0.996 & 0.309 & 12.248 & 0.524 & 0.752 & 0.962 & 0.942 \\ 
				&10 & 4.006 & 0.090 & 8.994 & 0.985 & 0.308 & 12.224 & 0.600 & 0.776 & 0.956 & 0.940 \\ 
				&&&&&&&&&&&\\
				&1/3 & 11.278 & 10.038 & 1.722 & 0.540 & 0.868 & 12.056 & 0.358 & 0.944 & 0.973 & 0.927 \\ 
				1&1/2 & 12.200 & 9.794 & 0.800 & 0.567 & 0.938 & 12.014 & 0.340 & 0.986 & 0.975 & 0.924 \\ 
				& 1 & 12.296 & 11.832 & 0.704 & 0.523 & 0.946 &  11.846 & 0.056 & 1.154 & 0.996 & 0.911  \\
				&5 & 1.070 & 0.042 & 11.930 & 0.979 & 0.082 & 12.892 & 0.462 & 0.108 & 0.968 & 0.992 \\ 
				&10 & 1.154 & 0.398 & 11.846 & 0.861 & 0.089 & 12.910 & 0.480 & 0.090 & 0.966 & 0.993 \\ 
				\hline
			\end{tabular}
		\end{scriptsize}
	\end{table}
	
	Results show, as expected, invariance of PC-LPGM with respect to all scalings, with high PPV and Se values. On the other side, scale-variance of LPGM is evident, particularly when variance inflation of the regressor is performed ($k=5,10$). If the predictor is scaled up a lot, the corresponding coefficient is not shrunken by LPGM, increasing the number of false negatives.

	\section{Real data analysis: inferring networks from next generation sequencing data}\label{realanalysis}
	\noindent
	 To make our evaluation of PC-LPGM stronger, we perform some biological validation by applying {PC-LPGM (with $m=3$)} to two different datasets, one, retrieved from the Cancer Genome Atlas, on level III breast cancer microRNAs (miRNAs) expression;  and one, downloaded from the Gene Expression Omnibus (GEO), on olfactory epithelium stem cell. Here, we expect to obtain results coherent with the current biological knowledge. 
	 	
	\subsection{Breast cancer}\label{breastcancer}
	miRNAs are non-coding RNAs that are transcribed but do not encode proteins. 
	miRNAs have been reported to play a pivotal role in regulating key biological processes, for example, 
	post-transcriptional modifications and translation processes. 
	Some studies revealed that some disease-related miRNAs  can 
	indirectly regulate the function of other miRNAs associated with the same phenotype. 
	In this perspective,   studying the features of the interaction pattern of  miRNAs in some conditions 
	might help understand complex phenotype conditions.
	
	Here, we consider level III breast cancer. Our interest lies in the pattern of interactions
	among miRNAs, with a particular focus on the existence of hubs. In fact,  nodes with atypically
	high numbers of connections represent sites of signalling convergence with potentially
	large explanatory power for network behaviour or utility for clinical prognosis and therapy.
	By applying our algorithm, we expect to obtain results in line with known associations between
	miRNAs and breast cancer, and possibly gain more understanding of the nature of their effect on other genes.
	In other words, we expect some miRNAs associated with this phenotype to be  the hubs of  our
	estimated structure.
	\begin{figure}[htbp]
		\begin{center}
			\includegraphics[width = 0.7\linewidth, height=0.25\textheight]{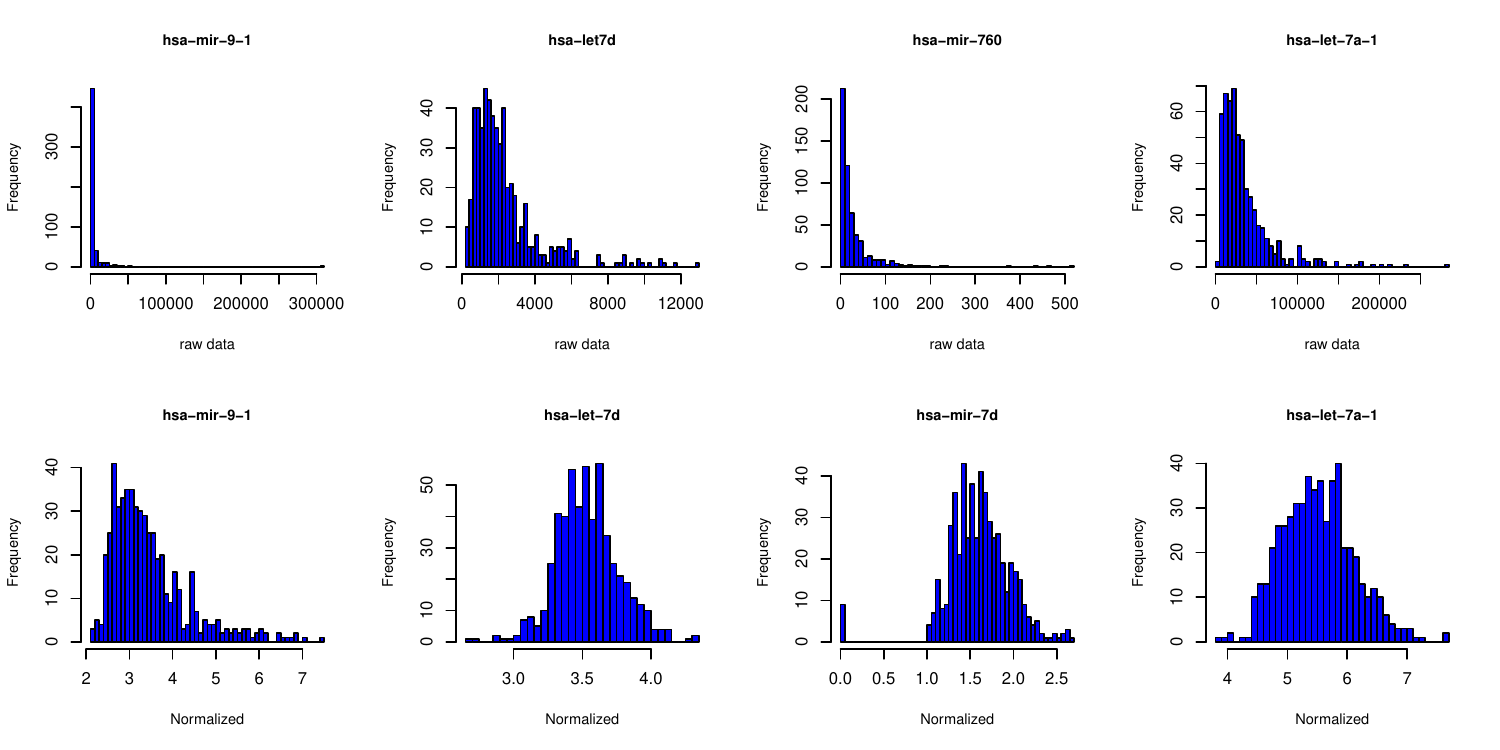}
			\caption{\scriptsize Distribution of four miRNA-Seq: raw data (top), normalized data (bottom).}
			\label{mir}
		\end{center}
		\vspace{-1.25em}
	\end{figure}
	
	miRNAs expression, obtained by high-throughput sequencing,  was downloaded from The Cancer Genome Atlas (TCGA) portal (\url{https://tcga-data.nci.nih.gov/docs/publications/brca_2012/}). The raw count data 
	set consisted of 544 patients and 1046 miRNAs. As measurements were zero-inflated and highly skewed, 
	with total count volumes depending on experimental condition,
	standard preprocessing was applied to the data  \cite[see][]{allen2013local}.
	In particular: we normalized the data by the $75\%$ quantile matching \citep{bullard2010evaluation};
	selected top $25\%$ most variable mirRNAs across the data;
	used a power transform $X^\alpha$ for $\alpha\in [0,1]$ with $\alpha$ chosen via the minimum Kolmogorov-Smirnov statistic \citep{li2012normalization} {giving rise to $\alpha=0.164$}. 
	The miRNAs with little variation across the samples were filtered out,  leaving 544 patients ($n=544$) and 261 miRNA ($p=261$). The effect of preprocessing  on four prototype miRNA are shown in Figure ~\ref{mir}. 
	
	
	Normalized data was used as input to PC-LPGM. A significance level of 5\% resulted in a spare graph is  shown in Figure~\ref{BRCgraph}. 
	\begin{figure}[htbp]
	\begin{tabular}[t]{ll}
	 \begin{minipage}{0.45\linewidth}
	
		\begin{center}
		
			\includegraphics[width = 1\linewidth, height=0.32\textheight]{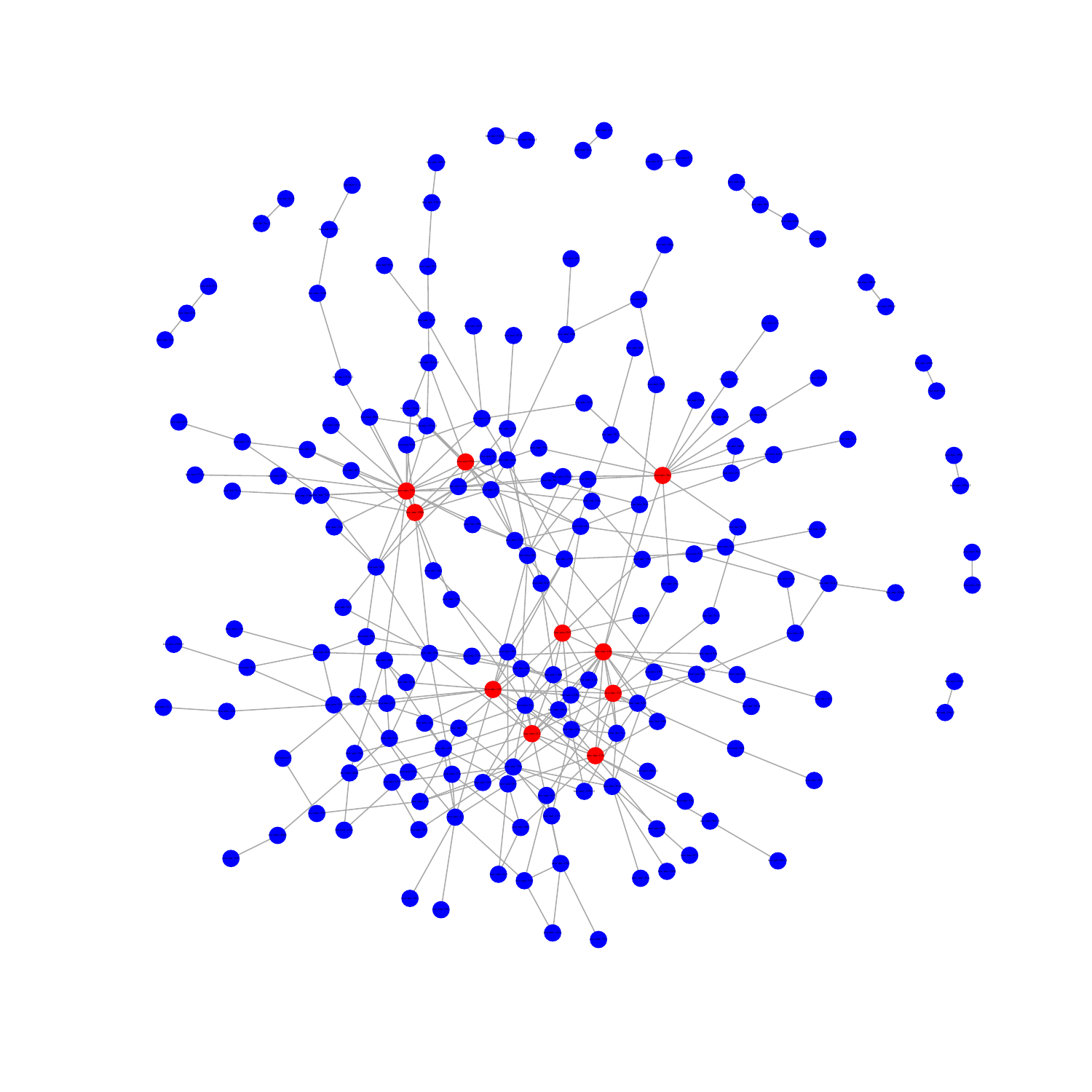}
		\end{center}

	 \end{minipage}
		& 
	 \begin{minipage}{0.45\linewidth}
		\begin{center}
			\includegraphics[width = 1\linewidth, height=0.3\textheight]{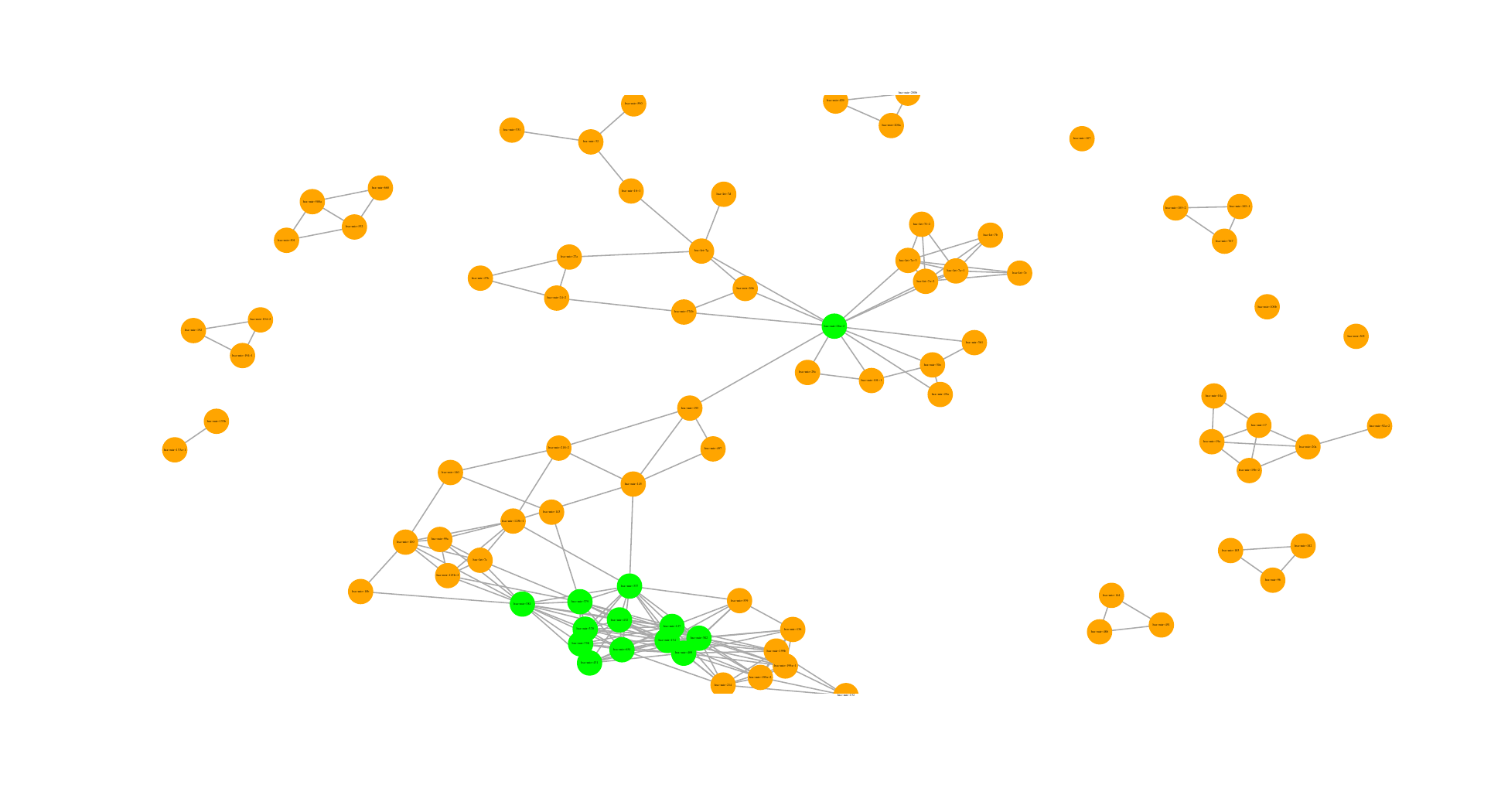}
		\end{center}
	  \end{minipage}
	  \end{tabular}
	  \caption{\scriptsize Breast cancer miRNA network estimated by the PC-LPGM algorithm (left, hub nodes coloured red),  and by the npn-COPULA algorithm (right, hub nodes coloured green).}
	  \label{BRCgraph}
	  \end{figure}
	%
	%
	
	We identified ten hub nodes (number of edges greater than 9) in the network,  miR-10b, -30a, -143, -375, -145, -210, -139, -934, -190b, -590. Almost all of them are known to be related to breast cancer \citep{volinia2012breast}, providing a biological validation of the potential of the algorithm to recover the sites of the network with high explanatory power. In particular, miR-10b and -210 highly express in breast cancer, when high expression is related to poor prognosis; miR-30a, -143 and -145 appear to be inhibitors of progression, and should therefore be low in patients with good survival \citep{zhang2014microrna, yan2014mir}.  These results play the role of a biological validation
	of the ability of PC-LPGM to retrieve structures reflecting existing relations among variables. 
	
	The reader is referred to \citet{allen2013local} for results of the application of LPGM  to the same dataset.  A structural comparison shows that PC-LPGM identifies less edges than LPGM and some common hub nodes, such as  miR-10b, and miR-375. 
{ To evaluate effectiveness of methods based on the Gaussian assumption, which require preliminary data transformation,  we ran the NPN-COPULA algorithm (tuning: $\beta=0.01$; $B=50$) on  log transformed data shifted by 1. Figure~\ref{BRCgraph} (right) shows the resulting graph. This algorithm identified 13 hub nodes: miR-26a-2, -127, -379, -134, -381, -337, -431, -409, -654, -758, -382, -370, -432, none of which coincides with those found by  PC-LPGM and only one,  miR-379,  is common to those found by LPGM. }

	\subsection{Olfactory epithelium stem cell} 
Recently, whole-transcriptome profiling of single cells by RNA sequencing has been developed as a powerful method for discriminating the heterogeneity of cell types and cell states in a complex population \citep{gadye2017injury}. 

Here, we re-analyse a subset of the data presented in \cite{gadye2017injury}. We focus on the olfactory sensory neurons lineage, which starts from the horizontal basal cell (HBC) stem cells and through a series of intermediate states generates mature olfactory neurons. Wild-type HBC stem cells were collected by fluorescence-
activated cell sorting (FACS), and profiled by single-cell RNA-seq. As before, we also focus on the existence of hub genes on the network of interactions among genes.  In fact, the identification of hubs in the gene network could help pinpoint important transcription factors, i.e., genes that regulate the expression of a large number of other genes in the system and could be targeted for follow-up experiments.  We therefore expect to obtain results in line with known associations between genes and the developmental trajectory of stem cells, and possibly gain more understanding on the nature of their effect on other genes. In other words, we expect some genes associated with this mechanism to be  the hubs of  our estimated structure.

Gene expression, obtained by high-throughput sequencing,  was downloaded from the Gene Expression Omnibus (GEO) (\url{https://www.ncbi.nlm.nih.gov/geo/query/acc.cgi?acc=GSE99251}). The raw count data 
set consisted of 542 cells and we selected 850 transcription factor genes (the list of transcription factors was downloaded from \url{https://github.com/diyadas/HBC-regen/tree/master/ref}). The data were zero-inflated and highly skewed. We then select top $20\%$ variables with highest mean to apply the standard preprocessing steps as in Section \ref{breastcancer}. In particular, we normalized the data by $75\%$ quantile matching;
	selected top $50\%$ most variable mirRNAs across the data; and
	used a power transform $X^\alpha$ { (with $\alpha=0.219$)}.
The genes with low mean and little variation across the samples were filtered out,  leaving 542 cells ($n=542$) and 85 genes ($p=85$). The effect of preprocessing  on four prototype genes are shown in Figure ~\ref{tfRNA}.

\begin{figure}[htbp]
	\begin{center}
		\includegraphics[width = 0.7\linewidth, height=0.25\textheight]{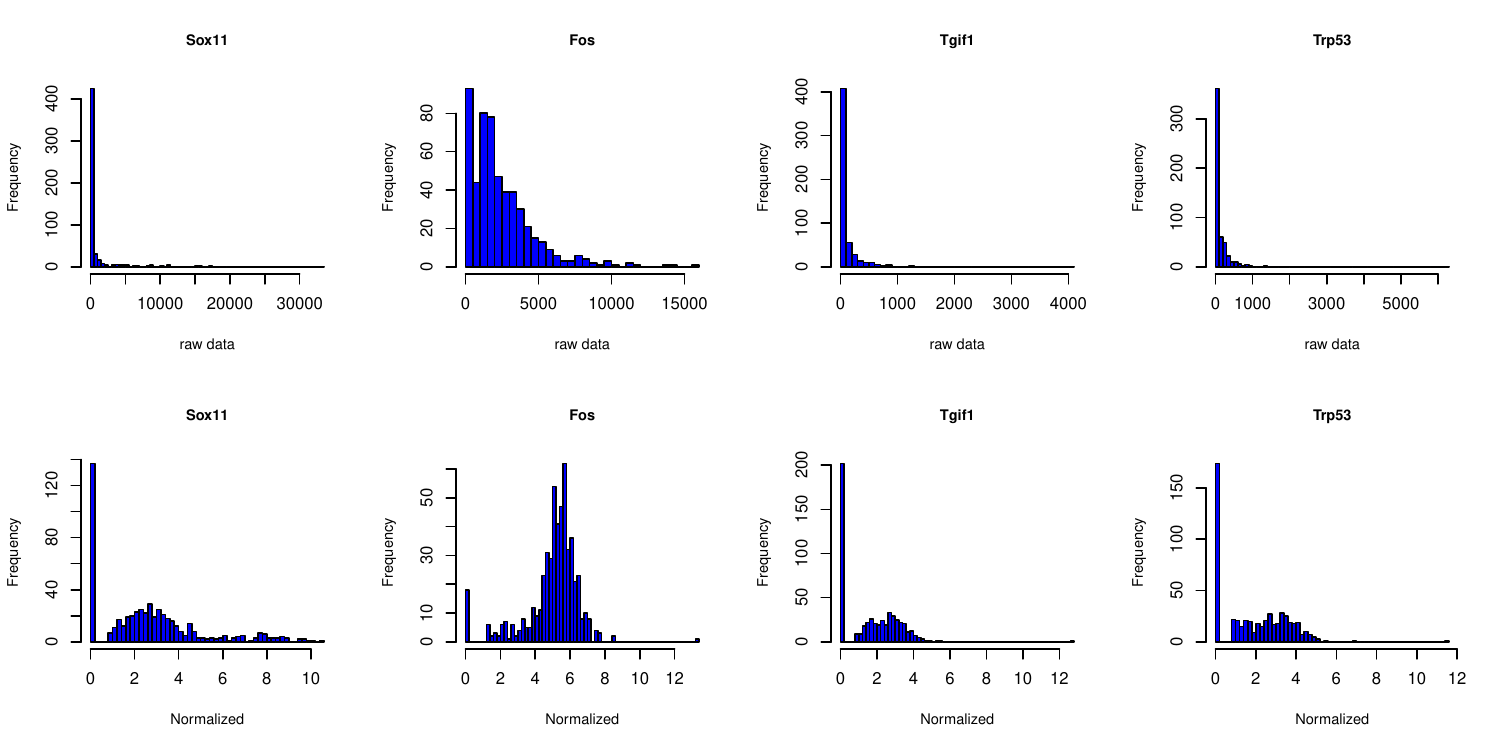}
		\caption{ Distribution of RNA-Seq read counts for four genes: raw data (top), normalized data (bottom).}
		\label{tfRNA}
	\end{center}
\end{figure}

A significance level of 5\% resulted in a spare graph by applying PC-LPGM algorithm on the normalized data, as shown in Figure~\ref{neurongraph}. 

\begin{figure}[htbp]
	\begin{center}
		\includegraphics[width = 1\linewidth, height=0.38\textheight]{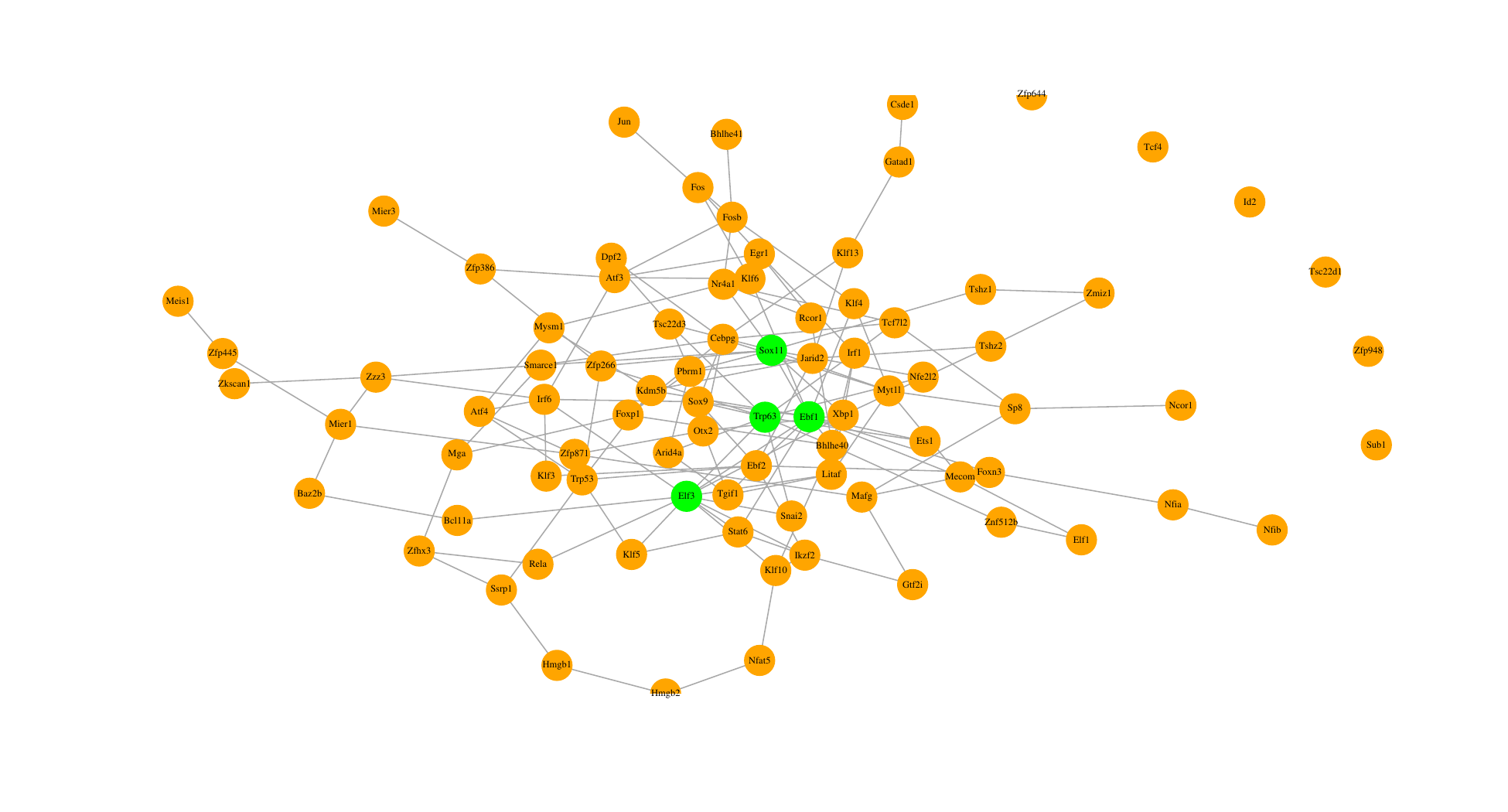}
		\vspace{-3em}
		\caption{ Olfactory Neuron gene network estimated by the PC-LPGM algorithm (hub nodes coloured green).}
		\label{neurongraph}
	\end{center}
\end{figure}


We identified four hub nodes, i.e. nodes with more than 9 edges, in the network: Sox11, Ebf1, Elf3, Trp63. Almost all of them are known to be related to the developmental trajectory of stem cells. In particular, Tpr63 is a gene essential to maintain the quiescent state of HBCs. In fact, by knocking out this gene HBCs will differentiate into mature cell types \citep{fletcher2017deconstructing}. It is therefore very reassuring that our method identified this gene as a hub node. Sox11 is known to be essential in neurogenesis \citep{ninkovic2013baf}.  Ebf1 is a transcription factor known to be involved with the later part of the lineage, in the final phase of maturing neurons \citep{wang1997characterization, garel1997family}. While there is no specific indication in the literature that Elf3 plays an important role in the olfactory epithelium, this gene is known to be important in other epithelial systems \citep{tymms1997novel}. It
will be interesting to follow up with experimental validation of the role of this gene.

 {   Gene networks estimated by the LPGM algorithm (tuning: $\beta=0.05$; $B=50;$ $\frac{\lambda_{min}}{\lambda_{max}}=0.01$; $\gamma = 10^{-6}$), and the NPN-COPULA algorithm (tuning: $\beta=0.05$; $B=50$;  log transformed data shifted by 1) are shown in Figure \ref{OEgraph}. The LPGM algorithm identified four hub nodes, namely: Sox11, Ebf1, Atf3, Fos,  two of which are common to PC-LPGM.  
The NPN-COPULA algorithm identified eleven hub nodes, namely: Fos,    Egr1,  Ebf1,  Klf6,   Atf3,   Ebf2,   Fosb,  Myt1l, Klf4,   Tcf7l2, Nr4a1, one of which is common to PC-LPGM.

In both the case studies that we considered in this Section, the true gene networks are unknown, and so is, to a large extent,  knowledge on the underlying biological processes. It is therefore difficult to interpret differences in the results obtained by different algorithms.  Based on the observation that the highest level of agreement between results occurs when comparing  PC-LPGM  and  LPGM outputs, this exercise highlights that methods that correctly exploit discreteness of the data tend to retrieve from data similar information.  }

\begin{figure}[htbp]
	\begin{tabular}[t]{ll}
	 \begin{minipage}{0.45\linewidth}
	
		\begin{center}
		
			\includegraphics[width = 1\linewidth, height=0.25\textheight]{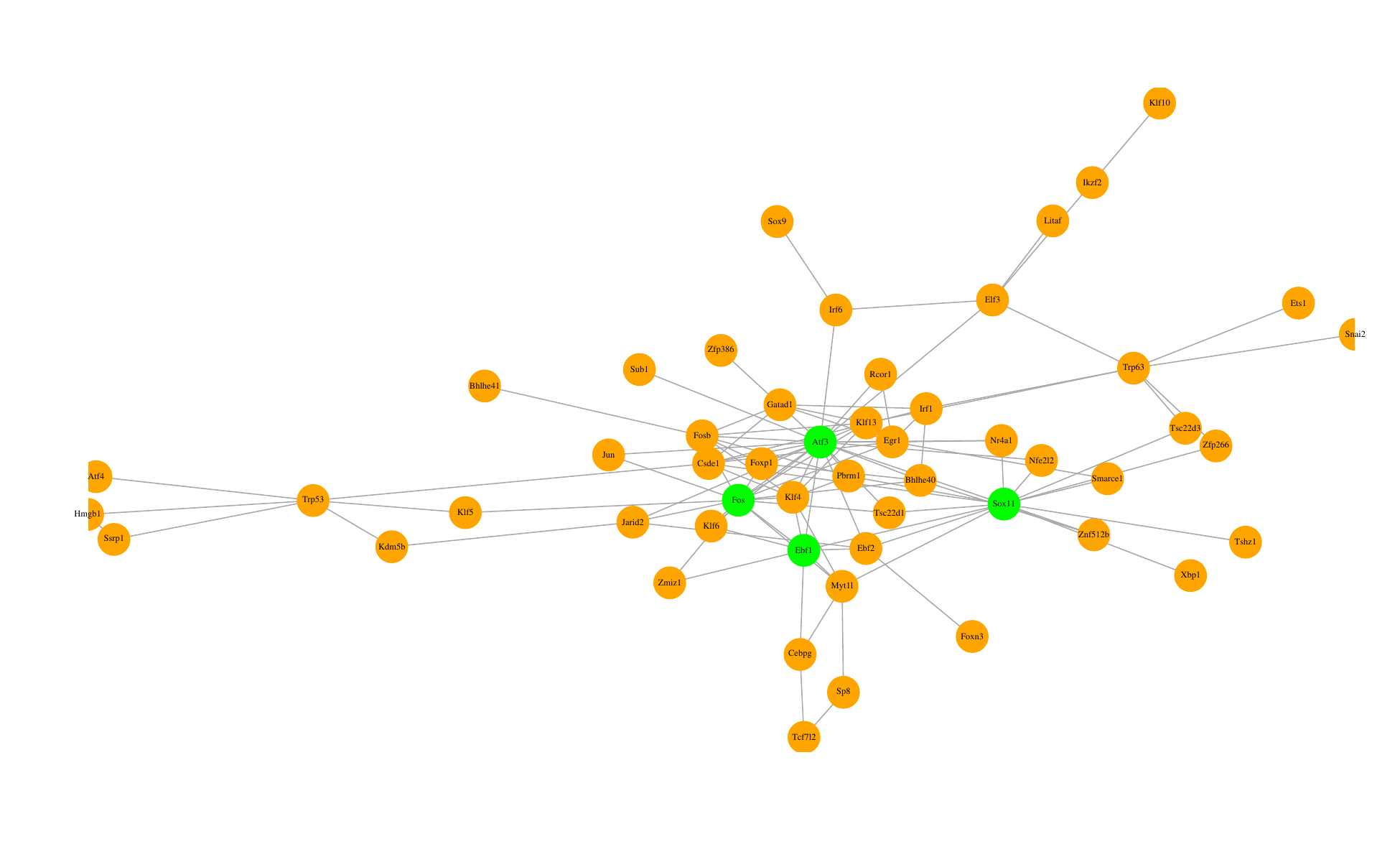}
		\end{center}

	 \end{minipage}
		& 
	 \begin{minipage}{0.45\linewidth}
		\begin{center}
			\includegraphics[width = 1\linewidth, height=0.23\textheight]{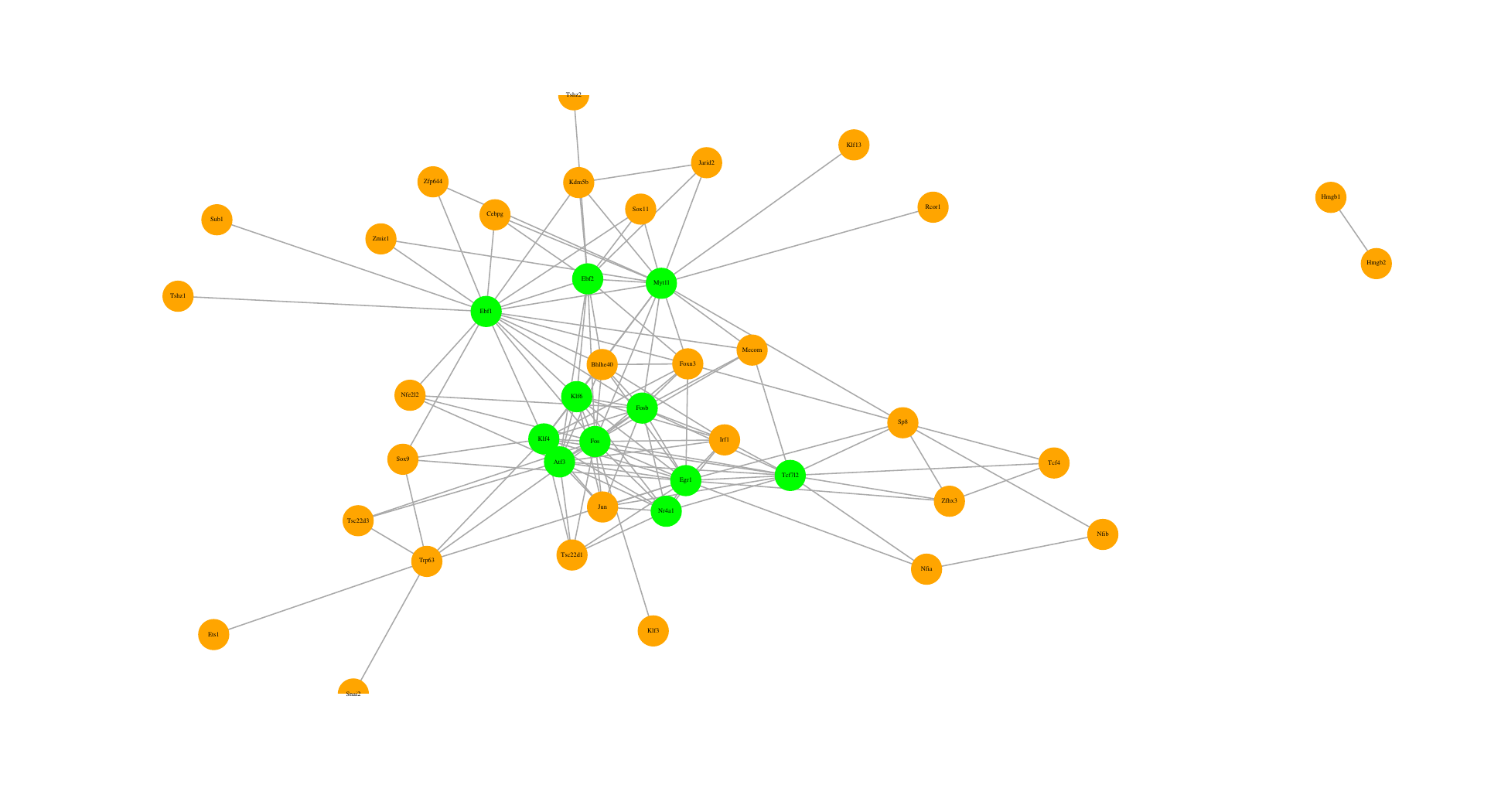}
		\end{center}
	  \end{minipage}
	  \end{tabular}
	  \caption{\scriptsize Olfactory Neuron gene network estimated by the LPGM algorithm (left, hub nodes coloured green),  and by the npn-COPULA algorithm (right, hub nodes coloured green).}
	  \label{OEgraph}
	  \end{figure}



\section{Conclusions}\label{conclusions}
\noindent

The main contribution of this paper is a careful analysis of the numerical and statistical efficiency of PC-LPGM, a simple method for structure learning of undirected graphical models for count data. A key strategy of our approach is controlling the number of variables in the conditional sets, as done in the PC algorithm. In this way, we control problems of estimation when the number of random variables $p$ is large possibly goes to infinity.

Our main theoretical result on truncated Poisson counts provides sufficient conditions on the set $(n, p, d,R)$ and on the model parameters for the method to succeed in consistently estimating the neighbours of every node in the graph. Precisely, Theorem \ref{mainresult} not only specifies sufficient conditions but it also provides the probability with which the method recovers the true edge set. Indeed, Equation \eqref{totalerror} shows that

\begin{eqnarray*}
	&&\mathbb{P}_{\boldsymbol{\theta}}(\text{ a type I or II error occurs in testing procedure})\\
	&&~\quad \le O_p\bigg(\exp\big\{ -c_1n^{2b}+c_0'd\log p\big\}-\exp\left\{-c_2\dfrac{n}{d^2}+c_0'd\log p\right\}\bigg).
\end{eqnarray*}
Hence, the right-hand side of the above given equation tends to 0 if  $n>O_p(d^3\log p)$. Moreover, Proposition \ref{md1}, and Lemma \ref{Fisher} require $$n> \max\left\{O_p\left(R^2\kappa_1\log p\right), O_p\left(\kappa_1^2R^4p^2\log p\right)\right\}=O_p\left(\kappa_1^2R^4p^2\log p\right)$$
{ Thus, the sufficient condition for convergence becomes
$$n> \max\left\{ O_p(d^3\log p), O_p\left(\kappa_1^2R^4p^2\log p\right)\right\}.$$  Appendix \ref{suppB} shows that $\kappa_1\le O_p(R^2)$. Hence,  the sufficient condition for consistency of PC-LPGM with exponentially decaying error is  $$n>\max\left\{ O_p(d^3\log p), O_p\left(R^8p^2\log p\right)\right\}.$$}

When $R$ is fixed, the condition reduces to $n>\max\left\{ O_p(d^3\log p), O_p\left(p^2\log p\right)\right\}$. However, it is worth remembering that when the maximum number of neighbours that one node is allowed to have is fixed to $d$,  a limitation is operated on the cardinalities $d$ or $d-1$ of the sets $\bold{K}$. In this situation, the condition for convergence is relaxed to $n>\max\left\{ O_p(d^3\log p), O_p\left(d^2\log d\right)\right\}=  O_p(d^3\log p)$(see Equation \eqref{totalerror} and Note \ref{cy2} for details). This condition is coherent with the sufficient condition for consistent neighbourhood selection in Ising model (see \cite{ravikumar2010}).

Our simulation results show that the algorithm perform well also when  Poisson conditional distributions with no constraints on the interaction parameters are taken as starting point for model specification. The empirical comparison shows that the algorithm outperforms its natural competitors.


	
	{\color{black}
	\acks{We thank the action editor and three anonymous reviewers  for their careful reading of the manuscript and their  many
insightful comments and suggestions. We also thank Davide Risso for generously sharing with 
		us his expertise in the area of single-cell data analysis.  This work 
		was supported  by grant BIRD172830,  University of Padova, Italy. }}
	
	

	\appendix

	\section{Proofs}\label{suppA}
	In this section, we provide proofs of Proposition \ref{pro11} and Theorem \ref{dl2}  stated in Section \ref{statistical} of the main paper. We begin by introducing  results for the case $\bold{K}=V$. Then, the same results for general case   {$\bold{K}\subset V$} are deduced.  { The rest of this section is devoted for the proof of Theorem~\ref{robust} in Section~\ref{robustness}. }
	
	Before going into details, we first prove the following Lemma, used in the proof of Theorem \ref{dl1}.
	\begin{bd}\label{Fisher}
		Assume \ref{assum2}. Then, for any $\delta>0$, we have
		{
		\begin{eqnarray*}
			\mathbb{P}_{\boldsymbol{\theta}}\left(\Lambda_{\max}\bigg(\frac{1}{n}\sum_{i=1}^{n} (\bold{X}_{V\backslash \{s\}}^{(i)})^T\bold{X}_{V\backslash \{s\}}^{(i)}\bigg)\le \lambda_{\max}+\delta\right)&\ge&1- \exp\left(-c_2\frac{\delta^2 n}{p^2}+c_3\log p\right)\\
			\mathbb{P}_{\boldsymbol{\theta}}\left(\Lambda_{\min}\left(Q_s(\boldsymbol{\theta_s})\right)\ge \lambda_{\min}-\delta\right)&\ge& 1- \exp\left(-c_2\frac{\delta^2 n}{p^2}+c_3\log p\right).
		\end{eqnarray*}
	}
	\end{bd}
	\begin{proof}
		The $(j,k)$ element of the matrix $Z^n=Q_s(\boldsymbol{\theta}_s)-I_s(\boldsymbol{\theta}_s)$ can be written as
		\begin{eqnarray*}
			Z_{jk}^n(\boldsymbol{\theta}_s)&=&\frac{1}{n}\sum_{i=1}^nD''(\langle\boldsymbol{\theta}_s,{X}^{(i)}_{V\backslash \{s\}}\rangle ){X}_{ij}X_{ik}-\mathbb{E}_{\boldsymbol{\theta}}\left(D''\big(\langle\boldsymbol{\theta}_s,{X}_{V\backslash \{s\}}\rangle \big){X}_jX_k\right)\\
			&=&\frac{1}{n}\sum_{i=1}^n Y_i-\mathbb{E}_{\boldsymbol{\theta}}\left(\frac{1}{n}\sum_{i=1}^n Y_i\right),
		\end{eqnarray*}
		where $Y_i=D''(\langle\boldsymbol{\theta}_s,{X}^{(i)}_{V\backslash \{s\}}\rangle ){X}_{ij}X_{ik},~i=1,\ldots,n$ are independent and bounded by
		$$|Y_i|\le \kappa_1 R^2.$$
		By the Azuma-Hoeffding inequality  \citep[Theorem 2 in][]{hoeffding1963probability}, for any $\epsilon>0$, we have
		$$\mathbb{P}_{\boldsymbol{\theta}}\left((Z_{ij}^n)^2\ge \epsilon^2\right)=\mathbb{P}_{\boldsymbol{\theta}}\left(|Z_{ij}^n|\ge \epsilon\right)\le 2\exp\left(-\frac{\epsilon^2n}{2\kappa_1^2R^4}\right).$$
		Moreover,
		\begin{eqnarray*}
			\Lambda_{\min}(I_s(\boldsymbol{\theta}_s))&=&\min_{\|\bold{y}\|_2=1}\bold{y}I_s(\boldsymbol{\theta}_s)\bold{y}^T\\
			&=&\min_{\|\bold{y}\|_2=1}\left\{\bold{y}Q_s(\boldsymbol{\theta}_s)\bold{y}^T+\bold{y}(I_s(\boldsymbol{\theta}_s)-Q_s(\boldsymbol{\theta}_s))\bold{y}^T\right\}\\
			&\le& \bold{y}Q_s(\boldsymbol{\theta}_s)\bold{y}^T+\bold{y}(I_s(\boldsymbol{\theta}_s)-Q_s(\boldsymbol{\theta}_s))\bold{y}^T,
		\end{eqnarray*}
		where $\bold{y}\in \mathbb{R}^{p-1}$ is an arbitrary vector with unit norm. Hence,
		\begin{equation}\label{Hessian}
		\Lambda_{\min}(Q_s(\boldsymbol{\theta}_s))\ge \Lambda_{\min}(I_s(\boldsymbol{\theta}_s))-\max_{\|\bold{y}\|_2=1}\bold{y}\big(I_s(\boldsymbol{\theta}_s)-Q_s(\boldsymbol{\theta}_s)\big)\bold{y}^T\ge \lambda_{\min}-|||I_s(\boldsymbol{\theta}_s)-Q_s(\boldsymbol{\theta}_s)|||_2.
		\end{equation}
		We now derive a bound on the spectral norm 
		$|||I_s(\boldsymbol{\theta}_s)-Q_s(\boldsymbol{\theta}_s)|||_2$. Let $\epsilon=\delta/p$, then 
		\begin{eqnarray}\label{Fisherdistance}
		\mathbb{P}_{\boldsymbol{\theta}}\left(|||I_s(\boldsymbol{\theta}_s)-Q_s(\boldsymbol{\theta}_s)|||_2\ge \delta\right)&\le&\mathbb{P}_{\boldsymbol{\theta}}\left(\bigg(\sum_{j,k\ne s}(Z_{jk}^n)^2\bigg)^{1/2}\ge \delta \right)\nonumber\\
		&\le& 2p^2\exp\left\{-\frac{\delta^2n}{2p^2\kappa_1^2R^4}\right\}\nonumber\\
		&\le& \exp\left(-c_2\frac{\delta^2 n}{p^2}+c_3\log p\right).
		\end{eqnarray}
		From Equation \eqref{Hessian} and \eqref{Fisherdistance}, we have
		$$\mathbb{P}_{\boldsymbol{\theta}}\left(\Lambda_{\min}(Q_s(\boldsymbol{\theta}_s))\ge \lambda_{\min}-\delta\right)\ge 1-\exp\left(-c_2\frac{\delta^2 n}{p^2}+c_3\log p\right).$$
		Similarly, we have 
		$$\mathbb{P}_{\boldsymbol{\theta}}\left(\Lambda_{\max}\left[\frac{1}{n}\sum_{i=1}^{n} (\bold{X}_{V\backslash \{s\}}^{(i)})^T\bold{X}_{V\backslash \{s\}}^{(i)}\bigg]\le \lambda_{\max}+\delta\right]\right)\ge 1-\exp\left(-c_2\frac{\delta^2 n}{p^2}+c_3\log p\right).$$
	\end{proof}
	
	We now introduce  results for the case $\bold{K}=V$.
	
	\noindent
\begin{md}\label{md1}
	Assume \ref{assum1}- \ref{assum2}. Then, for any $\delta>0$
	$$\mathbb{P}_{\boldsymbol{\theta}}(\|\nabla l(\boldsymbol{\theta}_s,\bold{X}_s;\bold{X}_{V\backslash \{s\}})\|_{\infty}\ge \delta)\le \exp\{-c_1n\delta^2+c_0\log p\},~\forall~ \boldsymbol{\theta}\in \boldsymbol{\Theta},$$
	when  $n\longrightarrow\infty$.
\end{md}




\begin{proof}
	A rescaled negative node conditional log-likelihood  can be written as 
	\begin{eqnarray*}
		l(\boldsymbol{\theta}_{s}, \, \mathbb{X}_{\{s\}} \, ; \mathbb{X}_{V\backslash \{s\}}) &=& -\frac{1}{n}\log \prod_{i=1}^{n}\mathbb{P}_{\boldsymbol{\theta}_{s}}(x_{is}|\bold{x}_{V\backslash \{s\}}^{(i)})\\
		&=&\frac{1}{n}\sum_{i=1}^{n}\left[-x_{is}\langle\boldsymbol{\theta}_{s},\bold{x}_{V\backslash \{s\}}^{(i)}\rangle+\log x_{is} ! +D(\langle\boldsymbol{\theta}_{s},\bold{x}_{V\backslash \{s\}}^{(i)}\rangle)\right],\nonumber
	\end{eqnarray*}
	The $t$-partial derivative of the node conditional log-likelihood $l(\boldsymbol{\theta}_s,\mathbb{X}_s;\mathbb{X}_{V\backslash \{s\}})$ is:
	\begin{eqnarray*}
		W_t=\nabla_t l(\boldsymbol{\theta}_s,\mathbb{X}_s;\mathbb{X}_{V\backslash \{s\}}) &=& \frac{1}{n}\sum_{i=1}^{n}\left[-x_{is}x_{it}+x_{it}D'(\langle\boldsymbol{\theta}_s,\bold{x}_{V\backslash \{s\}}^{(i)}\rangle )\right]
	\end{eqnarray*}
	Let $V_{is}(t)= X_{is}X_{it}-X_{it}D'(\langle\boldsymbol{\theta}_s,\bold{X}_{V\backslash \{s\}}^{(i)}\rangle )$. We have,
	\begin{eqnarray}\label{derivative}
	\mathbb{P}_{\boldsymbol{\theta}}(\|W\|_\infty \ge\delta)&=& \mathbb{P}_{\boldsymbol{\theta}}(\max_{t\in V\backslash \{s\}}|\nabla_t l(\boldsymbol{\theta}_s,X_s;\bold{X}_{V\backslash\{s\}})|\ge \delta)\nonumber\\
	&=&\mathbb{P}_{\boldsymbol{\theta}}\bigg( \max_{t\in V\backslash \{s\}}\bigg|\dfrac{1}{n}\sum_{i=1}^{n}V_{is}(t)\bigg|\ge\delta\bigg)\nonumber\\
	&\le& p \left[\mathbb{P}_{\boldsymbol{\theta}}\left(\dfrac{1}{n}\sum_{i=1}^{n}V_{is}(t)\ge\delta\right)+\mathbb{P}_{\boldsymbol{\theta}}\left(-\dfrac{1}{n}\sum_{i=1}^{n}V_{is}(t)\ge\delta\right)\right]\nonumber\\
	&\le& p\left[ \dfrac{\mathbb{E}_{\boldsymbol{\theta}}\left[\prod_{i=1}^{n}\exp\left\{hV_{is}(t)\right\}\right]}{\exp\{nh\delta\}}+\dfrac{\mathbb{E}_{\boldsymbol{\theta}}\left[\prod_{i=1}^{n}\exp\left\{-hV_{is}(t)\right\}\right]}{\exp\{nh\delta\}}\right]\nonumber\\
	&=& p \left[\dfrac{\prod_{i=1}^{n}\mathbb{E}_{\boldsymbol{\theta}}\left[\exp\left\{hV_{is}(t)\right\}\right]}{\exp\{nh\delta\}}+\dfrac{\prod_{i=1}^{n}\mathbb{E}_{\boldsymbol{\theta}}\left[\exp\left\{-hV_{is}(t)\right\}\right]}{\exp\{nh\delta\}}\right] \nonumber\\
	&=& p\bigg[\exp\bigg\{ \sum_{i=1}^{n}\log\mathbb{E}_{\boldsymbol{\theta}}\left[\exp\left\{hV_{is}(t)\right\}\right]-nh\delta\bigg\}\nonumber\\
	&&+\exp\bigg\{\sum_{i=1}^{n}\log\mathbb{E}_{\boldsymbol{\theta}}\left[\exp\left\{-hV_{is}(t)\right\}\right] -nh\delta\bigg\}\bigg],
	\end{eqnarray}
	for some $h>0$. We therefore need to compute $$\sum_{i=1}^{n}\log\mathbb{E}_{\boldsymbol{\theta}}\left[\exp\left\{hV_{is}(t)\right\}\right],$$ and $$\sum_{i=1}^{n}\log\mathbb{E}_{\boldsymbol{\theta}}\left[\exp\left\{-hV_{is}(t)\right\}\right].$$ First, we have
	\begin{eqnarray}\label{moment}
	\mathbb{E}_{\boldsymbol{\theta}_s}\left[\exp\{hV_{is}(t)\}|\bold{x}_{V\backslash \{s\}}^{(i)}\right] &=& \sum_{x_{is}=0}^{R}\exp\bigg\{h[x_{is}x_{it}-x_{it}D'(\langle\boldsymbol{\theta}_s,\bold{x}_{V\backslash \{s\}}^{(i)}\rangle )]\nonumber\\
	&&~+x_{is} \langle\boldsymbol{\theta}_s,\bold{x}_{V\backslash \{s\}}^{(i)}\rangle
	- \log x_{is}!	- D(\langle\boldsymbol{\theta}_s,\bold{x}_{V\backslash \{s\}}^{(i)}\rangle)\bigg\}\nonumber\\
	&=&\sum_{x_{is}=0}^{R}\exp\bigg\{x_{is}[hx_{it}+\langle\boldsymbol{\theta}_s,\bold{x}_{V\backslash \{s\}}^{(i)}\rangle ]- \log x_{is}!\nonumber\\
	&& ~- hx_{it}D'(\langle\boldsymbol{\theta}_s,\bold{x}_{V\backslash \{s\}}^{(i)}\rangle )- D(\langle\boldsymbol{\theta}_s,\bold{x}_{V\backslash \{s\}}^{(i)}\rangle)\bigg\}\nonumber\\
	&=& \exp\bigg\{D(hx_{it}+\langle\boldsymbol{\theta}_s,\bold{x}_{V\backslash \{s\}}^{(i)}\rangle)-D(\langle\boldsymbol{\theta}_s,\bold{x}_{V\backslash \{s\}}^{(i)}\rangle)\nonumber\\
	&&~ -hx_{it}D'(\langle\boldsymbol{\theta}_s,\bold{x}_{V\backslash \{s\}}^{(i)}\rangle) \bigg\}\nonumber\\
	&=& \exp\left\{\frac{h^2}{2}(x_{it})^2D''(vhx_{it}+\langle\boldsymbol{\theta}_s,\bold{x}_{V\backslash \{s\}}^{(i)}\rangle)\right\},\nonumber
	\end{eqnarray}
	$~\text{for some } v\in [0,1],$ where we move from line 2 to line 3 by applying $\sum_{x_{is}=0}^{R}\exp\big\{x_{is}[hx_{it}+\langle\boldsymbol{\theta}_s,\bold{x}_{V\backslash \{s\}}^{(i)}\rangle ]- \log x_{is}!-D(hx_{it}+\langle\boldsymbol{\theta}_s,\bold{x}_{V\backslash \{s\}}^{(i)}\rangle)\big\}=1$, and from line 3 to line 4 by using a Taylor expansion for function $D(.)$ at $\langle\boldsymbol{\theta}_s,\bold{x}_{V\backslash \{s\}}^{(i)}\rangle$. 
	
	Therefore,
	\begin{eqnarray}\label{pospart}
	\sum_{i=1}^{n}\log\mathbb{E}_{\boldsymbol{\theta}}\left[\exp\left\{hV_{is}(t)\right\}\right]&=&\sum_{i=1}^{n}\log\mathbb{E}_{\boldsymbol{\theta}_{V\backslash \{s\}}}\bigg[
	\mathbb{E}_{\boldsymbol{\theta}_s}\left[\exp\{hV_{is}(t)\}|\bold{x}_{V\backslash \{s\}}^{(i)}\right]\bigg]\nonumber\\
	&=&\sum_{i=1}^{n}\log\mathbb{E}_{\boldsymbol{\theta}_{V\backslash \{s\}}}\bigg[\exp\left\{\frac{h^2}{2}(X_{it})^2D''(vhX_{it}+\langle\boldsymbol{\theta}_s,\bold{X}_{V\backslash \{s\}}^{(i)}\rangle)\right\}\bigg]\nonumber\\
	&\le& n\frac{h^2}{2}R^2\kappa_1,\nonumber\\
	\end{eqnarray}
	where $D''(vhX_{it}+\langle\boldsymbol{\theta}_s,\bold{X}_{V\backslash \{s\}}^{(i)}\rangle)<\kappa_1,~\forall~ \boldsymbol{\theta}_s\in \Omega(\boldsymbol{\Theta})$ (\text{since }$ D''(.)$ \text{ is a continuous function, and }$\Omega(\boldsymbol{\Theta})$ \text{ is bounded},  see Appendix \ref{suppB} for details). Similarly,
	\begin{eqnarray}\label{moment}
	\mathbb{E}_{\boldsymbol{\theta}_s}\left[\exp\{-hV_{is}(t)\}|\bold{x}_{V\backslash \{s\}}^{(i)}\right] &=& 
	\exp\left\{\frac{h^2}{2}(x_{it})^2D''(-vhx_{it}+\langle\boldsymbol{\theta}_s,\bold{x}_{V\backslash \{s\}}^{(i)}\rangle)\right\},\nonumber
	\end{eqnarray}
	Therefore,
	\begin{eqnarray}\label{negpart}
	\sum_{i=1}^{n}\log\mathbb{E}_{\boldsymbol{\theta}}\left[\exp\left\{-hV_{is}(t)\right\}\right]&=&\sum_{i=1}^{n}\log\mathbb{E}_{\boldsymbol{\theta}_{V\backslash \{s\}}}\bigg[
	\mathbb{E}_{\boldsymbol{\theta}_s}\left[\exp\{-hV_{is}(t)\}|\bold{x}_{V\backslash \{s\}}^{(i)}\right]\bigg]\nonumber\\
	&=&\sum_{i=1}^{n}\log\mathbb{E}_{\boldsymbol{\theta}_{V\backslash \{s\}}}\bigg[\exp\left\{\frac{h^2}{2}(X_{it})^2D''(-vhX_{it}+\langle\boldsymbol{\theta}_s,\bold{X}_{V\backslash \{s\}}^{(i)}\rangle)\right\}\bigg]\nonumber\\
	&\le&n\frac{h^2}{2}R^2\kappa_1.
	\end{eqnarray}
	Let $h=\dfrac{\delta}{R^2\kappa_1}$, from \eqref{derivative}--\eqref{negpart}, we have
	\begin{eqnarray*} 	
		\mathbb{P}_{\boldsymbol{\theta}}(\|W\|_\infty \ge\delta)&\le& p\bigg[\exp\bigg\{ n\frac{h^2}{2}R^2\kappa_1-nh\delta\bigg\}
		+\exp\bigg\{n\frac{h^2}{2}R^2\kappa_1 -nh\delta\bigg\}\bigg]\\
		&=&  2p\bigg[\exp\bigg\{ \frac{-n\delta^2}{2R^2\kappa_1}\bigg\}\bigg]\\
		&\le&\exp\big\{ -c_1n\delta^2+c_0\log p\big\},
\end{eqnarray*}
	
\end{proof}

	\begin{dl}\label{dl1}
		Assume \ref{assum1}- \ref{assum2}. Then, there exists a non-negative decreasing sequence $\epsilon_n\rightarrow 0$, such that
			$$\mathbb{P}_{\boldsymbol{\theta}}(\|\hat{\boldsymbol{\theta}}_{V\backslash \{s\}}-\boldsymbol{\theta}_{V\backslash \{s\}}\|_2\le \epsilon_n)\ge 1-\exp\big\{ -c_1 n\epsilon_n^2+c_0 \log p\big\}-\exp\left\{ -c_2 \dfrac{n}{p^2}+c_3\log p\right\},~\forall~ \boldsymbol{\theta}\in \boldsymbol{\Theta},$$
			when $n\rightarrow \infty$.
	\end{dl}

	\begin{proof}
		For a fixed design $\mathbb{X}$,  define $G: \mathbb{R}^{p-1}\longrightarrow \mathbb{R}$ as
		$$G({\bold{u}}, \mathbb{X}_s;\mathbb{X}_{V\backslash\{s\}})=l(\boldsymbol{\theta}_s+{\bold{u}},\mathbb{X}_s;\mathbb{X}_{V\backslash \{s\}})- l(\boldsymbol{\theta}_s,\mathbb{X}_s;\mathbb{X}_{V\backslash \{s\}}).$$
		Then, $G(0, \mathbb{X}_s;\mathbb{X}_{V\backslash\{s\}})=0$. Moreover, let $\hat{\bold{u}}= \hat{\boldsymbol{\theta}}_s-\boldsymbol{\theta}_s$, we have $G(\hat{\bold{u}}, \mathbb{X}_s;\mathbb{X}_{V\backslash\{s\}})\le 0$.
		
		Given a value $\epsilon>0$, if $G({\bold{u}}, \mathbb{X}_s;\mathbb{X}_{V\backslash\{s\}})>0,~\forall \bold{u}\in \mathbb{R}^{p-1}$ such that $\|\bold{u}\|_2=\epsilon$, then $\|\hat{\bold{u}}\|_2\le \epsilon$, since $G(., \mathbb{X}_s;\mathbb{X}_{V\backslash\{s\}})$ is a convex function. Therefore,
		$$\mathbb{P}_{\boldsymbol{\theta}}\left(\|\hat{\boldsymbol{\theta}}_s-\boldsymbol{\theta}_s\|_2\le \epsilon\right)\ge\mathbb{P}_{\boldsymbol{\theta}}\left(G({\bold{u}}, \bold{X}_s;\bold{X}_{V\backslash\{s\}})>0),~\forall \bold{u}\in \mathbb{R}^{p-1} \text{ such that } \|\bold{u}\|_2=\epsilon\right).$$
		A Taylor expansion of the rescaled negative node conditional log-likelihood at $\boldsymbol{\theta}_s$ yields
		\begin{eqnarray*}\label{funcg}
			G(\bold{u},\bold{X}_s;\bold{X}_{V\backslash\{s\}}) &=& l(\boldsymbol{\theta}_s+\bold{u},\bold{X}_s;\bold{X}_{V\backslash\{s\}})- l(\boldsymbol{\theta}_s,\bold{X}_s;\bold{X}_{V\backslash\{s\}})\\
			&=& \nabla l(\boldsymbol{\theta}_s,\bold{X}_s;\bold{X}_{V\backslash\{s\}})) \bold{u}^T+\frac{1}{2}\bold{u}[\nabla^2(l(\boldsymbol{\theta}_s+v\bold{u},\bold{X}_s;\bold{X}_{V\backslash\{s\}})]\bold{u}^T\nonumber,
		\end{eqnarray*}
		for some $v\in [0,1]$. Let 
		\begin{eqnarray*}
			q&=& \Lambda_{\min}(\nabla^2(l(\boldsymbol{\theta}_s+v\bold{u},\bold{X}_s;\bold{X}_{V\backslash\{s\}})))\\
			&\ge& \min_{v\in[0,1]}\Lambda_{\min}(\nabla^2(l(\boldsymbol{\theta}_s+v\bold{u},\bold{X}_s;\bold{X}_{V\backslash\{s\}})))\\
			&=& \min_{v\in[0,1]}\Lambda_{\min} \left[\frac{1}{n}\sum_{i=1}^{n} D''(\langle\boldsymbol{\theta}_s+v\bold{u},\bold{X}^{(i)}_{V\backslash \{s\}}\rangle )(\bold{X}_{V\backslash \{s\}}^{(i)})^T\bold{X}_{V\backslash \{s\}}^{(i)}\right].
		\end{eqnarray*}
		By using Taylor expansion for $D''(\langle\boldsymbol{\theta}_s+v\bold{u},\bold{X}^{(i)}_{V\backslash \{s\}}\rangle )$ at $\langle\boldsymbol{\theta}_s,\bold{X}^{(i)}_{V\backslash \{s\}}\rangle$, we have
		\begin{eqnarray*}
			&&\frac{1}{n}\sum_{i=1}^{n} D''(\langle\boldsymbol{\theta}_s+v\bold{u},\bold{X}^{(i)}_{V\backslash \{s\}}\rangle )(\bold{X}_{V\backslash \{s\}}^{(i)})^T\bold{X}_{V\backslash \{s\}}^{(i)})\\ 
			&&=\quad\frac{1}{n}\sum_{i=1}^{n} D''(\langle\boldsymbol{\theta}_s,\bold{X}^{(i)}_{V\backslash \{s\}}\rangle )(\bold{X}_{V\backslash \{s\}}^{(i)})^T\bold{X}_{V\backslash \{s\}}^{(i)}+\\
			&&\qquad
			\frac{1}{n}\sum_{i=1}^{n} D'''(\langle\boldsymbol{\theta}_s+v'\bold{u},\bold{X}^{(i)}_{V\backslash \{s\}}\rangle )[v\bold{u}(\bold{X}_{V\backslash \{s\}}^{(i)})^T][(\bold{X}_{V\backslash \{s\}}^{(i)})^T\bold{X}_{V\backslash \{s\}}^{(i)}],
		\end{eqnarray*}
		for some $v'\in [0,1]$. Fixed $\delta=\dfrac{\lambda_{\min}}{8}$  in  Lemma \ref{Fisher}. We have 
		\begin{eqnarray*}
			q&\ge& \Lambda_{\min}\left[\frac{1}{n}\sum_{i=1}^{n} D''(\langle\boldsymbol{\theta}_s,\bold{X}^{(i)}_{V\backslash \{s\}}\rangle )(\bold{X}_{V\backslash \{s\}}^{(i)})^T\bold{X}_{V\backslash \{s\}}^{(i)}\right]\\
			&&-\max_{v'\in[0,1]}\Lambda_{\max}\bigg[\frac{1}{n}\sum_{i=1}^{n} \big|D'''(\langle\boldsymbol{\theta}_s+v'\bold{u},\bold{X}^{(i)}_{V\backslash \{s\}}\rangle )\big|[\bold{u}(\bold{X}_{V\backslash \{s\}}^{(i)})^T]\big[(\bold{X}_{V\backslash \{s\}}^{(i)})^T\bold{X}_{V\backslash \{s\}}^{(i)}\big]\bigg]\\
			&\ge& \lambda_{\min}-\delta-\max_{v'\in[0,1]}\Lambda_{\max}\bigg[\frac{1}{n}\sum_{i=1}^{n} \big|D'''(\langle\boldsymbol{\theta}_s+v'\bold{u},\bold{X}^{(i)}_{V\backslash \{s\}}\rangle )\big|[\bold{u}(\bold{X}_{V\backslash \{s\}}^{(i)})^T]\big[(\bold{X}_{V\backslash \{s\}}^{(i)})^T\bold{X}_{V\backslash \{s\}}^{(i)}\big]\bigg]\\
			&\ge& \lambda_{\min}-\delta-\max_{v'\in[0,1]}\big|D'''(\langle\boldsymbol{\theta}_{V\backslash \{s\}}+v'\bold{u},\bold{X}^{(i)}_{V\backslash \{s\}}\rangle )\big|\big|\bold{u}(\bold{X}_{V\backslash \{s\}}^{(i)})^T\big|\Lambda_{\max}\bigg[\frac{1}{n}\sum_{i=1}^{n} (\bold{X}_{V\backslash \{s\}}^{(i)})^T\bold{X}_{V\backslash \{s\}}^{(i)}\bigg]\\
			&\ge& \lambda_{\min}-2\delta- \kappa_2 R\sqrt{p}\|\bold{u}\|_2\lambda_{\max}\\
			&=& \lambda_{\min}-2\delta- \kappa_2\sqrt{p}R\epsilon\lambda_{\max} \\
			&>& \dfrac{\lambda_{\min}}{2},\quad \text{provided that } \epsilon< \dfrac{\lambda_{\min}}{4\sqrt{p}\lambda_{\max}\kappa_2R},
		\end{eqnarray*}
	with probability at least $1-\exp\left\{ -c_2 \dfrac{n}{p^2}+c_3\log p\right\}$, where $\big|D'''(\langle\boldsymbol{\theta}_s+v'\bold{u},\bold{X}^{(i)}_{V\backslash \{s\}}\rangle )\big|<\kappa_2,~\forall ~\boldsymbol{\theta}_s\in \boldsymbol{\Theta}$ (\text{since }$ D'''(.)$ \text{ is a continuous function, and }$\boldsymbol{\Theta}$ \text{ is bounded},  see Appendix \ref{suppB} for details).
		
		Let  $\delta=\dfrac{\lambda_{\min}}{4}\epsilon$ in  Proposition \ref{md1}. Then, from Proposition \ref{md1}, we have
		$$\nabla_t l(\boldsymbol{\theta}_s,\bold{X}_s;\bold{X}_{V\backslash \{s\}}))\ge -\dfrac{\lambda_{\min}}{4}\epsilon,$$
		with probability at least $1-\exp\big\{ -c_1 n\epsilon^2\big\}$. 
		Combining with the inequality of $q$, we have
		\begin{eqnarray*}
			G(\bold{u},\bold{X}_s;\bold{X}_{V\backslash\{s\}})&=& \nabla l(\boldsymbol{\theta}_s,\bold{X}_s;\bold{X}_{V\backslash\{s\}})) \bold{u}^T+\frac{1}{2}\bold{u}[\nabla^2(l(\boldsymbol{\theta}_s+v\bold{u},_s,\bold{X}_s;\bold{X}_{V\backslash\{s\}}))]\bold{u}^T\nonumber\\
			&>& -\dfrac{\lambda_{\min}}{4}\epsilon^2+\dfrac{\lambda_{\min}}{4}\epsilon^2=0,
		\end{eqnarray*}
		with probability at least $1-\exp\big\{ -c_1 n\epsilon^2+c_0 \log p\big\}-\exp\left\{ -c_2 \dfrac{n}{p^2}+c_3\log p\right\}$, $\text{provided that } \epsilon< \dfrac{\lambda_{\min}}{4\sqrt{p}\lambda_{\max}\kappa_2R}$. It means that $\|\hat{\bold{u}}\|_2<\epsilon$.\\
		When $n\rightarrow\infty$ we can choose a non-negative decreasing sequence $\epsilon_n$ such that $\epsilon_n<\dfrac{\lambda_{\min}}{4\sqrt{p}\lambda_{\max}\kappa_2R}$, then
		$$\mathbb{P}_{\boldsymbol{\theta}}(\|\hat{\boldsymbol{\theta}}_{V\backslash \{s\}}-\boldsymbol{\theta}_{V\backslash \{s\}}\|_2\le \epsilon_n)\ge 1-\exp\big\{ -c_1 n\epsilon_n^2+c_0 \log p\big\}-\exp\left\{ -c_2 \dfrac{n}{p^2}+c_3\log p\right\},$$
		when $n\rightarrow \infty$.
	\end{proof}
	
	Results for $\bold{K}\subset V$ are derived as following.\\\\
	\noindent
	{\bf Proposition 4.3}
	Assume \ref{assum1}- \ref{assum2} and let $\bold{K}\subset V$. Then, for all $s\in \bold{K}$ and any $\delta>0$
	$$\mathbb{P}_{\boldsymbol{\theta}}(\|\nabla l(\boldsymbol{\theta}_{s|\bold{K}},\bold{X}_{\{s\}}; {\bold{X}_{\bold{K}\backslash\{s\}}})\|_{\infty}\ge \delta)\le \exp\big\{ -c_1 n\delta^2+c_0 \log d\big\},$$
	$~\forall~ \boldsymbol{\theta_{s|\bold{K}}}\in \boldsymbol{\Theta},$ when  $n\rightarrow\infty$.\\

	\begin{proof}
			The proof of Proposition \ref{pro11} follows the lines of Proposition \ref{md1}. We note that the set of explanatory variables $\bold{X}_{\bold{K}\backslash\{s\}}$ in the generalized linear model $X_s$ given $\bold{X}_{\bold{K}\backslash\{s\}}$ does not include variables $X_t$, with $t\in V\backslash\bold{K}$. Suppose we zero-pad the true parameter $\boldsymbol{\theta}_{s|\bold{K}}\in\mathbb{R}^{|\bold{K}|-1}$ to include zero weights over $V\backslash\bold{K}$, then the resulting parameter would lie in $\mathbb{R}^{|p-1|}$. 

		
			Moreover, when the maximum number of neighbours that one node is allowed to have is fixed, a control is operated on the cardinality of the set $\bold{K}$,  $|\bold{K}|\le m+2\le d+1$. In this case, parameters $\theta_{st|\bold{K}}$ are estimated from models that are restricted on subsets of variables with their cardinalities less than or equal to $d$. Therefore, $p$ in Proposition \ref{md1} is replaced by $d$. In detail, for all $s\in \bold{K}$ and any $\delta>0$
			$$\mathbb{P}_{\boldsymbol{\theta}}(\|\nabla l(\boldsymbol{\theta}_{s|\bold{K}},\bold{X}_{\{s\}}; {\bold{X}_{\bold{K}\backslash\{s\}}})\|_{\infty}\ge \delta)\le \exp\{-c_1n\delta^2+c_0\log d\},$$
			$~\forall~ \boldsymbol{\theta_{s|\bold{K}}}\in \boldsymbol{\Theta}.$ 
	\end{proof}
	
	We take the same way as in the proof of Theorem \ref{dl1} to prove Theorem \ref{dl2}.
	
	\noindent
	{\bf Theorem 4.4}
  Assume \ref{assum1}- \ref{assum2} and let $\bold{K}\subset V$. Then, there exists a non-negative decreasing sequence $\epsilon_n\rightarrow 0$, such that
		$$\mathbb{P}_{\boldsymbol{\theta}}(\|\hat{\boldsymbol{\theta}}_{s|\bold{K}}-\boldsymbol{\theta}_{s|\bold{K}}\|_2\le \epsilon_n)\ge 1-\exp\{-c_1n\epsilon_n^2+c_0\log d\}-\exp\bigg\{-c_2\dfrac{n}{d^2}+c_3\log d\bigg\},$$
		$~\forall~ s\in \bold{K}, \boldsymbol{\theta}\in \boldsymbol{\Theta},$ when $n\rightarrow \infty$.		
	
	\begin{proof}
		Let $\hat{\bold{u}}= \hat{\boldsymbol{\theta}}_{s|\bold{K}}-\boldsymbol{\theta}_{s|\bold{K}}$, and define $G: \mathbb{R}^{|\bold{K}|-1}\longrightarrow \mathbb{R}$ as
		$$G(\hat{\bold{u}},\mathbb{X}_{\{s\}}; {\mathbb{X}_{\bold{K}\backslash\{s\}}})=l(\boldsymbol{\theta}_{s|\bold{K}}+\hat{\bold{u}},\mathbb{X}_{\{s\}}; {\mathbb{X}_{\bold{K}\backslash\{s\}}})- l(\boldsymbol{\theta}_{s|\bold{K}},\mathbb{X}_{\{s\}}; {\mathbb{X}_{\bold{K}\backslash\{s\}}}).$$
		Similar to Theorem \ref{dl1}, we have
		$$\mathbb{P}_{\boldsymbol{\theta}}(\|\hat{\boldsymbol{\theta}}_{s|\bold{K}}-\boldsymbol{\theta}_{s|\bold{K}}\|_2\le \epsilon)\ge\mathbb{P}_{\boldsymbol{\theta}}\left(G({\bold{u}}, \bold{X}_s; {\bold{X}_{\bold{K}\backslash\{s\}}})>0),~\forall \bold{u}\in \mathbb{R}^{|\bold{K}|} \text{ such that } \|\bold{u}\|_2=\epsilon\right).$$
		Recall the conditional rescaled negative log-likelihood function:
		\begin{equation*}
		l(\boldsymbol{\theta}_{s|\bold{K}},\bold{X}_{\{s\}}; {\bold{X}_{\bold{K}\backslash\{s\}}}) =\frac{1}{n}\sum_{i=1}^{n}\left[-X_{is}\langle\boldsymbol{\theta}_{s|\bold{K}}, {\bold{X}^{(i)}_{\bold{K}\backslash\{s\}}}\rangle +D(\langle\boldsymbol{\theta}_{s|\bold{K}}, {\bold{X}^{(i)}_{\bold{K}\backslash\{s\}}}\rangle)\right].\nonumber
		\end{equation*}
		
		By its Taylor expansion  at $\boldsymbol{\theta}_{s|\bold{K}}$, we have
		\begin{eqnarray*}\label{funcg}
			G(\bold{u}) &=& l(\boldsymbol{\theta}_{s|\bold{K}}+\bold{u},\bold{X}_{\{s\}}; {\bold{X}_{\bold{K}\backslash\{s\}}})- l(\boldsymbol{\theta}_{s|\bold{K}},\bold{X}_{\{s\}}; {\bold{X}_{\bold{K}\backslash\{s\}}})\\
			&=& \nabla l(\boldsymbol{\theta}_{s|\bold{K}},\bold{X}_{\{s\}}; {\bold{X}_{\bold{K}\backslash\{s\}}}) \bold{u}^T+\frac{1}{2}\bold{u}[\nabla^2(l(\boldsymbol{\theta}_{s|\bold{K}}+v\bold{u},\bold{X}_{\{s\}}; {\bold{X}_{\bold{K}\backslash\{s\}}}))]\bold{u}^T.\nonumber
		\end{eqnarray*}
		Let 
		\begin{eqnarray*}
			q&=& \Lambda_{\min}(\nabla^2(l(\boldsymbol{\theta}_{s|\bold{K}}+v\bold{u},\bold{X}_{\{s\}}; {\bold{X}_{\bold{K}\backslash\{s\}}})))\\
			&\ge& \min_{v\in[0,1]}\Lambda_{\min}(\nabla^2(l(\boldsymbol{\theta}_{s|\bold{K}}+v\bold{u},\bold{X}_{\{s\}}; {\bold{X}_{\bold{K}\backslash\{s\}}})))\\
			&=& \min_{v\in[0,1]}\Lambda_{\min} \left[\frac{1}{n}\sum_{i=1}^{n} D''(\langle\boldsymbol{\theta}_{s|\bold{K}}+v\bold{u}, {\bold{X}^{(i)}_{\bold{K}\backslash\{s\}}}\rangle )\big( {\bold{X}^{(i)}_{\bold{K}\backslash\{s\}}}\big)^T {\bold{X}^{(i)}_{\bold{K}\backslash\{s\}}}\right].\\
		\end{eqnarray*}
		By using Taylor expansion of $D''(\langle\boldsymbol{\theta}_{s|\bold{K}}+v\bold{u}, {\bold{X}^{(i)}_{\bold{K}\backslash\{s\}}}\rangle )$ at $\langle\boldsymbol{\theta}_{s|\bold{K}}, {\bold{X}^{(i)}_{\bold{K}\backslash\{s\}}}\rangle$, we have
		\begin{eqnarray*}
			&\dfrac{1}{n}&\sum_{i=1}^{n} D''\big(\langle\boldsymbol{\theta}_{s|\bold{K}}+v\bold{u}, {\bold{X}^{(i)}_{\bold{K}\backslash\{s\}}}\rangle \big)\big( {\bold{X}^{(i)}_{\bold{K}\backslash\{s\}}}\big)^T {\bold{X}^{(i)}_{\bold{K}\backslash\{s\}}}\\
			&=&\frac{1}{n}\sum_{i=1}^{n} D''\big(\langle\boldsymbol{\theta}_{s|\bold{K}}, {\bold{X}^{(i)}_{\bold{K}\backslash\{s\}}}\rangle \big)\big( {\bold{X}^{(i)}_{\bold{K}\backslash\{s\}}}\big)^T {\bold{X}^{(i)}_{\bold{K}\backslash\{s\}}}\\
			&&+\frac{1}{n}\sum_{i=1}^{n} D'''\big(\langle\boldsymbol{\theta}_{s|\bold{K}}+v'\bold{u}, {\bold{X}^{(i)}_{\bold{K}\backslash\{s\}}}\rangle \big)\big[v\bold{u}\big( {\bold{X}^{(i)}_{\bold{K}\backslash\{s\}}}\big)^T\big] \big[\big( {\bold{X}^{(i)}_{\bold{K}\backslash\{s\}}}\big)^T {\bold{X}^{(i)}_{\bold{K}\backslash\{s\}}}\big]	.	
		\end{eqnarray*}
		Hence, 
		\begin{eqnarray*}
			q&\ge& \Lambda_{\min}\left[\frac{1}{n}\sum_{i=1}^{n} D''\big(\langle\boldsymbol{\theta}_{s|\bold{K}}, {\bold{X}^{(i)}_{\bold{K}\backslash\{s\}}}\rangle \big)\big( {\bold{X}^{(i)}_{\bold{K}\backslash\{s\}}}\big)^T {\bold{X}^{(i)}_{\bold{K}\backslash\{s\}}}\right]\\
			&&-\max_{v'\in[0,1]}\Lambda_{\max}\left[\frac{1}{n}\sum_{i=1}^{n} \bigg|\frac{1}{n}\sum_{i=1}^{n} D'''\big(\langle\boldsymbol{\theta}_{s|\bold{K}}+v'\bold{u}, {\bold{X}^{(i)}_{\bold{K}\backslash\{s\}}}\rangle \big)\bigg|\big[v\bold{u}\big( {\bold{X}^{(i)}_{\bold{K}\backslash\{s\}}}\big)^T\big] \big( {\bold{X}^{(i)}_{\bold{K}\backslash\{s\}}}\big)^T {\bold{X}^{(i)}_{\bold{K}\backslash\{s\}}}\right]\\
			&\ge& \lambda_{\min}-\delta-\max_{v'\in[0,1]}\Lambda_{\max}\left[\frac{1}{n}\sum_{i=1}^{n} \big|D'''\big(\langle\boldsymbol{\theta}_{s|\bold{K}}+v'\bold{u}, {\bold{X}^{(i)}_{\bold{K}\backslash\{s\}}}\rangle \big)\big|\big[v\bold{u}\big( {\bold{X}^{(i)}_{\bold{K}\backslash\{s\}}}\big)^T\big] \big( {\bold{X}^{(i)}_{\bold{K}\backslash\{s\}}}\big)^T {\bold{X}^{(i)}_{\bold{K}\backslash\{s\}}}\right]\\
			&\ge&\lambda_{\min}-2\delta- \kappa_2\sqrt{p}R\epsilon\lambda_{\max}
		\end{eqnarray*}
		The second and  third inequality are due to  well-known results on eigenvalue inequalities for a matrix and its submatrix \citep[see, for example, ][]{johnson1981eigenvalue}. 
		Here, $$Q_{s|\bold{K}}(\boldsymbol{\theta}_{s|\bold{K}})=\frac{1}{n}\sum_{i=1}^{n} D''\left(\langle\boldsymbol{\theta}_{s|\bold{K}}, {\bold{X}^{(i)}_{\bold{K}\backslash\{s\}}}\rangle \right)\left( {\bold{X}^{(i)}_{\bold{K}\backslash\{s\}}}\right)^T {\bold{X}^{(i)}_{\bold{K}\backslash\{s\}}}$$
		is a sub-matrix of the Hessian matrix $Q_s(\boldsymbol{\theta}_s)$. Hence, 
		$$\Lambda_{\min}\left[\frac{1}{n}\sum_{i=1}^{n} D''\left(\langle\boldsymbol{\theta}_{s|\bold{K}}, {\bold{X}^{(i)}_{\bold{K}\backslash\{s\}}}\rangle \right)\left( {\bold{X}^{(i)}_{\bold{K}\backslash\{s\}}}\right)^T {\bold{X}^{(i)}_{\bold{K}\backslash\{s\}}}\right]\ge \Lambda_{\min}(Q_s(\boldsymbol{\theta}_s))\ge \lambda_{\min}-\delta.$$
		Similarly, for the matrix $\left( {\bold{X}^{(i)}_{\bold{K}\backslash\{s\}}}\right)^T {\bold{X}^{(i)}_{\bold{K}\backslash\{s\}}}$, we have
		\begin{eqnarray*}
			&&\max_{v'\in[0,1]}\Lambda_{\max}\bigg[\frac{1}{n}\sum_{i=1}^{n} \bigg|D'''\left(\langle\boldsymbol{\theta}_{s|\bold{K}}+v'\bold{u}, {\bold{X}^{(i)}_{\bold{K}\backslash\{s\}}}\rangle \right)\bigg|\left[v\bold{u}\left( {\bold{X}^{(i)}_{\bold{K}\backslash\{s\}}}\right)^T\right] \left( {\bold{X}^{(i)}_{\bold{K}\backslash\{s\}}}\right)^T {\bold{X}^{(i)}_{\bold{K}\backslash\{s\}}}\\
			&&\le \kappa_2\sqrt{p}R\epsilon\lambda_{\max}+\delta.
		\end{eqnarray*}
		Then, by performing the same analysis as in the proof of Theorem \ref{dl1} and Proposition \ref{md1}, we get the result.
	\end{proof}
	{
		\begin{cy}\label{cy2}
			In the proof of Theorem \ref{dl2}, we only require  the uniform convergence of a submatrix (restricted on $K$), $Q_{s|\bold{K}}(\boldsymbol{\theta}_{s|\bold{K}})$, of the sample Fisher information matrix $Q_{s}(\boldsymbol{\theta}_{s})$. Therefore, when the maximum neighbourhood size  is known,   $|\bold{K}|\le m+2\le d+1$, we have  convergence provided that $n>O_p\left(\kappa_1R^4d^2\log d\right)$. In detail, let $I_{s|\bold{K}}(\boldsymbol{\theta}_{s|\bold{K}})$ be the submatrix of $I_s(\boldsymbol{\theta}_s)$ indexed in $\bold{K}$, Equation \eqref{Fisherdistance} becomes

			\begin{eqnarray*}
	\mathbb{P}_{\boldsymbol{\theta}}\left(|||I_{s|\bold{K}}(\boldsymbol{\theta}_{s|\bold{K}})-Q_{s|\bold{K}}(\boldsymbol{\theta}_{s|\bold{K}})|||_2\ge \delta\right)&\le&\mathbb{P}_{\boldsymbol{\theta}}\left(\bigg(\sum_{j,k\in \bold{K}\backslash\{s\}}(Z_{jk}^n)^2\bigg)^{1/2}\ge \delta \right)\nonumber\\
	&\le& 2m^2\exp\left\{-\frac{\delta^2n}{2m^2\kappa_1^2R^4}\right\}\nonumber\\
	&\le& \exp\bigg\{-c_2\frac{n\delta^2}{d^2}+c_3\log d\bigg\}.
\end{eqnarray*}
	\end{cy}

{ 
\noindent
{\bf Theorem 5.1 }	
Assume that the log-likelihood function of models \eqref{Poison model} and \eqref{TPoisson model} have unique optimal solution on $\boldsymbol{\Theta}$. Then, $\hat{\boldsymbol{\theta}}_{ns}^{TP}$ converges to the true parameter $\boldsymbol{\theta}_s^*$ when $n$ tends to infinity provided that  $|\Lambda_{min}[Q_s^P({\boldsymbol{\theta}}_{s})]|=|\Lambda_{min}[\nabla^2\ell_{n}^{P}({\boldsymbol{\theta}}_{s},\mathbb{X}_s;\mathbb{X}_{V\backslash\{s\}})]|>\lambda_{min}>0,$ for all ${\boldsymbol{\theta}}_{s}\in \boldsymbol{\Theta}.$ \\

\begin{proof}
	We prove Theorem \ref{robust} by contradiction. Indeed, assume that $\hat{\boldsymbol{\theta}}_{ns}^{TP}$ does not converge to $\boldsymbol{\theta}_s^*$. Then, $\exists ~ \epsilon_0$ and $\forall n$ there exists $n_0>n$ such that $\|\hat{\boldsymbol{\theta}}_{n_0s}^{TP}-\boldsymbol{\theta}_s^*\|_2\ge \epsilon_0.$ 
	Moreover, $ \hat{\boldsymbol{\theta}}^P_{ns}$ converges to $\boldsymbol{\theta}_{s}^*$, then, for $\epsilon_0>0$, there exist $n_1>0$ such that $\forall n\ge n_1$,
	$$\|\hat{\boldsymbol{\theta}}^P_{ns}-\boldsymbol{\theta}_{s}^*\|_2\le \frac{1}{2}\epsilon_0.$$ 	Fix $n_0>n_1$ then,
	using Taylor expansion for $\ell_{n_0}^{P}(\hat{\boldsymbol{\theta}}_{n_0s}^{TP},\mathbb{X}_s;\mathbb{X}_{V\backslash\{s\}})$ at ${\boldsymbol{\theta}}^P_{n_0s}$, we have
	\begin{eqnarray*}
		\ell^{P}_{n_0}(\hat{\boldsymbol{\theta}}_{n_0s}^{TP},\mathbb{X}_s;\mathbb{X}_{V\backslash\{s\}})\!\!\!&=&\!\! \ell_{n_0}^{P}(\hat{\boldsymbol{\theta}}^P_{n_0s},\mathbb{X}_s;\mathbb{X}_{V\backslash\{s\}}) + \nabla\ell_{n_0}^{P}(\hat{\boldsymbol{\theta}}^P_{n_0s},\mathbb{X}_s;\mathbb{X}_{V\backslash\{s\}})(\hat{\boldsymbol{\theta}}_{n_0s}^{TP}-\hat{\boldsymbol{\theta}}^{P}_{n_0s})^T\\ &&~ +\frac{1}{2}(\hat{\boldsymbol{\theta}}_{n_0s}^{TP}-\hat{\boldsymbol{\theta}}^{P}_{n_0s})\nabla^2\ell_{n_0}^{P}(\boldsymbol{\theta}'_{s},\mathbb{X}_s;\mathbb{X}_{V\backslash\{s\}})(\hat{\boldsymbol{\theta}}_{n_0s}^{TP}-\hat{\boldsymbol{\theta}}^{P}_{n_0s})^T\\
		\!\!\!&=&\!\!\! \ell_{n_0}^{P}(\hat{\boldsymbol{\theta}}^P_{n_0s},\mathbb{X}_s;\mathbb{X}_{V\backslash\{s\}}) +\frac{1}{2}(\hat{\boldsymbol{\theta}}_{n_0s}^{TP}-\hat{\boldsymbol{\theta}}^{P}_{n_0s})\nabla^2\ell_{n_0}^{P}({\boldsymbol{\theta}}'_{s},\mathbb{X}_s;\mathbb{X}_{V\backslash\{s\}})(\hat{\boldsymbol{\theta}}_{n_0s}^{TP}-\hat{\boldsymbol{\theta}}^{P}_{n_0s})^T
	\end{eqnarray*}
	where $\boldsymbol{\theta}'_{s}= \hat{\boldsymbol{\theta}}_{n_0s}^{TP}+v(\hat{\boldsymbol{\theta}}^{P}_{n_0s}-\hat{\boldsymbol{\theta}}_{n_0s}^{TP})$, for some $v\in[0,1].$ Hence,
	\begin{eqnarray}\label{varianceell}
	&&\big\|\ell^{P}_{n_0}(\hat{\boldsymbol{\theta}}_{n_0s}^{TP},\mathbb{X}_s;\mathbb{X}_{V\backslash\{s\}})-\ell_{n_0}^{P}(\hat{\boldsymbol{\theta}}^P_{n_0s},\mathbb{X}_s;\mathbb{X}_{V\backslash\{s\}})\big\|_2\\
	&&~= \frac{1}{2}\big\|(\hat{\boldsymbol{\theta}}_{n_0s}^{TP}-\hat{\boldsymbol{\theta}}^{P}_{n_0s})\nabla^2\ell_{n_0}^{P}({\boldsymbol{\theta}}'_{s},\mathbb{X}_s;\mathbb{X}_{V\backslash\{s\}})(\hat{\boldsymbol{\theta}}_{n_0s}^{TP}-\hat{\boldsymbol{\theta}}^{P}_{n_0s})^T\big\|_2\nonumber\\
	&&~\ge\frac{1}{2}\big|\Lambda_{min}\left[\nabla^2\ell_{n_0}^{P}({\boldsymbol{\theta}}'_{s},\mathbb{X}_s;\mathbb{X}_{V\backslash\{s\}})\right]\big|\big\|(\hat{\boldsymbol{\theta}}_{n_0s}^{TP}-\hat{\boldsymbol{\theta}}^{P}_{n_0s})\big\|_2^2\nonumber\\
	&&~= \frac{1}{8}\lambda_{min}(\|\hat{\boldsymbol{\theta}}_{n_0s}^{TP}-\boldsymbol{\theta}_0\|_2-\|\hat{\boldsymbol{\theta}}^P_{n_0s}-\boldsymbol{\theta}_{s}^*\|_2)^2\nonumber\\
	&&~\ge\frac{1}{8}\lambda_{min}\epsilon^2_0.\nonumber
	\end{eqnarray}
	Choose $\epsilon_1=\dfrac{1}{16}\lambda_{min}\epsilon^2_0$ in Equation \eqref{Rconvergence}, then
	\begin{equation}\label{Rconvergence1}
	\|\ell^{TP}_{n_0}(\hat{\boldsymbol{\theta}}_{n_0s}^{TP},\mathbb{X}_s;\mathbb{X}_{V\backslash\{s\}})-\ell_{n_0}^{P}(\hat{\boldsymbol{\theta}}_{n_0s}^{TP},\mathbb{X}_s;\mathbb{X}_{V\backslash\{s\}})\|_2< \frac{1}{16}\lambda_{min}\epsilon^2_0.
	\end{equation}
	From Equation \ref{varianceell} and \ref{Rconvergence1}, we have
	\begin{eqnarray*}
		&&\|\ell^{TP}_{n_0}(\hat{\boldsymbol{\theta}}_{n_0s}^{TP},\mathbb{X}_s;\mathbb{X}_{V\backslash\{s\}})-\ell_{n_0}^{P}(\hat{\boldsymbol{\theta}}^P_{n_0s},\mathbb{X}_s;\mathbb{X}_{V\backslash\{s\}})\|_2\\
		&&~\ge \|\ell^{P}_{n_0}(\hat{\boldsymbol{\theta}}_{n_0s}^{TP},\mathbb{X}_s;\mathbb{X}_{V\backslash\{s\}})-\ell_{n_0}^{P}(\hat{\boldsymbol{\theta}}^P_{n_0s},\mathbb{X}_s;\mathbb{X}_{V\backslash\{s\}})\|_2\\
		&&\quad -\|\ell^{TP}_{n_0}(\hat{\boldsymbol{\theta}}_{n_0s}^{TP},\mathbb{X}_s;\mathbb{X}_{V\backslash\{s\}})-\ell^{P}_{n_0}(\hat{\boldsymbol{\theta}}_{n_0s}^{TP},\mathbb{X}_s;\mathbb{X}_{V\backslash\{s\}})\|_2\\
		&&~> \frac{1}{8}\lambda_{min}\epsilon^2_0-\frac{1}{16}\lambda_{min}\epsilon^2_0\\
		&&~= \frac{1}{16}\lambda_{min}\epsilon^2_0\\
		&&~>\|\ell^{TP}_{n_0}(\hat{\boldsymbol{\theta}}_{n_0s}^{TP},\mathbb{X}_s;\mathbb{X}_{V\backslash\{s\}})-\ell_{n_0}^{P}(\hat{\boldsymbol{\theta}}_{n_0s}^{TP},\mathbb{X}_s;\mathbb{X}_{V\backslash\{s\}})\|_2.
	\end{eqnarray*}
	It is easy to see that $\ell^{TP}_{n_0}({\boldsymbol{\theta}}_{s},\mathbb{X}_s;\mathbb{X}_{V\backslash\{s\}})>\ell^{P}_{n_0}({\boldsymbol{\theta}}_{s},\mathbb{X}_s;\mathbb{X}_{V\backslash\{s\}})$ for all ${\boldsymbol{\theta}}_{s}\in\boldsymbol{\Theta}$. Therefore, 
	$$\ell_{n_0}^{P}(\hat{\boldsymbol{\theta}}_{n_0s}^{TP},\mathbb{X}_s;\mathbb{X}_{V\backslash\{s\}})>\ell_{n_0}^{P}(\hat{\boldsymbol{\theta}}^P_{n_0s},\mathbb{X}_s;\mathbb{X}_{V\backslash\{s\}}),$$ contradict to $\hat{\boldsymbol{\theta}}^P_{n_0s}$ is the maximum likelihood estimate of $\ell_{n_0}^{P}({\boldsymbol{\theta}}_{s},\mathbb{X}_s;\mathbb{X}_{V\backslash\{s\}})$.
	
\end{proof}}

	{
		\section{A bound on the second and third derivative of the log normalizing term $D(.)$}\label{suppB}
		Here, we derive  bounds $\kappa_1$ and $ \kappa_2$ for the second and third derivative of the log normalizing term $D(.)$, that is, $D''(vhX_{it}+\langle\boldsymbol{\theta}_s,\bold{X}_{V\backslash \{s\}}^{(i)}\rangle),$ and $ D'''(vhX_{it}+\langle\boldsymbol{\theta}_s,\bold{X}_{V\backslash \{s\}}^{(i)}\rangle)$. For the sake of simplicity, we write 
		$$D(x)= \log\big(\sum_{k=0}^R\exp\big\{kx-\log k!\big\}\big),$$
		which we consider on a compact set $U\subset\mathbb{R}$. The first and second derivative of $D(.)$ is
		\begin{eqnarray*}
			D'(x)&=&\dfrac{\sum_{k=0}^R\exp\big\{kx-\log k!\big\}k}{\sum_{k=0}^R\exp\big\{kx-\log k!\big\}}\\
			D''(x)&=&\dfrac{\sum_{k=0}^R\exp\big\{kx-\log k!\big\}k^2\sum_{k=0}^R\exp\big\{kx-\log k!\big\}-\big(\sum_{k=0}^R\exp\big\{kx-\log k!\big\}k\big)^2}{\big(\sum_{k=0}^R\exp\big\{kx-\log k!\big\}\big)^2}\\
			&=&\dfrac{\sum_{k,h=0}^R\exp\big\{kx-\log k!\big\}\exp\big\{hx-\log h!\big\}(k^2-kh)}{\sum_{k,h=0}^R\exp\big\{kx-\log k!\big\}\exp\big\{hx-\log h!\big\}}.
		\end{eqnarray*}
		Hence, 
		\begin{eqnarray*}
			|D''(x)|&=&\left|\dfrac{\sum_{k,h=0}^R\exp\big\{kx-\log k!\big\}\exp\big\{hx-\log h!\big\}(k^2-kh)}{\sum_{k,h=0}^R\exp\big\{kx-\log k!\big\}\exp\big\{hx-\log h!\big\}}\right|\\
			&\le&\dfrac{\sum_{k,h=0}^R\exp\big\{kx-\log k!\big\}\exp\big\{hx-\log h!\big\}|k^2-kh|}{\sum_{k,h=0}^R\exp\big\{kx-\log k!\big\}\exp\big\{hx-\log h!\big\}}\\
			&\le&2R^2\dfrac{\sum_{k,h=0}^R\exp\big\{kx-\log k!\big\}\exp\big\{hx-\log h!\big\}}{\sum_{k,h=0}^R\exp\big\{kx-\log k!\big\}\exp\big\{hx-\log h!\big\}}\\
			&=&2R^2.
		\end{eqnarray*}
		Therefore, $\kappa_1\le O_p(R^2)$. Similarly, 
		\begin{eqnarray*}
			D'''(x)&=&\dfrac{N(x)}{\big(\sum_{k=0}^R\exp\big\{kx-\log k!\big\}\big)^4}
		\end{eqnarray*}
		where
		\begin{eqnarray*}
			N(x)&=&\bigg(\sum_{k=0}^R\exp\big\{kx-\log k!\big\}k^3\sum_{k=0}^R\exp\big\{kx-\log k!\big\}+\sum_{k=0}^R\exp\big\{kx-\log k!\big\}k^2\\
			&&~\sum_{k=0}^R\exp\big\{kx-\log k!\big\}k
			-2\sum_{k=0}^R\exp\big\{kx-\log k!\big\}k\sum_{k=0}^R\exp\big\{kx-\log k!\big\}k^2\bigg)\\
			&&~\big(\sum_{k=0}^R\exp\big\{kx-\log k!\big\}\big)^2
			-\bigg(\sum_{k=0}^R\exp\big\{kx-\log k!\big\}k^2\sum_{k=0}^R\exp\big\{kx-\log k!\big\}\\
			&&-\big(\sum_{k=0}^R\exp\big\{kx-\log k!\big\}k\big)^2\bigg)
			2\sum_{k=0}^R\exp\big\{kx-\log k!\big\}\sum_{k=0}^R\exp\big\{kx-\log k!\big\}k\\
			&=&\sum_{k,h,r,t=0}^R\exp\big\{kx-\log k!\big\}\exp\big\{hx-\log h!\big\}\exp\big\{rx-\log r!\big\}\exp\big\{tx-\log t!\big\}\\
			&&~\big(k^3-kh^2-2k^2t+2kht\big).
		\end{eqnarray*}
		Hence,
		\begin{eqnarray*}
			|D'''(x)|&\le&6R^3\dfrac{\sum_{k,h,r,t=0}^R\exp\big\{kx-\log k!\big\}\exp\big\{hx-\log h!\big\}\exp\big\{rx-\log r!\big\}\exp\big\{tx-\log t!\big\}}{\sum_{k,h,r,t=0}^R\exp\big\{kx-\log k!\big\}\exp\big\{hx-\log h!\big\}\exp\big\{rx-\log r!\big\}\exp\big\{tx-\log t!\big\}}\\
			&=&6R^3.
		\end{eqnarray*}
		Therefore, $\kappa_2\le O_p(R^3)$.
	}

	\section{The StARS algorithm}\label{suppC}
	The StARS algorithm  introduced in \cite{liu2010stability}, aims to seek the value of $\lambda$ leading to the most stable set of edges.  More precisely, it considers a range  $\Lambda=\{\lambda_1,\ldots,\lambda_k\}$ of values for $\lambda$, and fixes a number $n_B$, $1<n_B<n$ of observations in one sample. Then, $B$ samples of size $n_B$, $S_1,\ldots,S_B$, are generated from $\bold{x}_1,\ldots,\bold{x}_n$. For each $\lambda\in \Lambda$,  the graph is estimated by solving a lasso problem. Let $A_\lambda^{n_p}(S_1),\ldots,A_\lambda^{n_p}(S_B)$ be estimated adjacency matrices of the graph in the subsamples. The stability of one edge can be estimated by
	$$\epsilon_{s,t}^{n_B}(\lambda)=2\psi_{s,t}^{n_B}(\lambda)\big(1-\psi_{s,t}^{n_B}(\lambda)\big),$$
	where $\psi_{s,t}^{n_B}(\lambda)=\frac{1}{B}\sum_{i=1}^{B}A_\lambda^{n_B}(S_i)_{st}$ is the estimated probability of one edge between nodes $s$ and $t$. The optimal value $\lambda_{opt}$ is defined as the largest value that maximizes the total stability
	$$\bar{D}_{n_B}(\lambda)=sup_{0\le \rho\le \lambda}\sum_{s<t}\epsilon_{s,t}^{n_B}(\sigma)/\binom{p}{2},$$
	smaller than an upper bound $\beta$,  
	$\lambda_{opt}=\text{sup}\{\lambda: \bar{D}_B(\lambda)\le \beta\}.$
	
	\section{Simulation study results}\label{suppD}

	{ {\bf About the choice of the truncated Poisson distribution.} Table \ref{table10-TPandP} reports  TP, FP, FN, PPV, and Se for PC-LPGM obtained by simulating 500 datasets of size $n=1000$ from unrestricted Poisson conditional models. Data were generated as in Section~\ref{empirical} of the main paper, at both high ($\lambda_{noise}=0.5$) and low ($\lambda_{noise}=5$) SNR level.   Results refer to random graphs of $p = 10$ variables with varying probability of edge inclusion $\pi$. Here,  PC-LPGM is run with the proper test statistic, i.e., $Z_{st|\bold{K}}^P$ and with the misspecified one, i.e., $Z_{st|\bold{K}}^{TP}$. When $Z_{st|\bold{K}}^{TP}$ is used, the truncation point $R$ is fixed to be equal to the largest observation. 
	}
	
		\begin{table}[htbp!]
			\centering
			{\scriptsize
				\caption{\label{table10-TPandP}\scriptsize{ Monte Carlo  means of TP, FP, FN, PPV, and Se obtained by simulating 500 samples of size $n=1000$ from random graphs on $p = 10$ variables with Poisson node conditional distribution and level of noise $\lambda_{noise} = 0.5,\,5$. The probability  of edge inclusion $\pi$ runs from 0.1 to 0.4.}}
				\begin{tabular}{c|r|rrrrr|rrrrr}
					\hline
					&&\multicolumn{5}{c|}{$Z_{st|\bold{K}}^{TP}$}&\multicolumn{5}{c}{$Z_{st|\bold{K}}^{P}$}\\
					$\lambda_{noise}$&$\pi$& TP & FP & FN & PPV & Se  & TP & FP & FN & PPV & Se  \\ 
					\hline
					&0.1 & 8.390& 0.240& 0.610& 0.975 &0.932& 8.380 &0.233 &0.620& 0.975 &0.931  \\ 
					0.5&0.2 &9.970& 0.287& 0.030 &0.975& 0.997& 9.970& 0.287 &0.030 &0.975& 0.997  \\ 
					&0.3 & 12.837 & 0.153 & 0.163&  0.989&  0.987& 12.837 & 0.153 & 0.163 & 0.989&  0.987 \\ 
					&0.4 & 17.387&  0.117  &4.613 & 0.994&  0.790 &17.460 & 0.113 & 4.540 & 0.994 & 0.794  \\ 
					&&&&&&&&&&&\\
					&0.1 & 6.747 &0.367& 2.253& 0.955 &0.750 &6.747& 0.370& 2.253& 0.955& 0.750  \\ 
					5&0.2 & 8.470& 0.293 &1.530 &0.970 &0.847 &8.450 &0.293 &1.550 &0.970& 0.845  \\ 
					&0.3 & 11.283 & 0.080 & 1.717 & 0.993 & 0.868 & 11.757 & 0.097 & 1.243 & 0.992 & 0.904 \\ 
					&0.4 & 13.143 & 0.025 & 8.857 & 0.998 & 0.597 & 16.029 & 0.029 & 5.971 & 0.998 & 0.729 \\ 
					\hline
				\end{tabular}
			}
		\end{table}

\noindent	
	{ {\bf Unrestricted Poisson conditional models.}}  Table \ref{table1-chap1} to Table \ref{table4-chap1} report  TP, FP, FN, PPV and Se for each of  methods considered in Section \ref{empirical} of the main paper. Two different graph dimensions,  $p=10, 100$, and three graph structures (see Figure \ref{graphtypes} and Figure \ref{graphtypes100} of the main paper) are considered at one low ($\lambda_{noise}=5$) and one high ($\lambda_{noise}=0.5$) SNR levels. 
	
		\begin{landscape}
	\begin{table}[ht]
\scriptsize
\centering
\caption{\label{table10-poisnew}{ Monte Carlo  means of TP, PPV and Se obtained by simulating 500 samples from graphs in Figure \ref{graphtypes} with $p = 10$ variables with Poisson node conditional distribution with mean $\lambda= 1$ and levels of noise $ \lambda_{noise} = 0.5,5$. }}
\vspace{0.2cm}

\begin{tabular}{r|r|r|rrr|rrr|HHHrrr|HHHrrrHHH}
  \hline
&&&\multicolumn{3}{c|} {\bf PC-LPGM}&\multicolumn{3}{c|} {\bf LPGM}&\multicolumn{3}{H} {\bf PDN}&\multicolumn{3}{c|}{\bf VSL }&\multicolumn{3}{H}{\bf GLASSO }&\multicolumn{3}{c}{\bf NPN-Copula}&\multicolumn{3}{H}{\bf NPN-skeptic}\\
 $\lambda_{noise}$&type&$n$&TP& PPV & Se &TP & PPV & Se & TP & PPV & Se & TP & PPV & Se & TP & PPV & Se & TP&PPV&Se \\ 
    \hline
    &&50&2.440& 0.932 &0.271& 1.949& 0.962 &0.217&&& &2.437 &0.865& 0.271& 2.420 &0.862 &0.269 &2.577 &0.874 &0.286 &2.473  & NaN &0.275\\
    &&100&5.000 &0.956 &0.556& 2.354 &0.988& 0.262&&& &2.927& 0.976 &0.325 &2.927 &0.977 &0.325 &2.880& 0.983 &0.320 &2.840 &  NaN& 0.316\\
    0.5	&scalefree&200& 7.953& 0.975 &0.884 &4.493 &  0.986& 0.499 &5.872&    0.972 &  0.652 & 4.625 &  0.996 &  0.514 & 4.502 &  0.997 &  0.500 & 5.073 &  0.996 &  0.564 & 5.030 &  0.994 &  0.559 \\
	&&1000 & 9.000& 0.982 &1.000& 7.873 &0.890 &0.875 &5.780 &  1.000 &  0.642 & 4.954 & 1.000 &  0.550& 4.889 &  1.000 &  0.543 & 5.377 &  1.000 &  0.597 & 5.232 &  1.000 &  0.581  \\
	&&2000& 9.000 &0.982 &1.000& 8.417& 0.759 &0.935 &5.658 &  1.000 &  0.629& 5.566 &  1.000 &  0.618& 5.573 &  1.000 & 0.619 & 6.055 &  1.000 &  0.673 & 5.945 &  1.000  & 0.661 \\
&&&&&&&&&&&&&&&&&&&&&&&\\
&&50&2.103 &0.902& 0.263& 1.933& 0.926 &0.242&&&&2.237 &0.809& 0.280 &2.227 &0.807 &0.278 &2.287 &0.839& 0.286 &2.177 &  NaN& 0.272\\
&&100&4.237 &0.947& 0.530 &2.828 &0.957 &0.353 &&&&2.713 &0.950& 0.339 &2.683 &0.947 &0.335 &2.817 &0.960 &0.352 &2.793  & NaN &0.349\\
&Hub&200&7.160 &0.971 &0.895 &4.497 &  0.914& 0.562&6.680 &   0.926 &  0.835 & 4.316 & 0.995 & 0.540& 4.212 &  0.995& 0.527 & 4.636 & 0.996 &  0.580 &  4.506 &  0.995 &  0.563  \\
&&1000& 8.000 &0.976 &1.000 &7.863& 0.717& 0.983& 7.128 &  1.000 &  0.891& 5.908 &  1.000 &  0.739 & 5.842 &  1.000 & 0.730 & 6.000 & 1.000 &  0.750 &  5.818 & 1.000  &  0.727  \\
&&2000& 8.000& 0.980& 1.000 &7.990 &0.729& 0.999 &7.216 &  1.000 &  0.902 & 7.110 &  1.000 &  0.889 & 7.068 &  1.000 &  0.884 & 7.006 & 1.000 & 0.876  &  6.794  &   1.000 &0.849 \\
&&&&&&&&&&&&&&&&&&&&&&&\\
&&50&1.824 &0.875 &0.203 &1.805 &0.945& 0.201 &&&&2.343 &0.808 &0.260& 2.347& 0.808 &0.261 &2.310& 0.822 &0.257 &2.270 &  NaN &0.252\\
&&100&3.826 &0.939 &0.425 &2.483 &0.974& 0.276 &&&&2.680 &0.952 &0.298 &2.673 &0.952 &0.297 &2.863 &0.958 &0.318& 2.720  & NaN &0.302\\
&Random&200	& 6.930 &0.976& 0.770 &3.920&   0.981 &0.436& 4.800 &  0.675 & 0.600 & 3.510 &  0.993 &  0.439 & 3.464 &  0.995 & 0.433 & 3.934 &  0.995 &  0.492 &  3.826 &    0.995  &  0.478 \\
&&1000& 8.963& 0.980& 0.996 &7.320& 0.875& 0.813 & 5.066 &  0.703 &  0.633 & 3.190 &  1.000 & 0.399 & 3.110 &  1.000& 0.389 & 3.434 &  1.000 &  0.429 & 3.358 &    1.000 &  0.420\\
&&2000& 9.000& 0.981 &1.000& 8.320 &0.767& 0.924& 5.068 & 0.713 &  0.634 & 2.952 & 1.000 &  0.369 & 2.828 &  1.000 & 0.353 & 3.356 &  1.000 &  0.420 & 3.384 & 1.000 &  0.423 \\
\hline
&&&&&&&&&&&&&&&&&&&&&&&\\
&&50&0.711 &0.530 &0.079 &0.723 &0.615 &0.080&&&&1.053 &0.418 &0.117 &1.060 &0.423 &0.118 &1.177 &0.404 &0.131 &1.173  & NaN &0.130\\
&&100&1.110 &0.733 &0.123 &1.057 &0.849 &0.117 &&&&1.467& 0.594& 0.163 &1.463 &0.597& 0.163 &1.527 &0.608& 0.170 &1.500 &  NaN &0.167\\
5&scalefree&200&  1.875& 0.848& 0.208 &1.315 &0.979 &0.146& 3.824 &  0.486 &  0.425 & 1.934 &  0.797 & 0.215 & 1.914 & 0.796 & 0.213 & 2.012 &  0.840 &  0.224 & 1.832 &   0.821&  0.204 \\
&& 1000& 8.000& 0.964& 0.889 &4.609 &0.996 &0.512 & 6.148 & 0.948 &  0.683 & 3.212 &    0.999 &  0.357  & 3.194 &  0.997 &  0.355  & 3.302 &  0.999 &  0.367 & 3.058 &    0.999 &  0.340  \\
&&2000 & 8.977& 0.974 &0.997& 8.133 &1.000& 0.904& 6.258 &  0.997 &  0.695 & 4.238 &  1.000 &  0.471 & 4.222 & 1.000 &  0.469  & 4.408 & 1.000 &  0.490 & 4.198 &  1.000 &  0.466\\
				&&&&&&&&&&&&&&&&&&&&&&&\\
				&&50&0.733 &0.555 &0.092 &0.816 &0.717 &0.102 &&&&1.090 &0.368 &0.136 &1.093& 0.369 &0.137& 1.013 &0.357 &0.127& 0.943  & NaN& 0.118\\
				&&100&1.074& 0.710 &0.134 &1.157 &0.828 &0.145 &&&&1.303 &0.533 &0.163 &1.307& 0.537 &0.163 &1.343 &0.554 &0.168 &1.200 &  NaN &0.150\\
				&Hub&200 & 1.770 &0.860 &0.221& 1.258 &0.979 &0.157& 3.366 & 0.416 & 0.421 & 1.784 &  0.744 & 0.223 & 1.766 &  0.744 &  0.221 & 1.880 &  0.765 & 0.235 &  1.694 & 0.738  & 0.212 \\
&&1000& 7.143 &0.959 &0.893 &5.863 &  0.996& 0.733 &  6.594 &  0.897 &  0.824 & 3.152 & 1.000 & 0.394 & 3.142 & 1.000 &  0.393 & 3.168 &  1.000 &  0.396 &  2.990 & 0.998 &  0.374\\
&&2000 & 7.980& 0.970& 0.998 &7.030 &0.990& 0.879 & 7.158 &  0.994 &  0.895 & 3.900 &  1.000& 0.488 & 3.874 &  1.000 &  0.484 & 4.026 &  1.000 & 0.503 & 3.730 & 1.000  & 0.466  \\
		&&&&&&&&&&&&&&&&&&&&&&&\\		
		&&50&0.633 &0.517& 0.070& 0.842 &0.682 &0.094 &&&&1.043 &0.410 &0.116 &1.047 &0.412& 0.116 &1.073 &0.392 &0.119 &0.977 &  NaN& 0.109\\
		&&100&1.137 &0.708 &0.126& 1.177 &0.849& 0.131 &&&&1.547& 0.606 &0.172 &1.553 &0.608& 0.173 &1.553 &0.597& 0.173 &1.393 &  NaN& 0.155	\\
				&Random&200 & 1.776& 0.840 &0.197 &1.291& 0.970 &0.143 & 3.204 & 0.402 & 0.400 & 1.800 &  0.757 & 0.225 & 1.805 &  0.758 &  0.226 & 1.980 &  0.801 & 0.248 & 1.795  & 0.752 & 0.224 \\
&&1000& 7.517& 0.969& 0.835 &5.123 &  0.996 &0.569 & 4.816 & 0.653 & 0.602& 3.042 &  0.997 & 0.380 & 3.018  & 0.997 &  0.377 & 3.164 &  0.998 &  0.396 &  2.972 & 0.998 &  0.372  \\
&&2000 & 8.903 &0.970 &0.989 &7.890& 0.998 &0.877& 5.044 &  0.685 &  0.630 & 3.665  &  1.000& 0.458 & 3.640 &  1.000 & 0.455 & 3.785 & 1.000 &  0.473 & 3.610 & 1.000 &  0.451 \\
	\hline
    \end{tabular}
    \end{table}

    	\begin{table}[ht]
\scriptsize
\centering
\caption{\label{table100-poisnew}{ Monte Carlo  means of TP, PPV and Se obtained by simulating 500 samples from graphs in Figure \ref{graphtypes} with $p = 100$ variables with Poisson node conditional distribution with mean $\lambda= 1$ and levels of noise $ \lambda_{noise} = 0.5,5$. }}
\vspace{0.2cm}

\begin{tabular}{r|r|r|rrr|rrr|HHHrrr|HHHrrrHHH}
  \hline
&&&\multicolumn{3}{c|} {\bf PC-LPGM}&\multicolumn{3}{c|} {\bf LPGM}&\multicolumn{3}{H} {\bf PDN}&\multicolumn{3}{c|}{\bf VSL }&\multicolumn{3}{H}{\bf GLASSO }&\multicolumn{3}{c}{\bf NPN-Copula}&\multicolumn{3}{H}{\bf NPN-skeptic}\\
 $\lambda_{noise}$&type&$n$&TP& PPV & Se &TP & PPV & Se & TP & PPV & Se & TP & PPV & Se & TP & PPV & Se & TP&PPV&Se \\ 
    \hline 
    &&100& 18.880&  0.918 & 0.191 & 7.920  &0.831 & 0.080 &&&& 9.580 & 0.938 & 0.097 & 9.580&
  0.938  &0.097 &10.600  &0.952  &0.107 &10.640 & 0.935  &0.107\\
    0.5&Scalefree&200 & 54.236 & 0.958 & 0.548& 40.446 & 0.863 & 0.409  &53.080 & 0.673 &   0.536 & 63.915 &  0.760&   0.646 & 62.755 &   0.754 &  0.634 & 65.647 & 0.797 & 0.663 & 64.343 (& 0.759 & 0.650  \\
				
				&&1000 & 98.082&  0.977&  0.991& 88.196&  0.882&  0.891&65.357 &   0.999 &   0.660 & 94.438 &  0.999 & 0.954 & 93.830 &  0.998 & 0.948 &94.571& 1.000 &  0.955 & 94.277 & 0.999 & 0.952 \\
				
				&&2000& 98.994 & 0.977 & 1.000& 89.862 & 0.981 & 0.908  & 64.370 &   1.000 &   0.650& 96.821 &  1.000 & 0.978 & 96.518 & 1.000&  0.975 & 97.375 & 1.000 &  0.984 & 97.214 & 1.000 &  0.982\\
				
&&&&&&&&&&&&&&&&&&&&&&&\\			
&&100& 1.918 & 0.483 & 0.020&20.020 & 0.562 & 0.211 &&&& 3.850 & 0.306 & 0.041 & 3.880&
 0.307 & 0.041 & 4.200 & 0.337 & 0.044 & 3.860  &  NaN & 0.041	\\	
				&Hub&200& 7.340 & 0.729  &0.077 & 46.835 & 0.627 & 0.493 & 19.340 &   0.186&   0.204& 16.643 &   0.427 &   0.175 & 15.991 &   0.434 &  0.168 & 18.491 & 0.451 &  0.195 & 17.473 & 0.406 & 0.184  \\
				
				&&1000&81.360  &0.952 & 0.856& 83.350 & 0.803  &0.877 & 78.487 &    0.800 &   0.826& 29.651 & 0.998 &   0.312 & 29.341 &  0.998 &  0.309 & 37.746 & 0.999 &  0.397 & 35.476 & 0.998 & 0.373  \\

				&&2000& 94.788 & 0.959 & 0.998& 94.975 & 0.548 & 1.000  & 93.073 &    0.988 &   0.980 & 69.263 &   1.000 &   0.729 & 68.647 &   1.000 &  0.723 & 77.833 &  1.000 & 0.819 & 74.987 & 1.000 & 0.789  \\
&&&&&&&&&&&&&&&&&&&&&&&\\	
&&100&17.150 & 0.901&  0.157& 13.130  &0.815 & 0.120  &&&&9.890 & 0.918 & 0.091 & 9.880
& 0.917 & 0.091 &10.030 & 0.941 & 0.092 & 9.980 & 0.930 & 0.092		\\		
				&Random&200& 52.640 & 0.957 & 0.483 &48.970 & 0.774 & 0.449 & 52.007 &   0.619 &   0.477& 67.032 &  0.735 &  0.615 & 64.736 & 0.742 & 0.594 & 70.520 &  0.769 &  0.647 & 68.956 & 0.722 & 0.633 \\
				
				&&1000& 107.237  & 0.983  & 0.984  &96.360  & 0.788  & 0.884 & 63.020 &  0.870&   0.578& 102.676 &  0.999 & 0.942 & 101.904 & 0.999 & 0.935 & 104.820 & 0.999 &  0.962 & 104.392 &  0.998 &  0.958  \\

				&&2000&  107.237  & 0.983  & 0.984 & 96.360  & 0.788  & 0.884& 62.850 &   0.872 &   0.577 & 106.836 &  1.000&  0.980& 106.884 & 1.000 &  0.981 & 107.376 & 1.000 &  0.985 & 107.124 &  1.000 & 0.983 \\
\hline
&&&&&&&&&&&&&&&&&&&&&&&\\	
&&100& 1.063 &0.323 &0.011 & 2.860& 0.299 &0.029 &&&&2.630& 0.198 &0.027 &2.630 &0.198 &0.027&
2.730& 0.214& 0.028& 2.450  & NaN& 0.025\\
				5&Scalefree&200& 3.688& 0.560 &0.037 &1.106 &0.821 &0.011& 13.457 &   0.125 &   0.136 & 9.316 &   0.332 &  0.094 & 9.052 &   0.336&  0.091 & 10.012 & 0.359 &  0.101 & 9.868 &  0.320 & 0.100\\
				
				&&1000& 60.578  &0.939 & 0.612& 52.072 & 0.973 & 0.526& 52.827 &  0.630 &   0.534& 14.844 &  0.998 &   0.150 & 14.936 & 0.998&  0.151 & 17.124 &  0.998 &  0.173 & 16.708 &  0.996 &  0.169  \\
				
				&&2000& 92.632 & 0.959 & 0.936 &55.340 & 0.998  &0.559  &  67.917 &  0.939 &   0.686 & 24.579 & 1.000 & 0.255 & 25.733 & 1.000 & 0.264 & 33.672 &  1.000 & 0.335 & 32.267 & 1.000 & 0.321 \\
&&&&&&&&&&&&&&&&&&&&&&&\\					
				&&100&0.326& 0.139& 0.003& 3.210 &0.363& 0.034 &&&&1.080& 0.066 &0.011 &1.080 &0.066& 0.011& 1.130 &0.073& 0.012 &1.080 &  NaN& 0.011\\
				&Hub&200& 0.941&  0.246  &0.010  &15.092  &0.451 & 0.159 &6.630 &    0.060 &   0.070& 3.392 & 0.143 &  0.036 & 3.304 &  0.145 &   0.035 & 3.392 & 0.150 &  0.036 & 3.108 &  0.134 & 0.033 \\
				
				&&1000& 16.580 & 0.807 & 0.175 &39.208 & 0.876 & 0.413 & 23.427 & 0.217 &   0.247 & 7.424 &  0.884 &  0.078 & 7.364 &  0.883 &  0.078 & 8.440 &  0.895 &  0.089 & 8.208 & 0.860 & 0.086  \\

				&&2000& 49.210 & 0.920 & 0.518 &67.042 & 0.759 & 0.706 & 49.100 &   0.472 &   0.517 & 8.983 &  0.996& 0.095 & 8.924 &  0.996 & 0.094 & 9.797& 0.998 &  0.103 & 9.305& 0.995 &  0.098  \\
	&&&&&&&&&&&&&&&&&&&&&&&\\		
	&&100&1.065 &0.278 &0.010 &3.940 &0.309 &0.036 &&&&2.840 &0.208 &0.026& 2.860 &0.208 &0.026& 3.310 &0.195 &0.030 &3.030 &0.184& 0.028	\\	
				&Random&200&  3.739 &0.564& 0.034 &1.290 &0.775 &0.012& 13.573 & 0.126 &   0.125 & 10.548 &  0.353 & 0.097 & 10.160 &  0.358 & 0.093 & 11.064 &  0.382& 0.102 &10.648 & 0.341 &0.098 \\
				
				&&1000&64.270 & 0.941 & 0.590 &61.990&  0.961 & 0.569 & 53.207 &   0.616 &   0.488 &14.741 & 0.999 &  0.135& 14.741  & 0.999 & 0.135 & 16.333 & 0.999& 0.150 &15.178 & 0.998 &  0.139 \\
&&2000&101.457 &  0.962 &  0.931 & 67.477 &  1.000  & 0.619 & 65.093 &   0.841&   0.597 & 26.038 & 1.000 &  0.239 & 26.327 &  1.000 & 0.242 & 30.340 &  1.000 &  0.278 & 28.474 &  1.000 & 0.261 \\
    	\hline
    \end{tabular}
    \end{table}
\end{landscape}
	\setlength\belowcaptionskip{-3ex}
	\begin{center}
		\begin{scriptsize}
			
			\begin{longtable}{l| l | l r r r r r r}
				\caption{\label{table1-chap1} Simulation results from 500 replicates of the undirected graphs shown in Figure \ref{graphtypes} of the main paper for $p= 10$ variables with Poisson node conditional distribution and level of noise $\lambda_{noise} = 0.5$. Monte Carlo means (standard deviations) are shown for TP, FP, FN, PPV and Se.}  \\
				
				\toprule
				Graph&$n$	& Algorithm & TP & FP & FN & PPV & Se  \\
				\midrule
				\endfirsthead
				\multicolumn{8}{c}%
				{{\bfseries \tablename\ \thetable{} -- continued from previous page}} \\
				\toprule
				Graph&$n$	& Algorithm & TP & FP & FN & PPV & Se \\
				\midrule	
				
				\endhead
				&200 	&PC-LPGM& 6.838 (1.152)& 0.048 (0.230)& 2.163 (1.152)& 0.994 (0.208)& 0.760 (0.169) \\
				&	&LPGM & 4.732 (1.407)& 0.384 (0.644)& 4.268 (1.407)& 0.941 (0.097)& 0.526 (0.156)\\
				&	&PDN &5.872 (0.741)&  0.182 (0.430)& 3.128 (0.741)&  0.972 (0.065)&  0.652 (0.082)\\
				& 	& VSL& 4.625 (2.056)& 0.034 (0.181)& 4.375 (2.056)&  0.996 (0.021)&  0.514 (0.228)\\
				&	& GLASSO& 4.502 (1.961)& 0.023 (0.151)& 4.498 (1.961)&  0.997 (0.018)&  0.500 (0.218)\\
				&	&NPN-Copula& 5.073 (2.169)& 0.034 (0.191)& 3.927 (2.169)&  0.996 (0.023)&  0.564 (0.241)\\
				&	&NPN-Skeptic& 5.030 (2.177)& 0.039 (0.230)& 3.970 (2.177)&  0.994 (0.023)&  0.559 (0.242)\\
				&	& & & & & & \\
				
				&1000 	&PC-LPGM& 9.000 (0.000)& 0.071 (0.258)& 0.000 (0.000)&  0.993 (0.026)& 1.000 (0.000) \\
				&	&LPGM & 5.780 (1.253)& 0.692 (2.730)& 3.220 (1.253)&  0.964 (0.135)&  0.642 (0.139)\\
				&	&PDN &5.780 (0.661)& 0.000 (0.000)& 3.220 (0.661)&  1.000 (0.000)&  0.642 (0.073)\\
				Scale-free	&	& VSL& 4.954 (2.246)& 0.000 (0.000)& 4.046 (2.246)& 1.000 (0.000)&  0.550 (0.250)\\
				&	& GLASSO& 4.889 (2.234)& 0.000 (0.000)& 4.111 (2.234)&  1.000 (0.000)&  0.543 (0.248)\\
				&	&NPN-Copula& 5.377 (2.451)& 0.000 (0.000)& 3.623 (2.451)&  1.000 (0.000)&  0.597 (0.272)\\
				&	&NPN-Skeptic& 5.232 (2.609)& 0.000 (0.000)& 3.768 (2.069)&  1.000 (0.000)&  0.581 (0.290)\\
				&	& & & & & & \\
				
				&2000	&PC-LPGMC& 9.000 (0.000)& 0.071 (0.278)& 0.000 (0.000)&  0.993 (0.027)& 1.000 (0.000) \\
				&	&LPGM & 7.660 (1.611)& 5.180 (4.482)& 1.340 (1.611)&  0.703 (0.238)&  0.851 (0.179)	\\
				&	&PDN &5.658 (0.581)& 0.000 (0.000)& 3.342 (0.581)&  1.000 (0.000)&  0.629 (0.065)\\
				&	& VSL& 5.566 (2.381)& 0.000 (0.000)& 3.434 (2.381)&  1.000 (0.000)&  0.618 (0.265)\\
				&	& GLASSO& 5.573 (2.381)& 0.000 (0.000)& 3.427 (2.381)&  1.000 (0.000)& 0.619 (0.265)\\
				&	&NPN-Copula& 6.055 (2.509)& 0.000 (0.000)& 2.945 (2.509)&  1.000 (0.000)&  0.673 (0.279)\\
				&	&NPN-Skeptic& 5.945 (2.710)& 0.000 (0.000)& 3.055 (2.710)&  1.000 (0.000) & 0.661 (0.301)\\
				&	& & & & & & \\
				\hline		
				&	& & & & & & \\

				&200	&PC-LPGM& 6.618 (1.042)& 0.104 (0.132)& 1.382 (1.042)&  0.986 (0.042)& 0.827 (0.130) \\
				&	&LPGM &  3.072 (1.124)& 0.136 (0.505)& 4.928 (1.124)& 0.975 (0.077)& 0.384 (0.144)\\
				&	&PDN &6.680 (0.700)& 0.560 (0.769)& 1.320 (0.700)&   0.926 (0.099)&  0.835 (0.088)\\
				&	& VSL& 4.316 (1.933)& 0.030 (0.171)& 3.684 (1.933)& 0.995 (0.033)& 0.540 (0.242) \\
				&	& GLASSO& 4.212 (1.903)& 0.028 (0.177)& 3.788 (1.903)&  0.995 (0.033)& 0.527 (0.238)\\
				&	&NPN-Copula& 4.636 (1.936)& 0.024 (0.166)& 3.364 (1.936)& 0.996 (0.026)&  0.580 (0.242)\\
				&	&NPN-Skeptic&  4.506 (2.009)& 0.032 (0.187)& 3.494 (2.009)&  0.995 (0.028)&  0.563 (0.251)\\
				&	& & & & & & \\
				
				&1000	&PC-LPGM& 8.000 (0.000)& 0.122 (0.345)& 0.000 (0.000)& 0.987 (0.038)&  1.000 (0.000) \\
				&	&LPGM & 4.392 (2.669)& 1.452 (2.201)& 3.608 (2.669)& 0.885 (0.169)& 0.549 (0.334)\\
				&	&PDN & 7.128 (0.395)& 0.000 (0.000)& 0.872 (0.395)&  1.000 (0.000)&  0.891 (0.049)\\
				Hub		&	& VSL& 5.908 (1.920)& 0.000 (0.000)& 2.092 (1.920)&  1.000 (0.000)&  0.739 (0.240)\\
				&	& GLASSO& 5.842 (1.907)& 0.000 (0.000)& 2.158 (1.907)&  1.000 (0.000)& 0.730 (0.238) \\
				&	&NPN-Copula& 6.000 (2.094)& 0.000 (0.000)& 2.000 (2.094)& 1.000 (0.000)&  0.750 (0.262)\\
				&	&NPN-Skeptic&  5.818 (2.337)& 0.000 (0.000)& 2.182 (2.337)& 1.000 (0.000) &  0.727 (0.292)\\
				&	& & & & & & \\

				&2000	&PC-LPGM& 8.000 (0.000)& 0.132 (0.373)& 0.000 (0.000)&  0.986 (0.040)& 1.000 (0.000)\\
				&	&LPGM & 6.252 (2.688)& 2.480 (1.904)& 1.748 (2.688)& 0.790 (0.151)&  0.782 (0.336)\\
				&	&PDN &7.216 (0.488)& 0.000 (0.000)& 0.784 (0.488)&  1.000 (0.000)&  0.902 (0.061)\\
				&	& VSL& 7.110 (1.680)& 0.000 (0.000)& 0.890 (1.680)&  1.000 (0.000)&  0.889 (0.210)\\
				&	& GLASSO& 7.068 (1.681)& 0.000 (0.000)& 0.932 (1.681)&  1.000 (0.000)&  0.884 (0.210)\\
				&	&NPN-Copula& 7.006 (2.030)& 0.000 (0.000)& 0.994 (2.030)& 1.000 (0.000)& 0.876  (0.254)\\
				&	&NPN-Skeptic&  6.794 (2.272)& 0.000 (0.000)& 1.206 (2.272) &   1.000 (0.000) &  0.849 (0.284)\\
				
				&	& & & & & & \\
				\hline		
				&	& & & & & & \\
				
				&200	&PC-LPGM& 5.492 (1.581)& 0.052 (0.231)& 2.508 (1.581)& 0.991 (0.039) &  0.687 (0.198) \\
				&	&LPGM & 3.500 (1.120)& 0.244 (0.531)& 4.500 (1.120)&  0.950 (0.107)& 0.438 (0.140)\\
				&	&PDN & 4.800 (0.752)& 2.362 (0.817)& 3.200 (0.752)&  0.675 (0.085)& 0.600 (0.094)\\
				&	& VSL& 3.510 (1.655)& 0.034 (0.202)& 4.490 (1.655)&  0.993 (0.040)&  0.439 (0.207)\\
				&	& GLASSO& 3.464 (1.601)& 0.026 (0.171)& 4.536 (1.601)&  0.995 (0.036)& 0.433 (0.200)\\
				&	&NPN-Copula& 3.934 (1.823)& 0.028 (0.165)& 4.066 (1.823)&  0.995 (0.030)&  0.492 (0.228)\\
				&	&NPN-Skeptic&  3.826 (1.859)& 0.030 (0.182)& 4.174 (1.859)&    0.995 (0.031) &  0.478 (0.232)\\
				&	& & & & & & \\
				
				&1000	&PC-LPGM& 8.000  (0.000)& 0.078 (0.283)& 0.000 (0.000) & 0.991 (0.031)& 1.000 (0.000)\\
				&	&LPGM & 5.748 (1.989)& 3.584 (3.752)& 2.252 (1.989)&  0.758 (0.244)&  0.718 (0.249)\\
				&	&PDN & 5.066 (0.753)& 2.164 (0.634)& 2.934 (0.753)&  0.703 (0.068)&  0.633 (0.094)\\
				Random		&	& VSL& 3.190 (1.963)& 0.000 (0.000)& 4.810 (1.963)&  1.000 (0.000)& 0.399 (0.245)\\
				&	& GLASSO& 3.110 (1.897)& 0.000 (0.000)& 4.890 (1.897)&  1.000 (0.000)& 0.389 (0.237)\\
				&	&NPN-Copula& 3.434 (2.257)& 0.000 (0.000)& 4.566 (2.257)&  1.000 (0.000)&  0.429 (0.282)\\
				&	&NPN-Skeptic& 3.358 (2.351)& 0.000 (0.000)& 4.642 (2.351)&    1.000 (0.000)&  0.420 (0.294)\\
				&	& & & & & & \\

				&2000	&PC-LPGM& 8.000 (0.000)& 0.048 (0.214)& 0.000 (0.000)& 0.995 (0.024)&  1.000 (0.000)\\
				&	&LPGM & 7.484 (1.073)& 6.256 (2.369)& 0.516 (1.073)& 0.576 (0.140)&  0.936 (0.134)\\
				&	&PDN & 5.068 (0.730)& 2.082 (0.716)& 2.932 (0.730)& 0.713 (0.080)&  0.634 (0.091)\\
				&	& VSL& 2.952 (2.011)& 0.000 (0.000)& 5.048 (2.011)& 1.000 (0.000)&  0.369 (0.251) \\
				&	& GLASSO& 2.828 (1.886)& 0.000 (0.000)& 5.172 (1.886)&  1.000 (0.000)& 0.353 (0.236)\\
				&	&NPN-Copula& 3.356 (2,261)& 0.000 (0.000)& 4.644 (2.261)&  1.000 (0.000)&  0.420 (0.283)\\
				&	&NPN-Skeptic& 3.384 (2.321)& 0.000 (0.000)& 4.616 (2.321)& 1.000 (0.000)&  0.423 (0.290)\\
				
				\bottomrule
			\end{longtable}
		\end{scriptsize}
		
	\end{center}
	
	\setlength\belowcaptionskip{-3ex}
	\begin{center}
		\begin{scriptsize}
			
			\begin{longtable}{l| l | l r r r r r r}
				\caption{Simulation results from 500 replicates of the undirected graphs shown in Figure \ref{graphtypes} of the main paper for $p= 10$ variables with Poisson node conditional distribution and level of noise $\lambda_{noise} = 5$. Monte Carlo means (standard deviations) are shown for TP, FP, FN, PPV and Se. }  \\
				\toprule
				Graph&$n$	& Algorithm & TP & FP & FN & PPV & Se  \\
				\midrule
				\endfirsthead
				\multicolumn{8}{c}%
				{{\bfseries \tablename\ \thetable{} -- continued from previous page}} \\
				\toprule
				Graph&$n$	& Algorithm & TP & FP & FN & PPV & Se  \\
				\midrule	
				
				\endhead
				&200&PC-LPGM& 2.136 (1.617)& 0.744 (0.927)& 6.864 (1.617)& 0.756 (0.267) & 0.237 (0.180) \\
				&	&LPGM & 1.628 (1.249)& 1.920 (1.885)& 7.372 (1.249)&  0.524 (0.336)&  0.181 (0.139)\\
				&	&PDN & 3.824 (1.221)& 4.200 (1.655)& 5.176 (1.221)&  0.486 (0.164)&  0.425 (0.136)\\
				&	& VSL& 1.934 (1.142)& 0.658 (0.927)& 7.066 (1.142)&  0.797 (0.277)& 0.215 (0.127) \\
				&   & GLASSO& 1.914 (1.119)& 0.660 (0.937)& 7.086 (1.119)& 0.796 (0.278)& 0.213 (0.124)\\
				&   &NPN-Copula& 2.012 (1.214)& 0.550 (0.924)& 6.988 (1.214)&  0.840 (0.260)&  0.224 (0.135)\\
				&   &NPN-Skeptic& 1.832 (1.302)& 0.568 (0.927)& 7.168  (1.302)&   0.821 (0.237)&  0.204 (0.145)\\
				&	& & & & & & \\
				
				& 1000&PC-LPGM& 8.590 (0.764)& 1.060 (0.926)& 0.410 (0.764)& 0.898 (0.084)&  0.954 (0.085) \\
				&	&LPGM & 4.352 (1.818)& 2.020 (1.699)& 4.648 (1.818)&  0.719 (0.198)& 0.484 (0.202)\\
				&	&PDN & 6.148 (0.865)& 0.366 (0.604)& 2.852 (0.865)& 0.948 (0.082)&  0.683 (0.096)\\
				Scale-free&  & VSL& 3.212 (1.742)& 0.008 (0.089)& 5.788 (1.742)&    0.999 (0.015)&  0.357  (0.194)\\
				&   & GLASSO& 3.194 (1.734)& 0.008 (0.089)& 5.806 (1.734)&  0.997 (0.015)&  0.355  (0.193)\\
				&   &NPN-Copula& 3.302 (1.722)& 0.004 (0.063)& 5.698 (1.722)&  0.999 (0.017)&  0.367 (0.191)\\
				&   &NPN-Skeptic& 3.058 (1.867)& 0.004 (0.063)& 5.942 (1.867)&    0.999 (0.017)&  0.340 (0.207)\\
				&	& & & & & & \\
				
				&2000 &PC-LPGM& 8.996 (0.063)& 1.118 (1.017)& 0.004 (0.063)&  0.898 (0.085)& 1.000 (0.007) \\
				&	&LPGM &   4.828 (1.812)& 2.320 (2.006)& 4.172 (1.812)& 0.720 (0.178)& 0.536 (0.201)\\
				&	&PDN & 6.258 (0.803)& 0.020 (0.140)& 2.742 (0.803)&  0.997 (0.020)&  0.695 (0.089)\\
				&	& VSL& 4.238 (1.984)& 0.000 (0.000)& 4.762 (1.984)&  1.000 (0.000)&  0.471 (0.220)\\
				&   & GLASSO& 4.222 (1.975)& 0.000 (0.000)& 4.778 (1.975)& 1.000 (0.000)&  0.469  (0.219) \\
				&   &NPN-Copula& 4.408 (1.931)& 0.000 (0.000)& 4.592 (1.931)& 1.000 (0.000)&  0.490 (0.215)\\
				&   &NPN-Skeptic& 4.198 (2.102)& 0.000 (0.000)& 4.802 (2.102)&  1.000 (0.000)&  0.466 (0.234)\\
				
				&	& & & & & & \\
				\hline
				&	& & & & & & \\
				
				&200 &PC-LPGM& 2.132 (1.535)& 0.650 (0.830)& 5.868 (1.535)& 0.768 (0.278) & 0.267 (0.192) \\
				&	&LPGM &  1.588 (1.363)& 2.188 (2.212)& 6.412 (1.363)& 0.224 (0.334)& 0.099 (0.170)\\
				&	&PDN & 3.366 (1.265)& 4.876 (1.726)& 4.634 (1.265)& 0.416 (0.164)& 0.421 (0.158)\\
				&	& VSL& 1.784 (1.002)& 0.896 (1.236)& 6.216 (1.002)&  0.744 (0.300)& 0.223 (0.125)\\
				&   & GLASSO& 1.766 (1.003)& 0.890 (1.225)& 6.234 (1.003)&  0.744 (0.301)&  0.221 (0.125)\\
				&   &NPN-Copula& 1.880 (1.073)& 0.806 (1.109)& 6.120 (1.073)&  0.765 (0.297)& 0.235 (0.134)\\
				&   &NPN-Skeptic&  1.694 (1.157)& 0.842 (1.176)& 6.306 (1.157)& 0.738 (0.294) & 0.212 (0.145)\\
				&	& & & & & & \\
				
				&1000&PC-LPGM& 7.608 (0.586)& 1.150 (0.985)& 0.392 (0.586)&  0.879 (0.095)&  0.951 (0.073) \\
				&	&LPGM & 4.268 (1.175)& 2.636 (1.733)& 3.732 (1.751)& 0.636 (0.188)& 0.534 (0.219)\\
				&	&PDN &  6.594 (0.864)& 0.782 (0.914)& 1.406 (0.864)&  0.897 (0.116)&  0.824 (0.108)\\
				Hub     &   & VSL& 3.152 (1.628)& 0.012 (0.109)& 4.848 (1.628)& 1.000 (0.019)& 0.394 (0.203)\\
				&   & GLASSO& 3.142 (1.620)& 0.012 (0.109)& 4.858 (1.620)& 1.000 (0.019)&  0.393 (0.202) \\
				&   &NPN-Copula& 3.168 (1.647)& 0.006 (0.077)& 4.832 (1.647)&  1.000 (0.016)&  0.396 (0.206)\\
				&   &NPN-Skeptic&  2.990 (1.737)& 0.010 (0.100)& 5.010 (1.737)& 0.998 (0.021)&  0.374 (0.217)\\
				&	& & & & & & \\
				
				&2000 &PC-LPGM& 7.998 (0.045)& 1.160 (0.998)& 0.002 (0.045)&  0.883 (0.092)& 1.000 (0.006)\\
				&	&LPGM & 4.612 (2.231)& 2.708 (1.901)& 3.388 (2.231)& 0.632 (0.234)& 0.576 (0.279)\\
				&	&PDN & 7.158 (0.421)& 0.046 (0.210)& 0.842 (0.421)&  0.994 (0.027)&  0.895 (0.053)\\
				&	& VSL& 3.900 (1.823)& 0.000 (0.000)& 4.100 (1.823)&  1.000 (0.000)& 0.488 (0.228)\\
				&   & GLASSO& 3.874 (1.815)& 0.000 (0.000)& 4.126 (1.815)&  1.000 (0.000)&  0.484 (0.227)\\
				&   &NPN-Copula& 4.026 (1.881)& 0.000 (0.000)& 3.974 (1.881)&  1.000 (0.000)& 0.503 (0.235)\\
				&   &NPN-Skeptic& 3.730 (2.044)& 0.000 (0.000)& 4.270 (2.044)& 1.000 (0.000) & 0.466 (0.255)\\

				&	& & & & & & \\
				\hline
				&	& & & & & & \\
				
				&200 &PC-LPGM& 1.685 (1.437)& 0.740 (0.973)& 6.315 (1.437)& 0.716 (0.305)&  0.211 (0.180) \\
				&	&LPGM & 1.552 (1.189)& 2.264 (2.553)& 6.448 (1.189)& 0.513 (0.349)&  0.194 (0.149)\\
				&	&PDN & 3.204 (1.038)& 4.904 (1.507)& 4.796 (1.038)& 0.402 (0.137)& 0.400 (0.130)\\
				&	& VSL& 1.800 (1.103)& 0.850 (1.295)& 6.200 (1.103)&  0.757 (0.310)& 0.225 (0.138)\\
				&   & GLASSO& 1.805 (1.115)& 0.845 (1.300)& 6.195 (1.113)&  0.758 (0.309)&  0.226 (0.139) \\
				&   &NPN-Copula& 1.980 (1.194)& 0.735 (1.184)& 6.020 (1.194)&  0.801 (0.281)& 0.248 (0.149)\\
				&   &NPN-Skeptic& 1.795 (1.213)& 0.830 (1.265)& 6.205 (1.213) & 0.752 (0.291)& 0.224 (0.152)\\
				&	& & & & & & \\
				
				&1000&LRTPC& 7.470 (0.779)& 0.980 (1.044)& 0.530 (0.779)&  0.895 (0.101)&  0.934 (0.097)\\
				&	&LPGM &  3.724 (1.660)& 1.872  (1.850)& 4.276 (1.660)& 0.704 (0.250)& 0.466 (0.207)\\
				&	&PDN & 4.816 (0.709)& 2.600 (0.823)& 3.184 (0.709)& 0.653 (0.081)& 0.602 (0.089)\\
				Random  &   & VSL& 3.042 (1.588)& 0.016 (0.126)& 4.958 (1.588)&  0.997 (0.027)& 0.380 (0.198)\\
				&   & GLASSO& 3.018  (1.563)& 0.016 (0.126)& 4.982 (1.563)& 0.997 (0.027)&  0.377 (0.195)\\
				&   &NPN-Copula& 3.164 (1.588)& 0.008 (0.089)& 4.836 (1.588)&  0.998 (0.017)&  0.396 (0.199)\\
				&   &NPN-Skeptic&  2.972 (1.699)& 0.010 (0.100)& 5.028 (1.699)& 0.998 (0.022)&  0.372 (0.212)\\
				&	& & & & & & \\
				
				&2000 &LRTPC& 8.000 (0.000)& 0.848 (0.944)& 0.000 (0.000)& 0.914 (0.089)& 1.000 (0.000)\\
				&	&LPGM & 3.572 (1.533)& 1.424 (1.304)& 4.428 (1.533)& 0.758 (0.191)& 0.446 (0.192)\\
				&	&PDN & 5.044 (0.732)& 2.348 (0.645)& 2.956 (0.732)&  0.685 (0.065)&  0.630 (0.091)\\
				&	& VSL& 3.665 (1.803)& 0.000 (0.000)& 4.335 (1.803) &  1.000 (0.000)& 0.458 (0.225) \\
				&   & GLASSO& 3.640 (1.791)& 0.000 (0.000)& 4.360 (1.791)&  1.000 (0.000)& 0.455 (0.224)\\
				&   &NPN-Copula& 3.785 (1.823)& 0.000 (0.000)& 4.215 (1.823)& 1.000 (0.000)&  0.473 (0.228)\\
				&   &NPN-Skeptic& 3.610 (2.044)& 0.000 (0.000)& 4.390 (2.044)& 1.000 (0.000)&  0.451 (0.256)\\

				\bottomrule
			\end{longtable}
		\end{scriptsize}
		
	\end{center}
	
	
	\begin{center}
		\begin{scriptsize}
			
			\begin{longtable}{l| l | l r r r r r r }
				\caption{Simulation results from 500 replicates of the undirected graphs shown in Figure \ref{graphtypes100} of the main paper for $p= 100$ variables with Poisson node conditional distribution and level of noise $\lambda_{noise} = 0.5$. Monte Carlo means (standard deviations) are shown for TP, FP, FN, PPV and Se. }  \\
				\toprule
				Graph&$n$	& Algorithm &{TP} & {FP} & {FN} & {PPV} & {Se}    \\
				\midrule
				\endfirsthead
				\multicolumn{8}{c}%
				{{\bfseries \tablename\ \thetable{} -- continued from previous page}} \\
				\toprule
				Graph&$n$	& Algorithm &{TP} & {FP} & {FN} & {PPV} & {Se}  \\
				\midrule	
				
				\endhead
				&200 &PC-LPGM& 61.585 (4.316)&  8.490  (2.887) &37.415 (4.216)&  0.880 (0.038)&  0.622 (0.044)\\
				&	&LPGM & 5.564 (8.084)&  0.824 (5.594)& 93.436 (8.084)&  0.985 (0.067)&   0.056 (0.082)\\
				&	&PDN &53.080 (3.283)& 26.007 (4.942)& 45.920 (3.283)& 0.673 (0.052)&   0.536 (0.033)\\
				& 	& VSL& 63.915 (6.489)& 22.308 (13.433)& 35.085 (6.489)&  0.760 (0.095)&   0.646 (0.066)\\
				&	& GLASSO& 62.755 (6.306)& 22.642 (13.114)& 36.245 (6.306)&   0.754 (0.097)&  0.634 (0.064)\\
				&	&NPN-Copula& 65.647 (5.734)& 18.345 (11.701)& 33.352 (5.734)& 0.797 (0.088)& 0.663 (0.058)\\
				&	&NPN-Skeptic& 64.343 (6.316)& 22.918 (15.323)& 34.657 (6.316)& 0.759 (0.102)& 0.650 (0.064)\\
				&	& & & & & & \\
				
				&1000 	&PC-LPGM& 98.580 (0.610)&  9.589 (2.982)&0.420 (0.610)& 0.912 (0.025)&   0.996 (0.006)\\
				&	&LPGM & 51.520 (11.263)&  0.012 (0.109)& 47.480 (11.263)&  1.000 (0.002) &  0.520 (0.114)\\
				&	&PDN &65.357 (1.871)&  0.050 (0.218)& 33.643 (1.871)&   0.999 (0.003)&   0.660 (0.019)\\
				Scale-free	&	& VSL& 94.438 (2.316)&  0.089 (0.286)&  4.562 (2.316)&  0.999 (0.003)& 0.954 (0.023)\\
				&	& GLASSO& 93.830 (2.507)&  0.161 (0.393)&  5.170 (2.507)&  0.998 (0.004)& 0.948 (0.025)\\
				&	&NPN-Copula&94.571 (2.159)&  0.054 (0.226)&  4.429 (2.159)& 1.000 (0.002)&  0.955 (0.022)\\
				&	&NPN-Skeptic& 94.277 (2.089)&  0.134 (0.342)&  4.723 (2.089)& 0.999 (0.004)& 0.952 (0.021)\\
				&	& & & & & & \\
				
				&2000	&PC-LPGMC& 99.000 (0.000)&  9.759 (3.134)&  0.000 (0.000)& 0.911 (0.026)& 1.000 (0.000)\\
				&	&LPGM & 54.185 (2.379)&  0.010 (0.100)& 44.815 (2.379)&   1.000 (0.002)&   0.547 (0.024)\\
				&	&PDN & 64.370 (1.560)&  0.000 (0.000)& 34.630 (1.560)&   1.000 (0.000)&   0.650 (0.016)\\
				&	& VSL& 96.821 (1.422)&  0.000 (0.000)&  2.179 (1.422)&  1.000 (0.000)& 0.978 (0.014)\\
				&	& GLASSO& 96.518 (1.577)&  0.000 (0.000)&  2.482 (1.577)& 1.000 (0.000)&  0.975 (0.016)\\
				&	&NPN-Copula& 97.375 (1.402)&  0.000 (0.000)&  1.625 (1.402)& 1.000 (0.000)&  0.984 (0.014)\\
				&	&NPN-Skeptic& 97.214 (1.423)&  0.009 (0.000)&  1.786 (1.423)& 1.000 (0.000)&  0.982 (0.014)\\
				&	& & & & & & \\
				\hline		
				&	& & & & & & \\

				&200&PC-LPGM& 13.393 (2.484)& 14.518 (4.082)& 81.607 (2.484)& 0.486 (0.084)&   0.141 (0.026)\\
				&	&LPGM & 4.344 (4.368)&  5.840 (9.239)& 90.656 (4.368)&  0.426 (0.330)&   0.046 (0.046)\\
				&	&PDN & 19.340 (3.834)& 84.747 (5.935)& 75.660 (3.834)&   0.186 (0.038)&   0.204 (0.040)\\
				&	& VSL& 16.643 (6.546)& 26.982 (17.330)& 78.357 (6.546)&   0.427 (0.128)&   0.175 (0.069)\\
				&	& GLASSO& 15.991 (6.361)& 25.518 (16.665)& 79.009 (6.361)&   0.434 (0.135)&  0.168 (0.067)\\
				&	&NPN-Copula& 18.491 (6.864)& 26.625 (16.889)& 76.509 (6.864)& 0.451 (0.121)&  0.195 (0.072)\\
				&	&NPN-Skeptic& 17.473 (7.408)& 31.348 (22.170)& 77.527 (7.408)& 0.406 (0.123)& 0.184 (0.078)\\
				&	& & & & & & \\
				
				&1000&PC-LPGM& 84.794 (3.416)& 25.238 (5.079)& 10.206 (3.416)& 0.772 (0.036)&  0.893 (0.036)\\
				&	&LPGM & 4.555 (6.512)&  0.910 (1.349)& 90.445 (6.512)&  0.792 (0.324)& 0.048 (0.069)\\
				&	&PDN & 78.487 (3.585)& 19.650 (4.209)& 16.513 (3.585)&    0.800 (0.041)&   0.826 (0.038)\\
				Hub		&	& VSL& 29.651 (12.504)&  0.063 (0.303)& 65.349 (12.504)& 0.998 (0.010)&   0.312 (0.132)\\
				&	& GLASSO& 29.341 (12.233)&  0.056 (0.262)& 65.659 (12.233) &  0.998 (0.009)&  0.309 (0.129)\\
				&	&NPN-Copula& 37.746 (15.112)&  0.048 (0.248)& 57.254 (15.112)& 0.999 (0.004)&  0.397 (0.159)\\
				&	&NPN-Skeptic& 35.476 (16.277)&  0.119 (0.412)& 59.524 (16.277)& 0.998 (0.007)& 0.373 (0.171)\\
				&	& & & & & & \\

				&2000&PC-LPGM& 94.949 (0.221)& 26.942 (5.566)&  0.051 (0.221)& 0.781 (0.036)& 0.999 (0.002)\\
				&	&LPGM & 7.145 (9.369)&  0.625 (0.805)& 87.855 (9.369)&0.620  (0.478)&   0.075 (0.099)\\
				&	&PDN & 93.073 (1.205)& 1.113 (1.094)&  1.927 (1.205)&    0.988 (0.012)&   0.980 (0.013)\\
				&	& VSL& 69.263 (15.639)&  0.013 (0.113)& 25.737 (15.639)&   1.000 (0.001)&   0.729 (0.165)\\
				&	& GLASSO& 68.647 (14.931)&  0.013 (0.113)& 26.353 (14.931)&   1.000 (0.001)&  0.723 (0.157)\\
				&	&NPN-Copula& 77.833 (8.985)&  0.000 (0.000)& 17.167 (8.895)&  1.000 (0.000)& 0.819 (0.095)\\
				&	&NPN-Skeptic& 74.987 (9.809)&  0.013 (0.000)& 20.013 (8.985)& 1.000 (0.001)& 0.789 (0.103)\\
				
				&	& & & & & & \\
				\hline		
				&	& & & & & & \\
				
				&200&PC-LPGM& 62.432 (5.030)&   8.656 (2.998)&  46.568 (5.030)& 0.879 (0.039)&  0.573  (0.046)\\
				&	&LPGM & 8.190 (2.370)&  0.120 (0.326)& 100.810 (2.370) & 0.987 (0.036)&  0.075 (0.025)\\
				&	&PDN & 52.007 (3.302)& 32.167 (5.283)& 56.993 (3.302)&   0.619 (0.049)&   0.477 (0.030)\\
				&	& VSL& 67.032 (8.241)&  26.932 (15.060)&  41.968 (8.241)&  0.735 (0.100)&  0.615 (0.076)\\
				&	& GLASSO& 64.736 (8.543)&  25.440 (15.001)&  44.264 (8.543)& 0.742 (0.106)& 0.594 (0.078)\\
				&	&NPN-Copula& 70.520 (7.514)&  23.344 (13.387)& 38.480 (7.514)&  0.769 (0.091)&  0.647 (0.069)\\
				&	&NPN-Skeptic& 68.956 (8.123)&  29.956 (18.522)&  40.044 (8.123)& 0.722 (0.105)& 0.633 (0.075)\\
				&	& & & & & & \\
				
				&1000	&PC-LPGM& 105.748 (1.504)&   8.752 (2.939)& 3.252 (1.504)&  0.924 (0.024)& 0.970 (0.014)\\
				&	&LPGM & 43.800 (31.795)& 0.300 (0.593)& 65.200 (31.795)& 0.996 (0.009)&  0.402 (0.292)\\
				&	&PDN & 63.020 (2.491)& 9.470 (1.332)& 45.980 (2.491)&  0.870 (0.016)&   0.578 (0.023)\\
				Random		&	& VSL& 102.676 (3.506)&   0.136 (4.123)&   6.324 (3.506)&  0.999 (0.003)& 0.942 (0.032)\\
				&	& GLASSO& 101.904 (4.123)&   0.152 (0.142)&   7.096  (4.123)&  0.999 (0.004)& 0.935 (0.038)\\
				&	&NPN-Copula& 104.820 (2.159)&   0.104 (0.319)&  4.180 (2.159)& 0.999 (0.003)&  0.962 (0.020)\\
				&	&NPN-Skeptic& 104.392 (2.237)&  0.192 (0.424)& 4.608 (2.237)&  0.998 (0.004)&  0.958 (0.021)\\
				&	& & & & & & \\

				&2000	&PC-LPGM& 106.724 (1.212)& 8.664 (2.855)&  2.276 (1.212)& 0.925 (0.023)&  0.979 (0.011)\\
				&	&LPGM & 69.900 (7.493)&  0.280 (0.577)& 39.100 (7.493)& 0.996 (0.008)&  0.641 (0.069)\\
				&	&PDN & 62.850 (2.243)&  9.230 (1.439)& 46.150 (2.243)&   0.872 (0.018)&   0.577 (0.021)\\
				&	& VSL& 106.836 (1.365)&   0.000 (0.000)&   2.164 (1.365)&  1.000 (0.000)&  0.980 (0.013)\\
				&	& GLASSO& 106.884 (1.350)&  0.000 (0.000)&   2.116 (1.350)& 1.000 (0.000)&  0.981 (0.012)\\
				&	&NPN-Copula& 107.376 (1.253)&   0.000 (0.000)&  1.624 (1.253)& 1.000 (0.000)&  0.985 (0.011)\\
				&	&NPN-Skeptic& 107.124 (1.322)&  0.000 (0.000)&  1.876 (1.322)&  1.000 (0.000)& 0.983 (0.012)\\
				
				\bottomrule
			\end{longtable}
		\end{scriptsize}
		
	\end{center}

	\begin{center}
		\begin{scriptsize}
			
			\begin{longtable}{l| l | l r r r r r r}
				\caption{Simulation results from 500 replicates of the undirected graphs shown in Figure \ref{graphtypes100} of the main paper for $p= 100$ variables with Poisson node conditional distribution and level of noise $\lambda_{noise} = 5$. Monte Carlo means (standard deviations) are shown for TP, FP, FN, PPV and Se.}  \\
				\label{table4-chap1}\\
				\toprule
				Graph&$n$	& Algorithm & TP & FP & FN & PPV & Se  \\
				\midrule
				\endfirsthead
				\multicolumn{8}{c}%
				{{\bfseries \tablename\ \thetable{} -- continued from previous page}} \\
				\toprule
				Graph&$n$	& Algorithm & TP & FP & FN & PPV & Se  \\
				\midrule	
				\endhead
				&200 &PC-LPGM& 7.780 (2.843)& 14.470 (3.705)& 91.220 (2.843)& 0.348 (0.100)&   0.079 (0.029)\\
				&	&LPGM & 10.188 (4.126)& 65.352 (20.496)& 88.812 (4.126)&   0.152 (0.127)&  0.103 (0.042)\\
				&	&PDN & 13.457 (3.164)& 94.817 (6.073)& 85.543 (3.164)&   0.125 (0.030)&   0.136 (0.032)\\
				& 	& VSL& 9.316 (4.895)&  22.496 (16.852)&  89.684 (4.895)&   0.332 (0.119)&  0.094 (0.049)\\
				&	& GLASSO& 9.052 (4.775)&  21.372 (16.016)&  89.948 (4.775)&   0.336 (0.120)&  0.091 (0.048)\\
				&	&NPN-Copula& 10.012 (5.255)&  21.924 (16.439)&  88.988 (5.255)& 0.359 (0.135)&  0.101 (0.053)\\
				&	&NPN-Skeptic& 9.868 (5.979)&  27.424 (24.698)&  89.132 (5.979)&  0.320 (0.132)& 0.100 (0.060)\\
				&	& & & & & & \\
				
				&1000&PC-LPGM& 75.130 (4.420)& 24.805 (4.647)& 23.870 (4.420)&  0.753 (0.038)& 0.759 (0.045)\\
				&	&LPGM &1.480 (1.696)& 1.892 (3.146)& 97.520 (1.696)&  0.574 (0.412)&  0.015 (0.017)\\
				&	&PDN & 52.827 (3.386)& 31.153 (5.108)& 46.173 (3.386)&  0.630 (0.049)&   0.534 (0.034)\\
				Scale-free	&	& VSL& 14.844 (6.389)&   0.044 (0.224)&  84.156 (6.389)&  0.998 (0.013)&   0.150 (0.065)\\
				&	& GLASSO& 14.936 (6.455)&   0.044 (0.224)&  84.064 (6.455)& 0.998 (0.013)&  0.151 (0.065)\\
				&	&NPN-Copula& 17.124 (7.494)&   0.040 (0.196)& 81.876 (7.494)&  0.998 (0.009)&  0.173 (0.076)\\
				&	&NPN-Skeptic& 16.708 (8.088)&   0.116 (0.419)&  82.292 (8.088)&  0.996 (0.014)&  0.169 (0.082)\\
				&	& & & & & & \\
				
				&2000&PC-LPGMC& 96.400 (1.515)& 26.500 (5.147)&  2.600 (1.514)&  0.786 (0.033)& 0.974 (0.015)\\
				&	&LPGM &2.800 (2.138)& 1.004 (1.455)& 96.200 (2.138)&  0.785 (0.266)&  0.028 (0.022)\\
				&	&PDN &  67.917 (2.591)&  4.413 (2.140)& 31.083 (2.591)&  0.939 (0.029)&   0.686 (0.026)\\
				&	& VSL& 24.579 (11.580)&  0.000 (0.000)& 74.421 (11.580)& 1.000 (0.000)& 0.255 (0.117)\\
				&	& GLASSO& 25.733 (12.171)&  0.000 (0.000)& 73.267 (12.171)& 1.000 (0.000)& 0.264 (0.123) \\
				&	&NPN-Copula& 33.672 (14.879)&  0.000 (0.000)& 65.328 (14.879)&  1.000 (0.000)& 0.335 (0.150)\\
				&	&NPN-Skeptic& 32.267 (15.750)&  0.000 (0.000)& 66.733 (15.750)& 1.000 (0.000)& 0.321 (0.159)\\
				&	& & & & & & \\
				\hline		
				&	& & & & & & \\

				&200&PC-LPGM& 2.690 (1.705)& 13.600 (4.476)& 92.310 (1.705)&   0.166 (0.101)& 0.028 (0.018)\\
				&	&LPGM & 0.444 (1.175)& 34.632 (33.612)& 94.556 (1.175)&  0.046 (0.152) & 0.005 (0.012)\\
				&	&PDN &6.630 (2.373)& 103.063 (4.902)&  88.370 (2.373)&    0.060 (0.021)&   0.070 (0.025)\\
				&	& VSL& 3.392 (2.233)&  23.688 (15.017)&  91.608 (2.233)& 0.143 (0.097)&  0.036 (0.024)\\
				&	& GLASSO& 3.304 (2.139)&  22.964 (14.511)&  91.696  (2.139)&  0.145 (0.099)&   0.035 (0.023)\\
				&	&NPN-Copula& 3.392 (2.189)&  21.852 (13.797)&  91.608 (2.189)& 0.150 (0.097)&  0.036 (0.023)\\
				&	&NPN-Skeptic& 3.108 (2.297)&  23.476 (19.474)& 91.892 (2.297)&  0.134 (0.091)& 0.033 (0.024)\\
				&	& & & & & & \\
				
				&1000	&PC-LPGM& 29.525 (3.837)& 24.635 (5.206)& 65.475 (3.837)& 0.548 (0.029)& 0.311 (0.020)\\
				&	&LPGM & 0.892 (2.246)& 1.076 (2.639)& 94.108 (2.246)&  0.439 (0.389)&   0.009 (0.012)\\
				&	&PDN & 23.427 (3.516)& 84.433 (5.305)& 71.573 (3.316)& 0.217 (0.033)&   0.247 (0.037)\\
				Hub		&	& VSL& 7.424 (4.075)&   1.428 (2.091)&  87.576 (4.075)&  0.884 (0.137)&  0.078 (0.043)\\
				&	& GLASSO& 7.364 (4.053)&   1.424 (2.095)&  87.636 (4.053)&  0.883 (0.138)&  0.078  (0.043)\\
				&	&NPN-Copula& 8.440 (4.399)&   1.392 (2.018)&  86.560 (4.399)&  0.895 (0.126)&  0.089 (0.046)\\
				&	&NPN-Skeptic& 8.208 (4.629)&   1.804 (2.291)&  86.792 (4.629)& 0.860 (0.134)& 0.086 (0.049)\\
				&	& & & & & & \\

				&2000&PC-LPGM& 65.025 (4.253)& 29.855 (5.473)& 29.975 (4.253)& 0.687 (0.041)& 0.684 (0.045)\\
				&	&LPGM & 0.392 (0.796)&  1.712 (1.971)& 94.608 (0.796)&  0.187 (0.339)&  0.004 (0.008)\\
				&	&PDN & 49.100 (4.566)& 54.997 (5.883)& 45.900 (4.566)&   0.472 (0.047)&   0.517 (0.048)\\
				&	& VSL& 8.983 (6.782)&   0.068 (0.284)&  86.017 (6.782)&  0.996 (0.018)& 0.095 (0.071)\\
				&	& GLASSO& 8.924 (6.748)&   0.068 (0.284)&  86.076 (6.748)&  0.996 (0.018)& 0.094 (0.071)\\
				&	&NPN-Copula& 9.797 (7.547)&   0.042 (0.241)&  85.203 (7.547)& 0.998 (0.012)&  0.103 (0.079)\\
				&	&NPN-Skeptic& 9.305 (7.052)&   0.068 (0.284)&  85.695 (7.052)& 0.995 (0.020)&  0.098 (0.074)\\
				
				&	& & & & & & \\
				\hline		
				&	& & & & & & \\
				
				&200&PC-LPGM& 8.040 (2.884)& 14.805 (3.878)& 100.960 (2.884)& 0.350 (0.093)&  0.074 (0.026)\\
				&	&LPGM & 10.592 (4.318)& 69.316 (25.767)& 98.408 (4.318)&  0.175 (0.189)&  0.097 (0.040)\\
				&	&PDN & 13.573 (2.989)& 94.750 (5.987)& 95.427 (2.989)& 0.126 (0.029)&   0.125 (0.027)\\
				&	& VSL& 10.548 (5.213)&  22.848 (15.701)&  98.452 (5.213)&  0.353 (0.119)& 0.097 (0.048)\\
				&	& GLASSO& 10.160 (5.075)&  21.460 (14.871)&  98.840 (5.075)&  0.358 (0.119)& 0.093 (0.047)\\
				&	&NPN-Copula& 11.064 (5.327)&  22.136 (16.599)&  97.936 (5.327)&  0.382 (0.131)& 0.102 (0.049)\\
				&	&NPN-Skeptic&10.648 (6.242)&  26.632 (23.376)&  98.352 (6.642)& 0.341 (0.134)&0.098 (0.057)\\
				&	& & & & & & \\
				
				&1000&PC-LPGM& 81.055 (4.632)& 23.665 (4.941)& 27.945 (4.632)& 0.775 (0.038)& 0.744 (0.042)\\
				&	&LPGM &1.776 (2.675)&  3.196 (5.107)& 107.224 (2.675)& 0.397 (0.401) & 0.016 (0.025)\\
				&	&PDN & 53.207 (3.471)& 33.383 (10.084)& 55.793 (3.471)&   0.616 (0.046)&   0.488 (0.032)\\
				Random		&	& VSL&14.741 (6.294)&   0.022 (0.148)&  94.259 (6.294)& 0.999 (0.006)&  0.135 (0.058)\\
				&	& GLASSO& 14.741 (6.291)&   0.022 (0.148)&  94.259 (6.291) & 0.999 (0.006)& 0.135 (0.058)\\
				&	&NPN-Copula& 16.333 (7.249)&   0.022 (0.148)&  92.667 (7.249)& 0.999 (0.005)& 0.150 (0.067)\\
				&	&NPN-Skeptic&15.178 (7.307)&  0.044 (0.296)& 93.822 (7.307)& 0.998 (0.011)&  0.139 (0.067)\\
				&	& & & & & & \\

				&2000&PC-LPGM& 104.010 (1.992)&  24.370 (4.706)& 4.990 (1.992)& 0.811 (0.029)& 0.954 (0.018)\\
				&	&LPGM & 1.995 (1.800)&  1.260 (1.825)& 107.005 (1.880) &   0.671 (0.360)&  0.018 (0.017)\\
				&	&PDN & 65.093 (2.892)& 12.297 (1.837)& 43.907 (2.892)&   0.841 (0.021)&   0.597 (0.027)\\
				&	& VSL& 26.038 (12.457)&   0.000 (0.000)& 82.962 (12.457)& 1.000 (0.000)&  0.239 (0.114)\\
				&	& GLASSO& 26.327 (12.487)&  0.000 (0.000)& 82.673 (12.487)&  1.000 (0.000)& 0.242 (0.115)\\
				&	&NPN-Copula& 30.340 (14.496)&   0.000 (0.000)& 78.660 (14.496)&  1.000 (0.000)&  0.278 (0.133)\\
				&	&NPN-Skeptic& 28.474 (14.777)&  0.000 (0.000)&  80.526 (14.777)&  1.000 (0.000)& 0.261 (0.136)\\
				
				\bottomrule
			\end{longtable}
		\end{scriptsize}
		
	\end{center}

	\vskip 0.2in
	\bibliographystyle{natbib}
	\bibliography{NguyenChiogna_arxiv}
	
\end{document}